%% file: ms.tex
\documentclass[10pt]{extarticle}
\pdfoutput=1
\input{macros}

\usepackage{amsthm}
\usepackage{float}

\usepackage{hyperref}
\usepackage{textcomp}
\usepackage{amssymb, amsmath, mathtools}
\usepackage{graphicx}
\graphicspath{ {imgs/} }
\input{helpers/macros.tex}
\newcommand{\J}{\mathbb{J}}
\newcommand{\I}{\mathbb{I}}
\newcommand{\pathquery}{Q_{\textup{path}}}
\newcommand{\swingquery}{Q_{\textup{swing}}}
\newcommand{\seesawquery}{Q_{\textup{seesaw}}}
\newcommand{\headQ}{Q_{\textup{head}}}
\newcommand{\U}{\mathcal{U}}
\renewcommand{\S}{\mathcal{S}}
\renewcommand{\paragraph}[1]{\medskip \noindent {\bf #1}}
\makeatother
\newcommand{\eat}[1]{}
\newcommand{\bm}{\mathbf}

\begin{document}

\title{Aggregated Deletion Propagation for Counting Conjunctive Query Answers}

\author{
    Xiao Hu,  Shouzhuo Sun, Shweta Patwa, Debmalya Panigrahi, and Sudeepa Roy   \\
    Duke University, Durham, NC, USA\\
    \{xh102, ss1060, sjpatwa, debmalya, sudeea\}@cs.duke.edu \\
   }

\date{}

\maketitle

 \input{sections/abstract.tex}

%%% do not modify the following VLDB block %%
%%% VLDB block start %%%

\input{sections/introduction.tex}

\input{sections/related}
\input{sections/prelim.tex}

\input{sections/dichotomy.tex}

\input{sections/structure.tex}

\input{sections/approximation}
\input{sections/algorithms.tex}

\input{sections/experiment.tex}

\input{sections/conclusions.tex}

\clearpage
\bibliographystyle{abbrv}
\bibliography{grading}  

\clearpage
\input{sections/appendix}

\end{document}

%% file: macros.tex
\usepackage[left=1in,top=1in,right=1in,bottom = 1in, nohead]{geometry}

\renewenvironment{thebibliography}[1]{%
	\begin{oldthebibliography}{#1}%
		\small
		\setlength{\parskip}{0.1ex}%		% 1.8 is normal
		\setlength{\itemsep}{0.0ex}%
	}%
	{%
	\end{oldthebibliography}%
}
\makeatletter
\def\@cite#1#2{[\textbf{#1\if@tempswa , #2\fi}]}
\def\@biblabel#1{[\textbf{#1}]}
\makeatother
\usepackage{mdframed}
\usepackage{graphicx}   % need for figures
\usepackage{amsmath,amssymb}    % need for subequations
\usepackage{thm-restate}
\usepackage{mdwlist}
\usepackage{mathtools}
\usepackage{xspace}
\usepackage{verbatim}   % useful for program listings
\usepackage{xcolor}
\usepackage{makecell}
\usepackage[mathscr]{euscript}
\usepackage{tikz}
\usepackage{relsize}
\usepackage{booktabs}
\usepackage{tikz-qtree}
\usepackage{subcaption}
\usepackage{multirow}
\usetikzlibrary{arrows.meta}
\usetikzlibrary{fit,shapes,trees,shapes.geometric}
\usetikzlibrary{patterns,decorations.pathreplacing,calc}
\usetikzlibrary{matrix, positioning, arrows}
\usetikzlibrary{chains,shapes.multipart}
\usetikzlibrary{shapes,calc}
\usetikzlibrary{automata}
\usepackage[linesnumbered,algoruled, lined, noend]{algorithm2e}

\usepackage[normalem]{ulem}

\usepackage{amsthm}

\newcommand{\cut}[1]{}   % remove things inside the cut
\makeatletter
\newcommand{\mytag}[2]{%
	\text{#1}%
	\@bsphack
	\protected@write\@auxout{}%
	{\string\newlabel{#2}{{#1}{\thepage}}}%
	\@esphack
}
\makeatother

\newcommand{\ie}{{\em i.e.}\xspace}
\newcommand{\eg}{{\em e.g.}\xspace}

\usepackage{aliascnt} 

\newtheorem{definition}{Definition}

\newtheorem{proposition}{Proposition}
\newtheorem{corollary}{Corollary}
\newtheorem{conjecture}{Conjecture}
\newtheorem{lemma}{Lemma}
\newtheorem{example}{Example}
\newtheorem{remark}{Remark}

\newtheorem{theorem}{Theorem}
\newtheorem{observation}{Observation}

%% file: helpers/macros.tex
\newcommand{\allattr}{{\mathbb A}}

\newcommand{\attrset}{{\mathbb B}}

\newcommand{\allrel}{{\mathbb R}}

\newcommand{\attr}{{\texttt{attr}}}
\newcommand{\head}{{\texttt{head}}}
\newcommand{\rel}{{\texttt{rels}}}
\newcommand{\dom}{{\texttt{dom}}}
\newcommand{\ourprob}{{\texttt{ADP}}}
\newcommand{\svcb}{{\texttt{SVCB}}}
\newcommand{\clique}{{\texttt{CLIQUE}}}
\newcommand{\regular}{{\texttt{REGULAR-CLIQUE}}}
\newcommand{\kmc}{{\texttt{KMC}}}

\newcommand{\true}{{\texttt{true}}}
\newcommand{\false}{{\texttt{false}}}
\newcommand{\isptime}{{\textsc{IsPtime}}}

\newcommand{\optcost}{{\textsc{Opt}}}

\newcommand{\computeADP}{{\textsc{ComputeADP}}}
\newcommand{\decompose}{{\textsc{Decompose}}}

\newcommand{\greedy}{{\textsc{Greedy}}}
\newcommand{\drastic}{\textsc{Drastic}}
\newcommand{\bool}{{\textsc{Boolean}}}
\newcommand{\singleton}{{\textsc{Singleton}}}
\newcommand{\universe}{{\textsc{Universe}}}
\definecolor{darkgreen}{rgb}{0,0.5,0}
\definecolor{darkpurple}{rgb}{0.5,0,0.5}
\newcommand{\revm}[1]{{\color{black} #1}}
\newcommand{\reva}[1]{{\color{black}  #1}}
\newcommand{\revb}[1]{{\color{black}   #1}}
\newcommand{\revc}[1]{{\color{black}    #1}}

\definecolor{purple}{rgb}{0.65,0.05,0.3}
\definecolor{red}{rgb}{1,0,0}
\definecolor{RED}{rgb}{1,0,0}

\DeclareMathAlphabet{\mathbbold}{U}{bbold}{m}{n}

\cut{

\theoremstyle{definition}
\newtheorem{definition}{Definition}[section]

\theoremstyle{example}
\newtheorem{example}{Example}[section]

\newtheorem{theorem}{Theorem}[section]          	% Theorem environment.
\newaliascnt{lemma}{theorem}				% 1 alias counter
\newtheorem{lemma}[lemma]{Lemma}              	% Lemma environment.
\aliascntresetthe{lemma}  					% 3 set
\newaliascnt{conjecture}{theorem}			% 1 alias counter
\aliascntresetthe{conjecture}  				% 3 set
\newaliascnt{remark}{theorem}				% 1 alias counter
              
\aliascntresetthe{remark}  					% 3 set
\newaliascnt{corollary}{theorem}			% 1 alias counter
\aliascntresetthe{corollary}  				% 3 set
\newaliascnt{definition}{theorem}			% 1 alias counter
\newtheorem{definition}[definition]{Definition}    % Definition environment.
\aliascntresetthe{definition}  				% 3 set
\newaliascnt{proposition}{theorem}			% 1 alias counter
\aliascntresetthe{proposition}  				% 3 set
\newaliascnt{example}{theorem}			% 1 alias counter
\newtheorem{example}[example]{Example}  	% 2 environment.
\aliascntresetthe{example}  				% 3 set
\newaliascnt{observation}{theorem}			% 1 alias counter
\aliascntresetthe{observation}  				% 3 set
}

\renewcommand{\eg}{{e.g.}}
\renewcommand{\ie}{{i.e.}}

%% file: sections/abstract.tex
\begin{abstract}
We investigate the computational complexity of minimizing the source side-effect in order to remove a given number of tuples from the output of a conjunctive query. This is a variant of the well-studied {\em deletion propagation} problem, the difference being that we are interested in  %aggregate impact of removing input tuples on the output 
\revc{removing the smallest subset of input tuples to remove a given number of output tuples}
while deletion propagation focuses on removing a specific output tuple. We call this the {\em Aggregated Deletion Propagation} problem. We completely characterize the poly-time solvability of this problem for arbitrary conjunctive queries without self-joins. This includes a poly-time algorithm to decide solvability, as well as an exact structural characterization of NP-hard instances. We also provide a practical algorithm for this problem (a heuristic for NP-hard instances) and evaluate its experimental performance on real and synthetic datasets.
\end{abstract}

%% file: sections/introduction.tex
\section{Introduction}
\label{sec:introduction}
The problem of \emph{view update} (\eg, \cite{Bancilhon+1981, Dayal+1982}) -- how to change the input to achieve  desired changes to the query output or \emph{view} -- is a well-studied problem in the database literature. %This arises in contexts where 
View update problems enable users to tune the output in order to meet their prior expectation, satisfy external constraints, or examine and compare multiple options.
A particularly well-studied class of view update problems is what is known as \emph{deletion propagation} problems (see Buneman, Khanna, and  Tan \cite{Buneman+2002}; for follow up literature, see related work). In these problems, the goal is to remove a specific tuple from the output of a query by removing input tuples. In this paper, we study a natural variant of this problem where we seek to remove {\em at least a given number of output tuples} rather than any specific output tuple. We call this the \emph{Aggregated Deletion Propagation} problem.

%In this problem, we are given a database $D$, a monotone query $Q$, and a designated output tuple $t \in Q(D)$. The goal is to remove $t$ from $Q(D)$ by removing input tuples from $D$ subject to two alternative optimization criteria. In the \emph{source side-effect} version, the goal is to remove $t$ from $Q(D)$ by removing the smallest number of input tuples from $D$, whereas in the \emph{view side-effect} version, the goal is to remove $t$ such that the number of other output tuples deleted from $Q(D)$ is minimized. The intuition is that if the user considers tuple $t$ to be erroneous, then she would want to remove it from the output in a way that is minimal in terms of her intervention on the input, or in terms of the disruption caused in the output. \alert{Tighten the last few sentences - do we even need to define the view side-effect problem?}
\par
Formally, in the \emph{Aggregated Deletion Propagation} (\ourprob), we are given a query $Q$, a database $D$, %instead of removing a designated tuple $t$ from $Q(D)$, 
and a target integer $k$. The goal is to remove at least $k$ tuples from $Q(D)$ by removing the {\em minimum number of input tuples} from $D$ (this objective is called {\em source side-effect} in the literature). %(i.e., minimize the source side effect). While the standard deletion propagation problem intends to remove  a (potentially  erroneous) designated tuple $t$ from $Q(D)$ with minimum changes in the input data, \ourprob\ aims to help the users understand the dependence of the query outputs on the input at a coarser level with two-fold applications. 
Our main motivation for the \ourprob\ problem comes from two generic application settings.
First,  \ourprob\ can be used to obtain a desired change in the {\em output size} with minimum intervention on the input. As we will describe below, in many practical situations, the goal is to create a %sufficiently large %aggregate 
\revc{sufficiently large impact on the output by removing a given number of output tuples} rather than removing any specific tuple. Our problem applies to these situations. Second, \ourprob\ can be used to analyze the \emph{robustness} of the output with respect to possible disruptions in  the input. In other words, if there are inadvertent changes to the input that are not within our control, how badly can it effect the output of a query? We give examples of these two applications below.

%We illustrate these  two applications in another example below: 

\begin{example}
     Suppose a university wants to plan ahead in terms of managing waitlists for its classes. This can \revm{be} achieved via the following query:
     \begin{equation*}
     Q_{WL}(S, C) :- Major(S, M), Req(M, C), NoSeat(C)
     \end{equation*}
     The first query $Q_{WL}$ says that a student S is on the waitlist for a class C if the following happen: (1) S intends to major in M (we assume students can have multiple majors), (2) major M requires class C, and (3) there are no seats available in C.  The university may try to figure out the easiest alternative for reducing the size of the waitlist to some target, which amounts to reducing the size of the output of query $Q_{WL}$ by the same amount. The waitlist entries can be removed by steering students away from the major (or creating an entry condition), relaxing the requirements for the major, or by increasing the number of seats in the class; all of these options correspond to removing tuples from the input relations of $Q_{WL}$. 
\end{example}  

\begin{example}
     We consider the same context as in the previous example, but suppose the new task is to 
     estimate what classes can be reliably offered in a future semester. This can be
     done using the following query
     \begin{equation*}
     Q_{Possible}(C) :- Teaches(P, C), NotOnLeave(P).
     \end{equation*}
     This query lists the possible courses that can be offered in a semester. A course C can be offered if there is a professor P who is able to teach C and is not on leave. If all professors who are able to teach C go to leave (removal of entries from $NotOnLeave$) or do not want to teach C (removal of entries from $Teaches$), C cannot be offered.  While approving the leave requests and asking for teaching preferences, the university may want to study the robustness of $Q_{Possible}$ with respect to these changes: e.g., what is the minimum changes in the input that would lead to more than 10\% of the courses not being able to be offered in that semester. If this size is small, i.e., many courses are critically dependent on a few professors, the university would be able to decide whether all can be on leave or change teaching preferences appropriately. Alternatively, this information might also inform the decision to hire faculty in a particular area.
\end{example}

\begin{example}
We now turn to a third example from the area of robustness of networks. Consider a query
{
$$Q_{3-path}(A, B, C, D) :- R_1(A, B), R_2(B, C), R_3(C, D)$$
}
that stores all possible paths between two end vertices that go through two layers of intermediate vertices in a communication or transportation network. If it were possible to disrupt (say) 80\% of the paths by only removing (say) 1\% of links, then the network is clearly not robust. On the other hand, if this would require removing (say) 80\% of the links, that's a much more robust network. This is precisely the information the \ourprob\ can provide us on this query. Therefore, \ourprob\ can estimate the inherent robustness of a network to either malicious attacks or even just random failures.
\end{example}

%We are interested in studying this problem from the perspective of minimizing the source side-effect, i.e., we want to remove $k$ output tuples by removing the smallest number of input tuples from $D$.  
%%Consider an example where an airline has flights from a set of northern locations  to a set of central locations stored in a relation ${\tt R_{nc}(north, central)}$, and also from a set of central locations to a set of southern locations stored in ${\tt R_{cs}(central, south)}$. The conjunctive query $Q_{\rm all trips}(n, c, s):- R_{nc}(n, c), R_{cs}(c, s)$ shown in Datalog format finds all the north-central-south routes served by the airline. Now, consider a competitor that want to start new routes in a minimum number of these segments so as to affect at least $k$ of the routes being operated by the first airline. This can be exactly modeled by the $\ourprob$ problem. More generally, the $\ourprob$ problem can be used to analyze whether there is a small subset of input tuples with high impact on the output: if removal of a small number of input tuples causes a large change in the output, then there is significant dependence on this small subset which can be a potential source of vulnerability. 
\par

\noindent
\textbf{Our contributions.~} In this paper, we propose the \ourprob\ problem and study its complexity in depth for the class of {\em conjunctive queries without self-joins} (CQ). Here, the results can be an arbitrary projection of the \emph{natural join} of the relations appearing on the body of the query (as illustrated in $Q_{WL}$, $Q_{Possible}$, and $Q_{3-path}$ above).
\eat{\footnote{Self-join-free queries have been studied in most of the related papers on deletion propagation. The only work that we are aware of involves self-joins is~\cite{freire2019new}, which focuses on boolean CQ with binary relations, a very restricted sub-class of CQs. Moreover, their work only includes a dichotomy result if a relation is repeated at most 2 times. 
%However, if a relation is repeated more than 2 times, this problem is still open, which verifies our intuition that self-joins would make the deletion propagation very difficult. 
The complexity of \ourprob\ for CQ with self-joins is an interesting future work.}.
}
Our contributions can be summarized as follows:
%\alert{@XIAO: COULD YOU UPDATE THE FOUR ITEMS BELOW?} 
\begin{itemize}
	\item \textbf{Algorithmic Dichotomy:} We give an algorithm that only takes the query $Q$ as input, and decides in time that is polynomial in the size of the query, whether \ourprob\ can be efficiently solved (in polynomial time data complexity \cite{vardi1982complexity}) on $Q$ for all instances $D$ and all values of $k$. The algorithm uses a few simplification steps that preserve the complexity of the problem. At the end, the query is NP-hard if the simplification steps reduce it to a small number of `core' hard queries; otherwise, it is poly-time solvable.  \textbf{ (Section~\ref{SEC:DICHOTOMY})}
	%\medskip
	\item \textbf{Structural Dichotomy:} To complement our algorithmic characterization of the complexity of the \ourprob\ problem, we also provide a structural characterization of the complexity 
	%of CQs, on which \ourprob\ can be solved on $Q$ for all instances $D$ and all values of $k$ in polynomial time in terms of the data complexity.  
	by identifying three simple structures -- {\em triad-like}, {\em non-hierarchical head join}, and {\em strand} -- whose presence exactly captures all queries where  \ourprob\ in NP-hard. 
	%The presence of these three structures can be checked in poly-time in the size of the query.
	\textbf{(Section~\ref{sec:structure})}
	%\medskip
	\item \textbf{Approximation:} We study the approximation for the \ourprob\ problem when it is NP-hard. We show that greedy and prime-dual achieve approximation factors of $O(\log k)$ and $p$ respectively for full CQs, where $p$ is the number of relations in the input query. Meanwhile, we present some inapproximability result when projection exists, such that obtaining even sub-polynomial approximations for the \ourprob\ problem on general CQs is unlikely. 
    \textbf{(Section~\ref{SEC:APPROX})}
    \item  \textbf{Efficient unified algorithm:} We give a poly-time (in data complexity) algorithm for solving \ourprob\ for all CQs without self-joins. It returns the optimal solution for queries on which \ourprob\ is poly-time solvable, and provides a poly-time heuristic for queries on which \ourprob\ is NP-hard. We also extend the algorithm to support {\em selection} operations.
    %Then, for the remaining class of ``hard'' queries on which the \ourprob\ is NP-hard, we also provide heuristic algorithm to give an approximation solution. 
    \textbf{(Section~\ref{SEC:ALGORITHMS})}
    %\medskip
    \item  %\alert{update -- experiments only} 
\textbf{Experimental evaluations:} We provide experimental evaluation of our algorithms on synthetic and real datasets in terms of efficiency, quality, scalability, various classes of queries as well as data distribution. 
%The experimental results demonstrate its efficiency on different datasets.
\textbf{(Section~\ref{SEC:EXPERIMENT})}
\end{itemize}
%We discuss related work in Section~\ref{sec:related}, give the problem definition and review some concepts in Section~\ref{sec:prelim}, and conclude in Section~\ref{sec:conclusions}. 
%
%
%%Since the \ourprob\ problem is already NP-hard even for very simple queries like $Q_{\textup{cover}}, Q_{\textup{swing}}, Q_{\textup{seesaw}}$, we then study approximations to this problem (Section~\ref{sec:approx}). We give an approximation algorithm by a reduction to the \emph{partial set cover} problem. When $f$ is the maximum frequency of an element in the sets, Gandhi, Khuller, and Srinivasan \cite{GKS04} generalize the classic primal dual algorithm for the set cover problem to obtain an $f$-approximation for the partial set cover problem. Using this algorithm, we get a $p$-approximation for the $\ourprob$ problem, where $p$ is the number of relations in the schema.

%% file: sections/related.tex
\section{Related Work}\label{sec:related}
The classical view update problem, of which deletion propagation is an instantiation, has been studied extensively over the last four decades (\eg, \cite{Bancilhon+1981, Dayal+1982}). The deletion propagation problem has been popular more recently, starting with the seminal work by Buneman, Khanna, and Tan \cite{Buneman+2002}. They studied the complexity of both the \emph{source side-effect} (objective is to delete the \emph{minimum number of input tuples}) and the \emph{view side-effect} (objective is to delete the \emph{minimum number of other output tuples}) versions, \revc{in order to delete a particular output tuple}. For source side-effect and select-project-join-union (SPJU) operators, they showed that for PJ or JU queries, finding the optimal solution is NP-hard, while for others (\eg, SPU or SJ) it is poly-time solvable. This work was extended to multi-tuple deletion propagation by Cong, Fan, and Geerts \cite{Cong+2006}. They showed that for single tuple deletion propagation, a property called \emph{key preservation} makes the problem tractable for SPJ views; however, if multiple tuples are to be deleted, the problem becomes intractable for SJ, PJ, and SPJ views. Kimelfeld, Vondrak, and Williams~\cite{KimelfeldVW11, Kimelfeld12, KimelfeldVW13} extensively studied the complexity of deletion propagation for the view side-effect version and provided structural dichotomy and trichotomy (poly-time, APX-hard/constant approximation, and inapproximable) for single and multiple output tuple deletions. 
\par
Beyond the context of deletion propagation, several dichotomy results have been obtained for problems motivated by data management, \eg, in the context of probabilistic databases \cite{DalviS12},  responsibility \cite{MeliouGMS11}, or database repair \cite{LivshitsKR18}. Another problem related  to \ourprob\ is \emph{reverse data management} and \emph{how-to} queries \cite{MeliouGS11, MeliouS12}. Given some desired changes in the output (\eg, modifying aggregate values, creating or removing tuples), 
%can be specified by a Datalog-like language, and 
the goal is to obtain a feasible modification of the input that satisfies a given set of constraints and optimizes on some criteria. In this line of research, the focus has been on developing an end-to-end system using provenance and mixed integer programming, and not on the complexity of the problem. 
%Although \cite{MeliouS12} considered a much more general class of queries and update operations, their focus was to develop an end-to-end system using provenance and mixed integer programming, and not on the complexity of this problem. As discussed before, 
\ourprob\ is also related to explanations by intervention \cite{WM13, RoyS14, RoyOS15}, where the goal is to find a set of input tuples captured by a predicate whose deletion changes one or more aggregate answers to the maximum extent. \ourprob\ differs in that the aim is to make a desired change in the output by removing the minimum number of input tuples.

Finally, closely related to the \ourprob\ is the \emph{resilience} problem, originally studied by Freire et al.  for the class of CQs without self-joins and  functional dependencies \cite{FreireGIM15} (see also \cite{freire2019new} for an extension to a class of queries with self-joins). 
%\footnote{Self-join-free queries is a natural assumption for formal analysis of queries for deletion propagation and other contexts including our current work. The only work related to deletion propagation that we are aware of involves self-joins is~\cite{freire2019new}, which focuses on boolean CQ with binary relations, a very restricted sub-class of CQs. Moreover, their work only includes a dichotomy result if a relation is repeated at most 2 times. However, if a relation is repeated more than 2 times, this problem is still open, which verifies our intuition that self-joins would make the deletion propagation very difficult. The complexity of \ourprob\ for CQ with self-joins is an interesting future work.}
 The input to the resilience problem is a Boolean CQ and a database $D$ such that $Q(D)$ is true, and the goal is to remove a minimum set of tuples from $D$ to make $Q$ false on $D$. Observe that the resilience problem is identical to \ourprob\ with $k = |Q(D)|$. %\cite{freire2019new} 
 \revm{\cite{FreireGIM15}} gave a ``structural dichotomy'' characterizing whether a given query is poly-time solvable or NP-hard using a core hard structure called ``triad''. The generalization to arbitrary values of $k$ leads to interesting consequences, e.g., queries that are poly-time solvable for resilience become hard for \ourprob), whereas the presence of arbitrary projections in the output makes \ourprob\ even more NP-hard for \ourprob. Nevertheless, we use the characterization for resilience from \revm{\cite{FreireGIM15}} as a special case of our algorithmic and structural characterization for \ourprob\ and discuss the resilience problem further in subsequent sections.

%% file: sections/prelim.tex
\section{Preliminaries}
\label{sec:prelim}

In this section, we start with some basic definitions in relational databases. Then, we formally define the \ourprob\ problem and discuss some special cases that will motivate our general technique.
 
\subsection{Background} \label{sec:background}
We consider the standard setting of multi-relational data-bases and conjunctive queries. Let $\allrel$ be a database schema that contains $p$ tables $R_1, \cdots, R_p$.
Let $\allattr$ be the set of all attributes in the database $\allrel$. Each relation $R_i$ is defined on a subset of attributes $\attr(R_i) \subseteq \allattr$.  A relation $R_i$ is {\em vacuum} if $\attr(R_i) = \emptyset$, and {\em non-vacuum} otherwise. We use $A, B, C, A_1, A_2 , \cdots$ etc.
to denote the attributes in $\allattr$ and $a, b, c, \cdots$ etc. to denote their values. For each attribute $A \in \allattr$, %$\dom(A)$ denotes the domain of $A$ and 
$\rel(A)$ denotes the set of relations 
that $A$ appears, \ie, $\rel(A) = \{R_i: A \in \attr(R_i)\}$. 

\begin{figure}
	\centering
\begin{tabular}{|c|c|}
\multicolumn{2}{c}{$R_1$}\\
\hline
A & B \\\hline\hline
a1 & b1 \\
a2 & b2 \\
a3 & b3 \\\hline
\end{tabular}
\begin{tabular}{|c|c|}
\multicolumn{2}{c}{$R_2$}\\
\hline
B & C \\\hline\hline
b1 & c1 \\
b2 & c2 \\
b2 & c3 \\
b3 & c3 \\
\hline
\end{tabular}
\begin{tabular}{|c|c|}
\multicolumn{2}{c}{$R_3$}\\
\hline
C & E \\\hline\hline
c1 & e1 \\
c2 & e3 \\
c3 & e3 \\
\hline
\end{tabular}
\begin{tabular}{|c|c|c|c|}
\multicolumn{4}{c}{$Q_1(D)$}\\
\hline
A & B & C & E\\\hline\hline
a1 & b1 & c1  & e1\\
a2 & b2 & c2  & e3\\
a2 & b2 & c3  & e3\\
a3 & b3 & c3  & e3\\
\hline
\end{tabular}
\begin{tabular}{|c|c|}
	\multicolumn{2}{c}{$Q_2(D)$}\\
	\hline
	A  & E\\\hline\hline
	a1  & e1\\
	a2  & e3\\
	a3  & e3\\
	\hline
\end{tabular}
\caption{An example of database schema $\allrel = \{R_1, R_2, R_3\}$ with $\allattr$ $= \{A, B, C, E\}$, $\attr(R_1) =$ $\{A, B\}$, $\attr(R_2)= $ $\{B, C\}$, and $\attr(R_3) =$ $\{C, E\}$. An instance $D$ with $10$ tuples is also shown. The results for $Q_1(A, B, C, E):$ $-R_1(A, B)$, $ R_2(B, C)$, $R_3(C, E)$ and $Q_2(A,E):$ $-R_1(A, B)$, $R_2(B, C)$, $R_3(C, E)$ are $Q_1(D)$ and $Q_2(D)$.}
\label{fig:example_setup}
\end{figure}

\par
Given the database schema $\allrel$, let $D$ be a given instance of $\allrel$, and the corresponding instances of $R_1, \cdots, R_p$ be $R_1^{D}, \cdots$, $R_p^{D}$. Where $D$ is clear from the context, we will drop the superscript and use $R_1, \cdots, R_p$ for both the schema and instances. Any tuple $t \in R_i$ is defined on $\attr(R_i)$. For any attribute $A \in \attr(R_i)$, $ \pi_{A} t \in  \dom(A)$ denotes the value of attribute $A$ in tuple $t$.	
Similarly, for a set of attributes $\attrset \subseteq \attr(R_i)$, $\pi_\attrset t$ denotes the values of attributes in $\attrset$ for $t$ with an implicit ordering on the attributes. It should be noted that for a vacuum relation $R_i$, either $R_i = \{\emptyset\}$ or $R_i = \emptyset$ (respectively interpreted as ``true'' and ``false'').

We consider the class of \emph{conjunctive queries without self-joins},  formally defined as
\[Q(\bm{A}): -R_1(\mathbb{A}_1), R_2(\mathbb{A}_2), \cdots, R_p(\mathbb{A}_p)\]
where $\bm{A} \subseteq \mathbb{A}$ denotes the {\em output attributes} and $\mathbb{A} - \bm{A}$ the {\em non-output attributes} (\revc{also called the \emph{existential variables}}). Note that we do not have any projection in the body. Each $R_i$ in $Q$ is distinct, i.e., the CQ does not have a self-join. If $\bm{A} = \mathbb{A}$, such a CQ query is known as {\em full CQ} which represents the natural join among the given relations. If $\bm{A} = \emptyset$, such a CQ  is {\em boolean} which indicates whether the result of natural join among the given relations is empty or not; otherwise, it is {\em non-boolean}. 

\par
Extending the notation, we use $\rel(Q)$ to denote all the relations that appear in the body of $Q$, $\attr(Q)$ to denote all the attributes 
that appear in the body of $Q$, and $\head(Q) \subseteq \attr(Q)$ to denote all the attributes that appear in the head of $Q$ (so, $\head(Q) = \mathbb{A}$ in the previous paragraph). 
When a full CQ query $Q$ is evaluated on an instance $D$, if  $R_i = \emptyset$ for some vacuum relation $R_i \in \rel(Q)$, then $Q(D)$ is also empty; otherwise, the result $Q(D)$ is evaluated on non-vacuum relations. When a CQ query $Q$ is evaluated on an instance $D$, the result is exactly the projection of the full join result on attributes in $\head(Q)$ (after removing duplicates). We give an example in Figure~\ref{fig:example_setup}.

A classical representation of a CQ $Q$ is to model it as a hypergraph, where each attribute in $\attr(Q)$ is a vertex and each relation in $\rel(Q)$ is a hyperedge. In this work, we use a simpler representation for capturing the {\em connectivity} of queries and model it as a graph $G_Q$, where each relation is a vertex and there is an edge between $R_i, R_j \in \rel(Q)$ if $\attr(R_i) \cap \attr(R_j) \neq \emptyset$. This graph is denoted $G_Q$. A CQ $Q$ is {\em connected} if $G_Q$ is connected, and {\em disconnected} otherwise. An example is illustrated in Figure~\ref{fig:representation}.
\begin{figure}[t]
	\centering
	\includegraphics[scale = 0.8]{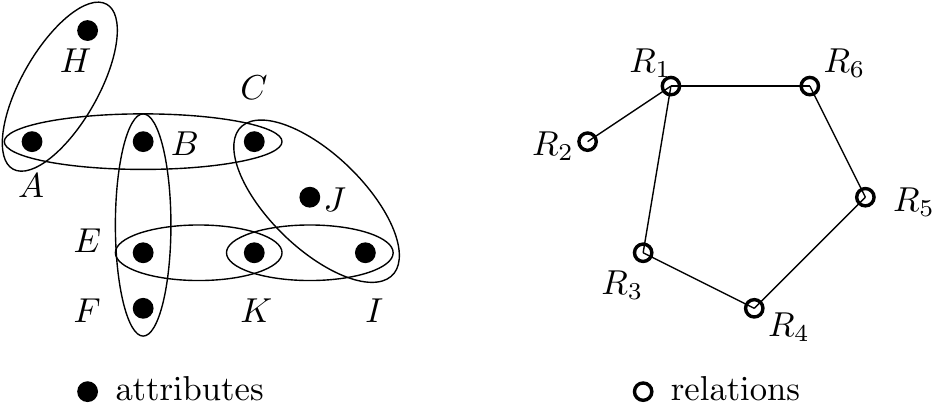}
	\caption{Hypergraph (left) and graph (right) representation for an example CQ $Q(A,C,F,K):-$ $R_1(A,B,C)$, $R_2(A,H)$, $R_3(B,E,F)$, $R_4(E,K)$, $R_5(K,I)$, $R_6(C,I,J)$.}
	\label{fig:representation}
\end{figure}

\subsection{Problem Definition}
Below, we formally define the \ourprob\ problem in terms of the count of output tuples of a CQ: 

\begin{definition}\label{def:problem}
	Given a CQ $Q$ on $\allrel$, an instance $D$, and a positive integer $k \geq 1$, the \emph{aggregated deletion propagation (\ourprob)} problem aims to remove at least $k$ results from $Q(D)$ by removing the minimum number of input tuples from $D$.
\end{definition}
Given $Q$, $k$, and $D$, we denote the above problem by $\ourprob(Q, D, k)$. Note that an implicit constraint on the input parameter $k$ is $1 \le k \le |Q(D)|$. For instance, in Figure~\ref{fig:example_setup}, \revm{\ourprob$(Q_1, D, 2)$} will return a single tuple $R_3(c3, e3)$ since removing it would remove the last two output tuples in $Q_1(D)$.  In this paper, we study the data complexity~\cite{vardi1982complexity} of the \ourprob\ problem, i.e., the size of the query and schema are fixed, and the complexity is in terms of the size of the database $D$. More precisely, we say that \ourprob($Q,D,k$) is {\em polynomial-time solvable} for a query $Q$ if, for an arbitrary instance $D$ and integer $k$, the solution of $\ourprob(Q,D,k)$ can be computed in polynomial time in the size of $D$; otherwise, it is {\em NP-hard}. 

For simplicity, we assume that all relations have distinct set of attributes in an input CQ $Q$, i.e., $\attr(R_i) \neq \attr(R_j)$ for every pair of relations $R_i, R_j \in \rel(Q)$. The rationale is that removing duplicated relations won't change the poly-time solvability of the original CQ. 

\subsection{Special Cases}
\label{sec:special-case}

Before we discuss the complexity of the \ourprob\ problem in general, we note the following special cases:

\paragraph{\reva{\ourprob} on boolean CQ.} 
The \ourprob\ problem on boolean CQ is also known as the {\em resilience} problem, i.e., removing the minimum number of input tuples to make the true query become false. The next theorem in~\cite{FreireGIM15} gives a decidability result of the \ourprob\ problem on boolean CQ.
\begin{theorem}[\cite{FreireGIM15}]
	\label{thm:boolean}
	On a boolean CQ $Q$, the poly-time solvability (in data complexity) of the $\ourprob(Q, D, 1)$ problem can be decided in polynomial time (in query complexity).
\end{theorem}

\paragraph{\reva{\ourprob} on CQ with vacuum relations.} The \ourprob\ problem becomes easy when $Q$ contains a vacuum relation. Consider an arbitrary input instance $D$ for $Q$ and integer $k$. If every vacuum relation in $Q$ has instance $\{\emptyset\}$, we can remove query results in $Q(D)$ by removing the tuple $\{\emptyset\}$ in any one vacuum relation; otherwise, $Q(D) = \emptyset$ by definition, and there is no need to remove anything. Therefore:
\begin{lemma}
	\label{lem:vacuum}
	For a CQ $Q$, if there exists some vacuum relation, the \ourprob$(Q, D,k)$\ problem is poly-time solvable (in data complexity). 
\end{lemma}

\noindent {\bf $\ourprob$ with different choices of $k$:}  When $k = |Q(D)|$ or $k = 1$, the \ourprob\ problem is equivalent to the resilience problem, which implies that $\ourprob(Q, D, k)$ is NP-hard even for a constant $k$ for general CQs.  In contrast, \ourprob\ can be shown to be poly-time solvable (in data complexity) for any fixed $k$ if the query $Q$ is a full CQ.
	
For full CQs, it is indeed the case that \ourprob$(Q , D, k)$ is polynomial-time solvable for constant $k$. Enumerate all $|Q(D)| \choose k$ ways of selecting the output tuples to be removed, which is polynomial in $|D| = n$ assuming data complexity. So, the problem reduces to finding a minimum set of input tuples whose removal results in a {\em fixed} set of $k$ output tuples being removed. Let us fix such a set of $k$ output tuples.  Now partition the $n$ input tuples into $2^k$ subsets depending on which subset of these $k$ output tuples they remove -- since the CQ is full, each input tuple will remove zero or more output tuples from the $k$ chosen output tuples. All input tuples in any subset of this partition behave identically with respect to the $k$ output tuples we chose to delete; hence, we can only keep any one of these input tuples. That leaves us with $2^k$ input tuples and the input size becomes constant for fixed $k$. So, by any brute force method (e.g., trivially enumerating all $2^{2^k}$ subsets of these $2^k$ input tuples), the problem can be solved in $O(2^{2^k})$ for the fixed set of $k$ output tuples. Overall, the running time becomes $O(|Q(D)|^k \cdot 2^{2^k})$ time, which is polynomial for fixed $k$.

%% file: sections/dichotomy.tex
\section{Poly-time Decidability %of the \ourprob\ Problem
} %%Procedural Dichotomy}
\label{SEC:DICHOTOMY}

In this section, we give an algorithm that can decide poly-time solvability of the \ourprob\ problem on general CQs.

\begin{theorem}\label{thm:dichotomy}
On a CQ $Q$, \isptime$(Q)$ can decide poly-time solvability of the $\ourprob(Q, D,k)$ problem, which runs in polynomial time. 
\end{theorem}

\begin{algorithm}[h]
	\caption{$\isptime(Q)$}
	\label{algo:isptime}
	
	Remove all universal attributes from each relation in $Q$\;
	\If{$\head(Q) = \emptyset$}{
		\If{there is no triad structure in $Q$}{	
			\Return \ \true
		}
	}
	\Else{
		\If{there exists a relation $R_i$ with $\attr(R_i) = \emptyset$}{
			\Return \ \true
		}
		\Else{
			\If{$Q$ is disconnected}{
				Let $Q_1, Q_2,\cdots, Q_s$ be its connected components\;
				\Return $\cap_{i=1}^s \isptime(Q_i)$\;
			}
		}
	}
	\Return \ \false
\end{algorithm}

The procedure \isptime$(Q)$\ is illustrated in Figure~\ref{fig:isptime}. Note that when $\isptime(Q)$ returns $\true$, the $\ourprob(Q, D,k)$ problem is poly-time solvable, and NP-hard otherwise.  The algorithmic description of \isptime\ is given in Algorithm~\ref{algo:isptime}.
$\isptime(Q)$ runs in polynomial time in the query size.

The high-level idea is to alternately apply two simplifications steps on the input query, until a ``base case'' is arrived at.
The first simplification step is that of removing all {\em universal} attributes in the input query. An attribute is {\em universal} if it is an output attribute appearing in all relations. After applying this step, if $Q$ becomes boolean or contains a vacuum relation (two of the base cases), it is decidable in polynomial time by Theorem~\ref{thm:boolean} and Lemma~\ref{lem:vacuum}. 

\par
Next, we check whether $Q$ is connected or not. For a disconnected query $Q$, we can {\em decompose} it into multiple {\em connected subqueries} as follows: apply breadth-first search or depth-first search algorithm on the graph $G_Q$, and find all connected components for $G_Q$. The set of relations corresponding to the set of vertices in one connected component of $G_Q$ form a connected subquery of $Q$. In this case, we perform the second simplification step of decomposing $Q$ into multiple connected subqueries, followed by calling  \isptime\ recursively on each connected subquery. More specifically, let $Q_1, Q_2, \cdots, Q_s$ be the connected subqueries of $Q$; then, $\isptime(Q)$ will return $\bigwedge_{i=1}^s \isptime(Q_i)$. Otherwise, $Q$ ends up in ``Others'' (the third base case). In this case, $Q$ is connected, non-boolean, and does not contain either a vacuum relation or a universal attribute. For all queries in ``Others'', \isptime\ returns \false. 
\begin{figure}
	\centering
	\includegraphics[scale=0.8]{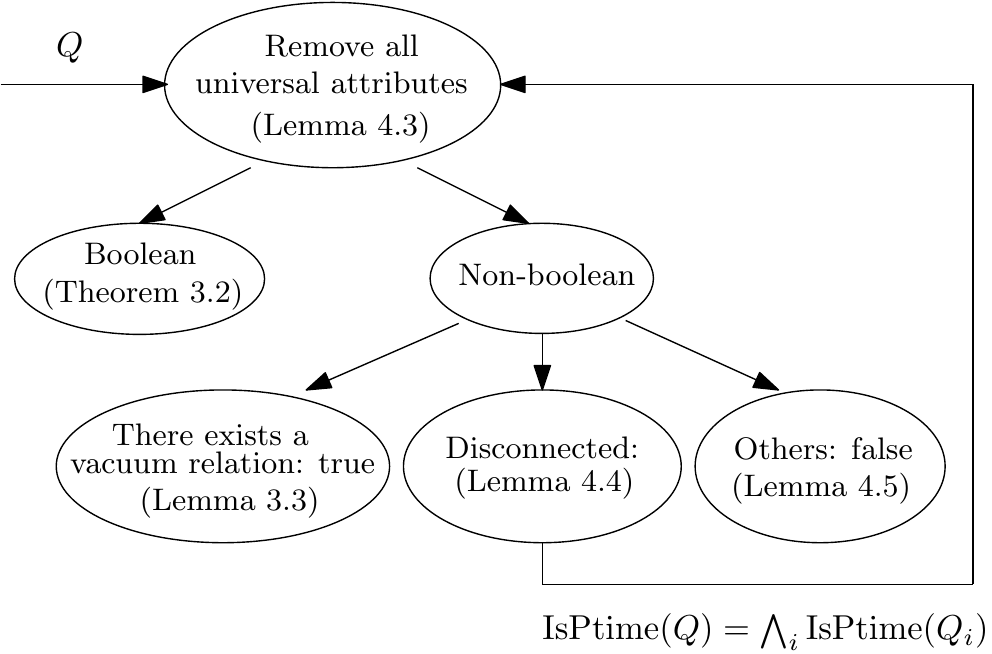}
	\caption{\revm{Procedure~\isptime$(Q)$.}}
	\label{fig:isptime}
\end{figure}

%We illustrate \isptime\ with an example below.

\begin{example}\label{eg:isptime}
	Consider an example CQ $Q(A,F,G,H):- R_1(A, B), \\R_2(F, G), R_3(B, C), R_4(C), R_5(G, H)$. Observe that $Q$ is non-boolean without any universal attribute and vacuum relations. The simplification step applied to $Q$ is to decompose it into two connected subqueries, $Q_1$ (with $R_1, R_3, R_4$) and $Q_2$ (with $R_2, R_5$). For $Q_2$, after removing the universal attribute $G$, it becomes disconnected. On applying the simplification step again to $Q_2$, it decomposes into two connected subqueries, $Q_{21}$ (with $R_2$) and $Q_{22}$ (with $R_5$).  After removing the universal attribute $F$ in $Q_{21}$, relation $R_2$ becomes vacuum and \isptime$(Q_{21})$ returns \true. Similarly, \isptime$(Q_{22})$ returns \true. However, $Q_1$ is non-boolean and contains no vacuum relation. Both simplifications fail on $Q_1$, so \isptime$(Q_1)$ returns \false. Therefore, \isptime($Q$) returns \false\ and $\ourprob(Q,D,k)$ is {\em NP-hard}.
\end{example}

The essence of \isptime\ is in the two simplifications steps: removing universal attributes and decomposing a disconnected query. Both these steps preserve the complexity of the problem as formally stated in Lemma~\ref{LEM:COMMON} and Lemma~\ref{LEM:DECOMPOSE}. Intuitively, for any universal attribute, we can partition the query results by the value of the universal attribute, and interpret each class in the partition as the result of the same query over a distinct sub-instance. Moreover, the deletion of any input tuple $t$ can only affect a single sub-instance that shares the value of the universal attribute with $t$. The original \ourprob\ instance now degenerates to finding an optimal combination of solutions to the \ourprob\ problem defined over each of the sub-instances, after removing the universal attribute.  Similarly, if the query is disconnected, the results of all connected subqueries will join by cross product. Then, the original \ourprob\ instance also degenerates to finding an optimal combination of solutions to the \ourprob\ problem defined for each connected subqueries. Finding the optimal combination is polynomial-time solvable since the size of the query as well as the query result is polynomial. Thus, the complexity of the original query can be deduced from that of the simplified queries.

Our proof of Theorem~\ref{thm:dichotomy} also follows the logical diagram of \isptime$(Q)$, which is divided into two parts. First, we show that these two simplification steps preserve the complexity of the problem, as described above. Then, we deal with the base cases. Note that the correctness for boolean queries and vacuum relations are implied by Theorem~\ref{thm:boolean} and Lemma~\ref{lem:vacuum}. Therefore, it suffices to show the NP-hardness of the \ourprob\ problem on $Q$, when $Q$ is non-boolean, connected, and contains no universal attribute or vacuum relation; we show this in Lemma~\ref{lem:others}. Putting everything together, the correctness for Theorem~\ref{thm:dichotomy} then follows from induction over the size of the query.

\subsection{Hardness Preservation in Simplifications}
\label{sec:hard-recursion}

In the first part, we show that when the simplifications are applied to the input query, the complexity of the \ourprob\ problem is preserved. 

\begin{lemma} 
	\label{LEM:COMMON}
	Let $A$ be a universal attribute in $Q$. Then, $\ourprob(Q,D,k)$ is NP-hard if and only if $\ourprob(Q_{-A},D,k)$  is NP-hard, where $Q_{-A}$ is the residual query after removing attribute $A$ from all relations in $Q$.
\end{lemma}

\begin{lemma}
	\label{LEM:DECOMPOSE}
	Let $Q_1, Q_2, \cdots, Q_s$ be the connected subqueries of $Q$ for $s \ge 2$. The $\ourprob(Q, D, k)$ problem is NP-hard if and only if there exists some $Q_i$ for which the $\ourprob(Q_i, D, k)$ problem is NP-hard. 
\end{lemma}

The proofs of these lemmas are similar in spirit. Namely, we have two parts corresponding to the ``if'' and ``only if'' directions. To prove the ``if'' direction, we show that if \ourprob\ is NP-hard for $Q_{-A}$ (resp., there exists some $Q_i$ for which \ourprob\ is NP-hard), then the \ourprob\ problem on $Q$ is also NP-hard. To prove the ``only-if'' direction, we show that if \ourprob\ is poly-time solvable for $Q_{-A}$ (resp., \ourprob\ is poly-time solvable for each connected subquery $Q_i$), then \ourprob\ is also poly-time solvable for $Q$ as well. More specifically, given a poly-time algorithm for solving \ourprob\ on $Q_{-A}$ (resp., given poly-time algorithms for solving \ourprob\ on each $Q_i$), we design a poly-time algorithm for solving \ourprob\ problem on $Q$. 

\begin{proof}[Proof of Lemma~\ref{LEM:COMMON}]
	{\bf The ``if'' direction.} Given any instance $D'$ for $Q_{-A}$, we construct another instance $D$ for $Q$ as follows. Consider any relation $R_i' \in \rel(Q_{-A})$. For each tuple $t' \in R_i'$, we create a new tuple $t \in R_i$ such that $\pi_A t = *$ (a fixed value for all tuples and all relations in attribute $A$), and $\pi_B t = \pi_B t'$ for every other attribute $B \in \attr(R_i) - A$. 
	
	Hence there is a one-to-one correspondence between the output tuples in $Q(D)$ and $Q_{-A}(D')$, and also in the input $D$ and $D'$. Therefore, a solution to $\ourprob(Q, D, k)$ of size $c$ corresponds to a solution to $\ourprob(Q_{-A}, k, D')$ of size $c$, and vice versa. The proof follows.
	
	{\bf The ``only-if'' direction.}  Assume there is a poly-time algorithm $\mathcal{A}$ for computing $\ourprob(Q_{-A}, D, k)$ for any instance $D$ and integer $k$. We design a poly-time algorithm $\mathcal{A}'$ for $\ourprob(Q, D, k)$ as follows:
	
	Consider any input instance $D$ for $Q$ and integer $k$. We first partition $D$ into $D_1, D_2, \cdots, D_g$ corresponding to $a_1, a_2, \cdots, a_g$, which are all the possible values in the domain of attribute $A$. In $D_i$, each tuple $t$ has $\pi_A t= a_i$. Note that the query result $Q(D)$ is a disjoint union of the subquery results $Q(D_1), Q(D_2),\cdots, Q(D_i)$.
	\par
	Now, we run a dynamic program to compute the optimal solution with cost $\optcost$. Let $\optcost[i][s]$ denote the minimum number of input tuples that have to be removed in order to remove at least $s$ output tuples from $Q(D)$, under the constraint that the input tuples can only be chosen from $D_1$ to $D_i$. Using this notation, we can now write the following dynamic program:
	\begin{equation}
	\label{eq:universal}
	\optcost[i][s] = \min_{m = 0}^{s}\Big \{ \optcost[i-1][s - m] + c_{i, m} \Big\}.
	\end{equation}
	Here, $m$ denotes the number of output tuples being removed from the subproblem on $D_i$. And, $c_{i, m}$ is the cost of the solution for subproblem \ourprob$(Q, D_i,m)$, i.e., the minimum number of input tuples in $D_i$ whose removal would remove at least $m$ output tuples from $Q(D_i)$. Note that $c_{i,0} = 0$ for every $i$. 
	
	Note that each tuple in $D_i$ has the same value $a_i$ in attribute $A$. Hence, computing \ourprob$(Q, D_i, m)$ is equivalent to computing \ourprob$(Q_{-A}, D_i, m)$, which can be solved in poly-time by algorithm $\mathcal{A}$. 
	Recall that there are $g$ distinct values in attribute $A$, thus $g \le |D|$. Moreover, $k$ is bounded by the size of query results, i.e. $k \le |Q(D)|$. The number of cells in $\optcost$ is $g \cdot k = O(|D| \cdot |Q(D)|)$, which is polynomial in terms of $|D|$. Thus, algorithm $\mathcal{A}$ runs in polynomial time in data complexity.
\end{proof}

\begin{proof}[Proof of Lemma~\ref{LEM:DECOMPOSE}]
	{\bf The ``if'' direction.} W.l.o.g., assume the \ourprob\ problem on $Q_1$ is NP-hard. Given an instance $D'$ for $Q_1$, we construct another instance $D$ for $Q$ as follows.  All relations in $Q_1$ have the same tuples as in $D'$. Set $L = |Q_1(D')| \cdot |D'|$. Recall that $|Q_1(D')|$ denotes the number of results in query $Q_1$ over instance $D'$. 
	Each relation $R_j \in \rel(Q_\ell)$ for $\ell \ge 2$ contains $L$ tuples, where each tuple is given a unique label that appears as the value of every attribute in that tuple. 
	(Note that the size of $D$ is  polynomial in the size of $D'$.) This ensures that for any connected subquery $Q_\ell$, there are exactly $L$ output tuples in $Q_\ell(D)$ corresponding to the $L$ unique labels given to the tuples in every relation. Then, the number of output tuples in $Q(D)$ is $|Q_1(D')| \cdot L^{s-1}$, since the join across the disconnected components results in a cross product. 
	
	We argue that $\ourprob(Q_1, D', k')$ has a solution of size $\le c$ if and only if $\ourprob(Q,  D, k' \cdot L^{s-1})$ has a solution of size $\le c$. 
	
	In one direction, if we can remove $k'$ results from $Q_1(D')$ by removing at most $c$ tuples from $D'$, removing these tuples from $D$ removes $k' \cdot L^{s-1}$ results from $Q(D)$, which is also a solution for $\ourprob(Q, D, k' \cdot L^{s-1})$. 
	
	In the other direction, suppose we are given a solution for $\ourprob(Q, D, k' \cdot L^{s-1})$ of size at most $c$. Observe that $c \le |D'|$; otherwise, there is always a better solution for $\ourprob(Q, D, k)$ by removing all input tuples from relations in $Q_1$. Let $x_i$ be the number of input tuples removed from relations in $Q_i$, and $y_i$ be the number of output tuples removed from $Q_i(D)$. 
	A key observation is that there exists a solution for $\ourprob(Q, D, k' \cdot L^{s-1})$ of size $\le c$ such that 
	(\romannumeral 1) $y_i = x_i$ for any $i \ge 2$; 
	(\romannumeral 2) $x_i \neq 0$ for at most one $i \ge 2$; and 
	(\romannumeral 3) $y_1 \ge k'$. 
	We will prove these one by one.
	
	For (\romannumeral 1), we can always remove $x_i$ output tuples from $Q_i(D)$ by removing $x_i$ tuples from one specific relation in $Q_i$. Thus, the total number of results removed can be written as: \begin{small}
		\[f(x_1, x_2, \cdots, x_s) 
		= |Q_1(D')| \cdot L^{s-1} - (|Q_1(D')| - y_1) \cdot \prod_{i \ge 2}^s(L - x_i) \ge k.\]
	\end{small}
	For (\romannumeral 2), suppose $s \ge 3$ and $x_2, x_3 \neq 0$ without loss of generality. We can construct another solution for $\ourprob(Q, D, k' \cdot L^{s-1})$ with $x'_i = x_i$ for $i \notin \{2,3\}$, $x_2' = x_2+x_3$, and $x_3' = 0$, which is no worse. This is because:
	\[f(x_1, x_2 + x_3,0, x_4, \cdots, x_s) \ge f(x_1, x_2, \cdots, x_s).\]
	After applying this argument repeatedly, we can obtain a solution for $\ourprob(Q, D, k' \cdot L^{s-1})$ that removes $x_1$ tuples from relations in $Q_1$ and $x_2$ tuples from relations in $Q_2$, where $x_1 + x_2 \le c$, with $\ge k$ results removed from $Q(D)$. 
	
	For (\romannumeral 3), suppose $y_1 < k'$. As $x_1 + x_2 \le c$, there comes
	\[f(x_1, c-x_1, 0,\cdots,0) \ge f(x_1, x_2, 0, \cdots, 0) \ge k\] 
	Expanding $f(x_1, c-x_1, 0,\cdots,0)$ and $k$, we get: 
	\[|Q_1(D')| \cdot L^{s-1} - (|Q_1(D')| -y_1)(L - c + x_1) \cdot L^{s-2} \ge k' \cdot L^{s-1}\]
	Rearranging this inequality, we will get
	\begin{align*}
	(|Q_1(D')| - k') \cdot L \ge& (|Q_1(D')| -y_1)(L - c+ x_1) 
	\ge (|Q_1(D')| -y_1)(L - |D'|) 
	\end{align*}
	where the last inequality is implied by the fact that $c \le |D'|$. We can further rewrite the inequality above as 
	\[ (|Q_1(D')| -y_1) \cdot |D'| \ge (k' -y_1) \cdot L > L \ (\textrm{since }  y_1 < k').\]
	This contradicts: $(|Q_1(D')| -y_1) \cdot |D'| \le |Q_1(D')| \cdot |D'| = L$. 
	
	Thus, there exists a solution for $\ourprob(Q, D, k' \cdot L^{s-1})$ of size $\le c$ such that $y_1 \ge k'$. Removing those $x_1$ tuples from relations in $Q_1$ is a solution for $\ourprob(Q_1, D', k')$ of size $\le c$. 
	
	{\bf The ``only-if'' direction.} Assume that for each $Q_i$, there is a poly-time algorithm $\mathcal{A}_i$ for computing $\ourprob(Q_i, D, k)$ for any instance $D$ and integer $k$. We next present another poly-time algorithm $\mathcal{A}$ for $\ourprob(Q, D, k)$. Consider an arbitrary input instance $D$ and integer $k$. Let $|Q_i(D)| = m_i$. Note that if removing $k_i$ output tuples from $Q_i(D)$, there are $m_i - k_i$ remaining output tuples in $Q_i(D)$, which together form $\prod_{i=1}^s(m_i - k_i)$ output results overall. In other words, $\prod_{i=1}^s m_i - \prod_{i=1}^s(m_i - k_i)$ output tuples are removed from $Q(D)$ in total. Therefore, the overall optimal solution is given by:
	\begin{equation}
	\label{eq:decompose}
	%	\begin{split}
	\ourprob(Q, D, k) = 
	\min_{(k_1, k_2, \cdots, k_s) \in K} \sum_{i=1}^s \ourprob(Q_i, D, k_i)	
	%	\end{split}
	\end{equation}
	where $K =\{(k_1, k_2, \cdots, k_s): \prod_{i=1}^s m_i - \prod_{i=1}^s(m_i - k_i) \geq k, k_i \in \mathbb{Z}^+, \forall i \in \{1,2,\cdots,s\}\}$.
	Note that the  $\ourprob(Q_i, D, k_i)$ is solved in polynomial time by algorithm $\mathcal{A}_i$. 
	Note that there are at most $k^s = O(|Q(D)|^s)$ different combinations of $k_1, k_2, \cdots, k_s$, which is still polynomial in terms of data complexity. Overall, the running time of $\mathcal{A}$, which simply enumerates all these options and chooses the best one, is polynomial.  
\end{proof}

\subsection{NP-Hardness
	%when $\isptime$ goes to 
	for ``Others''}
\label{sec:hard-all-fail}

In this part, we prove the hardness of the class of queries characterized by ``others'' bracket in Figure~\ref{fig:isptime}, as stated in Lemma~\ref{lem:others}.

\begin{lemma}
	\label{lem:others}
	For a CQ $Q$, if $\isptime(Q)$ goes to ``others'' in Figure~\ref{fig:isptime}, i.e., if (1) $Q$ contains no universal attributes; (2) $Q$ is non-boolean; (3) $Q$ contains no vacuum relations;  and (4) $Q$ is connected, then $\ourprob(Q, D, k)$ is NP-hard.
\end{lemma}

We start by identifying three simple but {\bf NP-hard} queries for the \ourprob\ problem that will be at the core of showing the above lemma. Then we present a general framework of proving the hardness for a given CQ by {\em mapping} it to another query on which the \ourprob\ problem is known (or has been proven) to be NP-hard. Finally, we classify all queries in Lemma~\ref{lem:others} into three groups using the flowchart in Figure~\ref{fig:others}, and give a mapping from queries ending up in each leaf of the flowchart to a core query identified at the beginning.  
\begin{figure}[t]
	\centering 
	\includegraphics[scale=0.75]{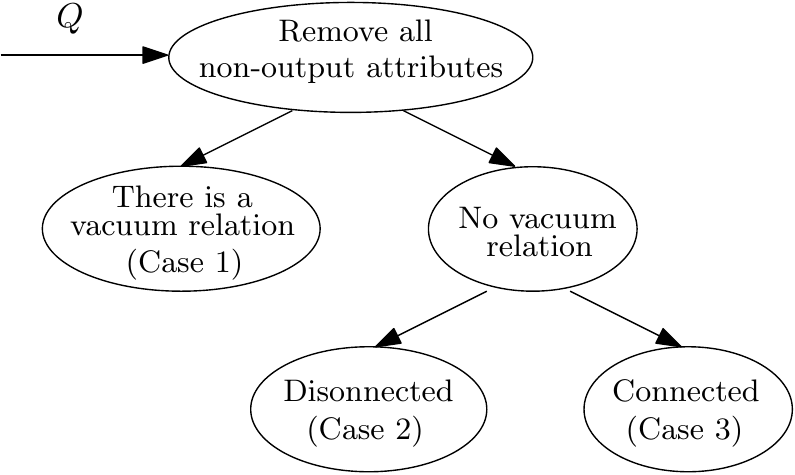}
	\caption{\small Proof plan of Lemma~\ref{lem:others}.}
	\vspace{-3mm}
	\label{fig:others}
\end{figure}

\subsubsection{Core Queries}
\label{sec:hardest}

The three queries we focus on are the following:
\begin{enumerate}
	\itemsep0em
	\item[] $Q_{\textup{cover}}(A, B) :- R_1(A), R_2(A, B), R_3(B)$.
	\item[] $Q_{\textup{swing}}(A): -R_2(A,B), R_3(B)$.
	\item[] $Q_{\textup{seesaw}}(A): - R_1(A), R_2(A,B), R_3(B)$.
\end{enumerate}
Careful inspection reveals that these queries have a common property: w.l.o.g., we can assume that an optimal solution of $\ourprob(Q,D,k)$ won't remove any tuples from relation $R_2(A,B)$. The effect of the removal of any tuple $(a,b) \in R_2$ can also be achieved by removing tuple $(a) \in R_1$ or $(b) \in R_3$. %(The formal proof is deferred to the full version of the paper~\cite{hu2020ADP}.)
(The formal proof is in Appendix~\ref{appendix:endogenous}.) 
Therefore, an optimal solution for \ourprob\ on any one of these three queries could be restricted to removing tuples only from $R_1(A)$ and $R_3(B)$. In this way, the \ourprob\ problem on these queries can be interpreted as optimization problems on bipartite graphs, which turn out to be {\em NP-hard} (Lemma~\ref{LEM:BIPARTITE-GRAPH}). 

\begin{lemma}
	\label{LEM:BIPARTITE-GRAPH}
	Given an undirected bipartite graph $G(A \cup B, E)$ where $E$ is the set of edges between two sets of vertices $A$
	and $B$, and an integer $k$, each of the following problems is NP-hard:
	\begin{enumerate}
		\itemsep0em
		\item[(1)] Remove the minimum number of vertices in $A \cup B$ such that at least $k$ edges in $E$ are removed.\footnote{\noindent A {\bf remove} procedure on a graph is defined as: (1) when a vertex is removed, all the incident edges are also removed; (2) when all the incident edges on a vertex are removed, this vertex is also removed. \label{remove}}
		\item[(2)] Remove the minimum number of vertices in $B$ such that at least $k$ vertices in $A$ are removed; %\footref{remove};
		\item[(3)] Remove the minimum number of vertices in $A \cup B$ such that at least $k$ vertices in $A$ are removed; %\footref{remove};
	\end{enumerate}
\end{lemma}

Problem (1) is exactly {\em partial vertex cover for bipartite graphs}, which is known to be NP-hard~\cite{CaskurluMPS17}. The NP-hardness proofs for (2) and (3) are deferred to  Appendix~\ref{APPENDIX:HARD-CORE}.

\subsubsection{Hardness Preserving Mapping}

The high-level idea of relating an arbitrary query $Q$ characterized by Lemma~\ref{lem:others} to the core queries is to divide the attributes in $\attr(Q)$ into two groups, one mapped to $A$ and the other mapped to $B$. In this way, each relation in $Q$ plays the role of $R_1(A)$, $R_2(A,B)$ or $R_3(B)$ in the core queries. The notion of ``query mapping'' is formally defined below:
\begin{definition}[Query Mapping]
	Suppose we are given a function $f:\attr(Q_1) \to \attr(Q_2) \cup \{*\}$. Let
	\[g(R_i) = \{Y\in \attr(Q_2): \exists X \in \attr(R_i) \text{ s.t. } f(X) = Y\}.\]
	$f$ is said to be a query mapping if the following properties hold:
	(\romannumeral 1) 
	for every relation $R_i \in \rel(Q_1)$, there is a (unique) relation $R_j \in \rel(Q_2)$ such that
	$g(R_i) = \attr(R_j)$.
	(\romannumeral 2) for every relation $R_j \in \rel(Q_2)$, there exists at least one relation $R_i \in \rel(Q_1)$ such that $g(R_i) = \attr(R_j)$.
\end{definition}
In the definition above, if $g(R_i) = \attr(R_j)$ for relations $R_i \in \rel(Q_1)$ and $R_j \in \rel(Q_2)$, then $R_i \in \rel(Q_1)$ is said to be {\em mapped} to relation $R_j \in \rel(Q_2)$. The next lemma %, whose proof is deferred to the full version~\cite{hu2020ADP}, 
shows that query mappings preserve hardness of the \ourprob\ problem.

\begin{lemma}
	\label{lem:mapping-hard}
	If there is a mapping from a CQ $Q_1$ to another CQ $Q_2$, and $\ourprob(Q_2, D, k)$ is NP-hard, then $\ourprob(Q_1, D, k)$ is also NP-hard.
\end{lemma}

\begin{proof}
	Assume $Q_1$ is mapped to $Q_2$ under the mapping function $f$.  Given any instance $D_2$ for $Q_2$, we construct an instance $D_1$ for $Q_1$ as follows.  Consider an arbitrary relation $R_i \in \rel(Q_1)$ that is mapped to relation $R_j \in \rel(Q_2)$ under $f$. If there is a tuple $t' \in R_j$, we create a tuple $t \in R_i$ such that for any $X \in \attr(R_i)$, $\pi_{X} t = \pi_{f(X)} t'$ if $f(X) \in \attr(R_j)$, and $\pi_{X} t = *$ otherwise. Overloading notation, we will say that $t$ is also mapped to $t'$. Note that there is a one-to-one correspondence between the output tuples in $Q_1(D_1)$ and $Q_2(D_2)$.

	We next show that the problem $\ourprob(Q_1, D_1, k)$ has a solution of size $\le c$ if and only if $\ourprob(Q_2, D_2, k)$ has a solution of size $\le c$. 
	
	{\bf The ``only-if'' direction.} Suppose we are given a solution $\S_1$ for $\ourprob(Q_1, D_1, k)$ has of size $\le c$. We next construct a solution $\S_2$ for $\ourprob(Q_2, D_2, k)$ as follows. For any relation $R_i \in \rel(Q_1)$, if tuple $t \in R_i$ is removed by $\S_1$, then tuple $t' \in R_j$
	is removed by $\S_2$, where $R_i \in \rel(Q_1)$ is mapped to $R_j \in \rel(Q_2)$ and $t$ is mapped to $t'$. 
	Since multiple tuples from different relations in $D_1$ could be mapped to $t'$, $|\S_2| \le |\S_1| \le c$.
	As a result, if an output tuple from $Q_1(D_1)$ is removed, its corresponding tuple from $Q_2(D_2)$ will also be removed. Thus, $\S_2$ removes at least $k$ results from $Q_2(D_2)$, with size $\le c$.
	
	{\bf The ``if'' direction.} Suppose we are given a solution $\S_2$ for $\ourprob(Q_2, D_2, k)$ of size $\le c$. We next construct a solution $\S_1$ for $\ourprob(Q_1, D_1, k)$ as follows. Consider any relation $R_j \in \rel(Q_2)$ with some tuples removed by $\S_2$. Let $R_i \in \rel(Q_1)$ be any one relation mapped to $R_j$ under $f$. If $t' \in R_j$ is removed,  remove the tuple $t$in $\S_1$ that is mapped to $t'$. Clearly, $|\S_1| = |\S_2| \le c$.
	As a result, if an output tuple from $Q_2(D_2)$ is removed, its corresponding tuple from $Q_1(D_1)$ will also be removed. Thus, $\S_1$ removes at least $k$ results from $Q_1(D_1)$, with size $\le c$.
\end{proof}

\subsubsection{Mapping to the core}
\label{SEC:MAPPING}
To prove the NP-hardness of the \ourprob\ problem on a query $Q$, it suffices to show a mapping to any core query, implied by Lemma~\ref{lem:mapping-hard}. The high-level idea is that for any query characterized by Lemma~\ref{lem:others}, we find a partition of attributes in $Q$ as $(\I, \J, \attr(Q) - \I -\J)$ where $\I \cap \J = \emptyset$ and define the mapping function $f: X \to \{A,B,*\}$ as follows:
\begin{displaymath}
f(X) = \left\{ \begin{array}{ll}
A & \textrm{if $X \in \I$}\\
B & \textrm{if $X \in \J$}\\
* & \textrm{otherwise}
\end{array} \right.
\end{displaymath}
Then it remains to show that $f$ is a mapping from $Q$ to one of the three core queries. As mentioned, we distinguish $Q$ into three cases in Figure~\ref{fig:others}, and identify the mapping for each case separately. %%Due to the page limit, t
%The mapping constructed for each case %, as well as specific
% with examples are given in the full version~\cite{hu2020ADP}. %Appendix~\ref{appendix:mapping}.

Note that any query in Lemma~\ref{lem:others} is connected and does not have any universal attribute or vacuum relation. For simplicity, {\em head join} is defined as the residual query after removing all non-output attributes from all relations in $Q$, denoted as $\headQ$. In a CQ $Q$, a {\em path} between a pair of attributes $A,B \in \attr(Q)$, is a sequence of relations starting with some $R_i \in \rel(A)$ and $R_j \in \rel(B)$ such that each consecutive pair of relations share a common attribute. %Due to the page limit, the detailed mappings constructed for each case in Figure~\ref{fig:others} are in the full version~\cite{hu2020ADP}.

\paragraph{Case 1: Head join has at least one vacuum relation.}
In this case, observe that there must exist some relation $R_i \in \rel(Q)$ such that $\attr(R_i) \subseteq \attr(Q) - \head(Q)$.  Let $\I = \head(Q)$ and $\J = \attr(Q) - \head(Q)$. We next show that $f$ is a valid mapping from $Q$ to $\swingquery$ if there exists some relation $R_j \in \rel(Q)$ such that $\attr(R_j) \subseteq \head(Q)$, and to $\seesawquery$ otherwise.

Note that every relation $R_i \in \rel(Q)$ is mapped to $R_1(A)$, $R_2(A,B)$, or $R_3(B)$. 
Crucially, there is at least one relation that is mapped to $R_3(B)$, e.g., $R_i$.  Moreover, there is
at least one relation that is mapped to $R_2(A,B)$; otherwise attributes in $\I$ and $\J$ are not connected, contradicting the fact that $Q$ is connected. (Note that $Q$ is connected irrespective
of whether the head join is connected or not.) If there exists some relation $R_j \in \rel(Q)$ such that $\attr(R_j) \subseteq \head(Q)$, then $R_j$ will be mapped to $R_1(A)$; and 
$f$ is a valid mapping from $Q$ to $\seesawquery$. Otherwise, $f$ is a valid mapping from $Q$ to $\swingquery$.

\paragraph{Case 2: Head join is disconnected (and no vacuum relation).}
In this case, we can always identify a pair of attributes $X,Z \in \head(Q)$ such that there is no path between $X,Z$ in $\headQ$. As $Q$ is connected, every path between $X,Z$ in $Q$ uses at least one attribute in $\attr(Q) - \head(Q)$. In other words, removing $\attr(Q) - \head(Q)$ decomposes $Q$ into multiple connected subqueries, where $X,Z$ are in different ones.
Let $\I$ be the set of attributes appearing in the connected subquery containing $X$. Note that $\head(Q) - \I \neq \emptyset$ since $X,  Z$ are in different connected subqueries. 

Observe that there must exist a relation $R_\ell \in \rel(Q)$ such that $\attr(R_\ell) \cap \I
\neq \emptyset$ and $\attr(R_\ell) \cap (\attr(Q) - \head(Q)) \neq \emptyset$; otherwise, there is no path between $X$ and any non-output attribute, contradicting the fact that $Q$ is connected. Applying a similar argument to the connected subquery that doesn't contain $X$, there must exist a relation $R_h \in \rel(Q)$ such that $\attr(R_h) \cap (\head(Q) - \I)
\neq \emptyset$ and $\attr(R_h) \cap (\attr(Q) - \head(Q)) \neq \emptyset$.
Depending on whether there exists some relation $R_i \in \rel(Q)$ such that $\attr(R_i) \subseteq \I$ and some relation $R_j\in \rel(Q)$ such that $\attr(R_j) \subseteq \head(Q) - \I$, we have two different cases.

\textbf{Case 2.1:} Both relations $R_i$ and $R_j$ as described above exist. Set $\J = \attr(Q) - \I$. On one hand, each relation in $Q$ is mapped to any one of $R_1(A)$, $R_2(A,B)$ or $R_3(B)$. On the other hand, relations $R_i, R_\ell, R_j$ are mapped to $R_1, R_2, R_3$ respectively. Thus, $f$ is a valid mapping from $Q$ to $\pathquery$.

\textbf{Case 2.2:} At least one of $R_i, R_j$ doesn't exist, say $R_j$.  Set $\J = \attr(Q) - \head(Q)$. In this mapping, no relation has all of its attributes mapped to $*$; otherwise, $R_j$ exists, which is a contradiction. So, each relation in $Q$ is mapped to any one of $R_1(A)$, $R_2(A,B)$ or $R_3(B)$. On the other hand, relations $R_\ell, R_h$ are mapped to $R_2, R_3$ respectively. If $R_i$ exists, it will be mapped to $R_1(A)$ and $f$ is a valid mapping from $Q$ to $\seesawquery$. Otherwise, $f$ is a valid mapping from $Q$ to $\swingquery$.

\paragraph{Case 3: Head join is connected (and no vacuum relation).}
In this case, the head join is connected but has no vacuum relation. We further distinguish $Q$ into two cases: (3.1) there exists a pair of relations $R_i, R_j \in \rel(Q)$ such that $\attr(R_i) \cap \attr(R_j) \cap \head(Q) = \emptyset$; (3.2) for each pair of relations $R_i, R_j \in \rel(Q)$, we have $\attr(R_i) \cap \attr(R_j) \cap \head(Q) \neq \emptyset$.

\textbf{Case 3.1.} Set $\I = \attr(R_i) \cap \head(Q)$ and $\J = \head(Q) - \attr(R_i)$. In this mapping, no relation has its all attributes mapped to $*$; otherwise, there is a vacuum relation in the head join, which is a contradiction. 
So, each relation in $Q$ is mapped to any one of $R_1(A)$, $R_2(A,B)$ or $R_3(B)$. Moreover,  $R_i, R_j$ are mapped to $R_1(A), R_3(B)$ respectively. Note that there must also exist some relation mapped to $R_2(A,B)$; otherwise, $R_i$ is a single connected subquery of the head join, contradicting the fact that the head join is connected. Thus, $f$ is a valid mapping from $Q$ to $\pathquery$.

\textbf{Case 3.2.} In this case, we first observe that $|\attr(R_i) \cap \head(Q)| \ge 2$ for any relation $R_i \in \rel(Q)$. Suppose not, say $\attr(R_i) \cap \head(Q) = \{C\}$. Since $\attr(R_i) \cap \attr(R_j) \cap \head(Q) \neq \emptyset$ for any $R_j \in \rel(Q)$, then $C$ is a universal attribute of $Q$, which is a contradiction. For simplicity, assume no pair of relations in the head join have exactly the same attributes; otherwise, we just keep one of them in the mapping construction.

We label all relations in an increasing order of the number of output attributes, as $R_1, R_2, \cdots, R_p$, breaking ties arbitrarily.
For simplicity, denote $\attr(R_i) \cap \attr(R_j) \cap \head(Q)$ as $\allattr_{ij}$ with ordering $(i,j)$
if $i < j$, and $\allattr_{ji}$ with ordering $(j,i)$ otherwise. Let $R_i, R_j$ be the pair of relations whose intersection contains smallest number of output attributes. Without loss of generality, assume $i<j$. If there are multiple pairs with the same number of attributes in their intersection, we just break ties by their lexicographical order. 
We further distinguish the mappings into two cases as follows.

Case 3.2.1: $i>1$. We observe that $\allattr_{1i} - \allattr_{1j} \neq \emptyset$ and $\allattr_{1j} - \allattr_{1i} \neq \emptyset$. Suppose not, say $\allattr_{1i} - \allattr_{1j} = \emptyset$. This implies $\allattr_{1i} \subseteq \allattr_{1j} \subseteq \allattr_{ij}$, contradicting the fact that $\allattr_{ij}$ has smaller number of attributes than $\allattr_{1i}$. (Note that $(1, i)$ is lexicographically earlier than $(i, j)$ in the case of a tie.) Similarly, we can also show that $\allattr_{1j} - \allattr_{1i} \neq \emptyset$.  Moreover, there exists no relation $R_\ell \in \rel(Q)$ such that $\attr(R_\ell) \cap \head(Q) \subseteq \allattr_{ij}$. This is because of the fact that no pair of relations have exactly the same attributes, which in combination with $\attr(R_\ell) \cap \head(Q) \subseteq \allattr_{ij}$ would imply that  $\attr(R_\ell) \cap \head(Q) \subsetneq R_i\cap \head(Q)$. This would in turn imply $\ell < i$, and consequently, that $\allattr_{\ell i}$ has smaller number of attributes than $\allattr_{ij}$ (or is lexicographically earlier in the case of a tie), which is a contradiction.

Set $\I = (\attr(R_i) \cap \head(Q)) - \attr(R_j)$ and $\J = \head(Q) - \attr(R_i)$.
In this mapping, no relation gets all attributes mapped to $*$,
since there is no relation $R_\ell$ such that $\attr(R_\ell)\cap \head(Q)\subseteq \allattr_{ij}$ as discussed above. So, each relation in $Q$ is mapped to any one of $R_1(A)$, $R_2(A,B)$ or $R_3(B)$. Moreover, relations $R_i, R_1, R_j$ are mapped to $R_1, R_2, R_3$ respectively. Thus, $f$ is a valid mapping from $Q$ to $\pathquery$.

Case 3.2.2: $i=1$. For any attribute $C \in \allattr_{1j}$, there must exist a relation $R_\ell$ such that $C \notin \attr(R_\ell)$; otherwise, $C$ is an universal attribute, which is a contradiction. W.l.o.g., assume $\ell < j$. We claim that $\allattr_{\ell j} - \attr(R_1) \neq \emptyset$; otherwise, $\allattr_{\ell j}  \subseteq \allattr_{1j}$. Since $C \in \allattr_{1j} - \allattr_{\ell j}$, $|\allattr_{\ell j}| < |\allattr_{1j}|$, contradicting the fact that $R_1, R_j$ share the smallest number of output attributes among all pair of relations.  Moreover, there exists no relation $R_h \in \rel(Q)$ such that $\attr(R_h) \cap \head(Q) \subseteq \allattr_{1 \ell}$. Otherwise, either $\attr(R_h) \cap \head(Q) \subsetneq \attr(R_1) \cap \head(Q)$ which contradicts the fact that $h > 1$,  or $\attr(R_h) \cap \head(Q) = \attr(R_1) \cap \head(Q)$ which contradicts the fact that no pair of relations have exactly the same output attributes.

Set $\I = (\attr(R_1) \cap \head(Q)) - \attr(R_\ell)$ and $\J = \head(Q) - \attr(R_1)$. In this mapping, no relation gets all attributes mapped to $*$,  since there exists no relation $R_h$ such that $\attr(R_h) \cap \head(Q) \subseteq \allattr_{1 \ell}$ as discussed above. So, each relation in $Q$ is mapped to any one of $R_1(A)$, $R_2(A,B)$ or $R_3(B)$. Moreover, $R_1, R_j, R_\ell$ are mapped to $R_1, R_2, R_3$ respectively. Thus, $f$ is a valid mapping from $Q$ to $\pathquery$.

\medskip 
We show examples for each case in Figure~\ref{fig:others} separately. 

\begin{example} Consider an example query $Q_1(A,C,F):- R_1(A,C), R_2(B), R_3(B,C), R_4(C,E,F)$, with a vacuum relation $R_2$ in head join $Q'_1(A,C,F):- R_1(A,C), R_2(), R_3(C), R_4(C,F)$. In this example, we map attributes $A,C,F$ to $A$ and $B,C$ to $B$, yielding a new query $Q''_1(A):- R_1(A), R_2(B), R_3(B), R_4(A,B)$, i.e., the $\seesawquery$ query.  If $R_1(A,C)$ does not appears in $Q_1$, the same mapping yields another query $Q'''_1(A):- R_2(B),R_3(B), R_4(A,B)$, i.e., the $\swingquery$ query.
\end{example}

\begin{example}
	Consider an example query $Q_2(A,B):-R_1(A),R_2(A,C), R_3(C,B),R_4(B)$, where the head join $Q'_2(A,B):- R_1(A), R_2(A),R_3(B),R_4(B)$ is disconnected. For one connected subquery containing $A$, we can identify relation $R_2$ such that $A \in \attr(R_2)$ and $\attr(R_2)\cap (\attr(Q) - \head(Q)) \neq \emptyset$. Similarity, for the other connected subquery containing $B$, we can identify relation $R_3$ such that $B \in \attr(R_3)$ and $\attr(R_3)\cap (\attr(Q) - \head(Q)) \neq \emptyset$. In this case, we map attributes $B,C$ to $B$, yielding a new query $Q'_2(A,B):-R_1(A), R_2(A,B), R_3(B)$, i.e., the $\pathquery$ query. If $R_4(B)$ does not appear in $Q_2$, we map $B$ to $*$, yielding a new query $Q''_2(A):- R_1(A), R_2(A,C),R_3(C)$, i.e., the $\seesawquery$ query. If both $R_1(A), R_4(B)$ does not appear in $Q_2$, we map $B$ to $*$, yielding a new query $Q'''_2(A):- R_2(A,C),R_3(C)$, i.e., the $\swingquery$ query.
\end{example}

\begin{example}
	We show two examples for (3.1) and (3.2) separately.  In (3.1), there is a pair of relations $R_i, R_j \in \rel(Q)$ such that $\attr(R_i) \cap \attr(R_j) = \emptyset$. Consider a full CQ $Q_3(A,B,C,E):- R_1(A,C), R_2(C,E),\\ R_3(E,B)$. There is a pair of relations $R_1, R_3$ such that $\attr(R_1) \cap \attr(R_3)  \neq \emptyset$. We map attributes $A,C$ to attribute $A$ and $B,E$ to attribute $B$, yielding a new query $Q'_3(A,B):- R_1(A),R_2(A,B), R_3(B)$, i.e., the $\pathquery$ query. In (3.2), for every pair of relations $R_i, R_j \in \rel(Q)$, $\attr(R_i) \cap \attr(R_j) \cap \head(Q) \neq \emptyset$. Consider an example full CQ $Q_4(A,B,C,E,F):- R_1(A,B,C,E,F),R_2(B,C,E), R_3(A,C)$. We map attributes $C,E,F$ to $*$ and obtain a new query $Q'_4(A,B):- R_1(A,B),R_2(B), R_3(A)$, i.e., the $\pathquery$ query.
\end{example}

%% file: sections/structure.tex
\section{Structural Characterization}
\label{sec:structure}

In the last section, we provided a simple poly-time algorithm \isptime\ to decide the poly-time solvability of the \ourprob\ problem for CQs without self-join. \revm{However, this algorithm does not provide structural insight into what makes the \ourprob\ problem NP-hard or poly-time solvable for individual queries. Namely, it does not provide a structural characterization for solvability of the \ourprob\ problem, %on individual queries, 
such as the one shown for the special case of the resilience problem in \cite{FreireGIM15}. To rectify this shortcoming and complement the procedural dichotomy established in the last section, we provide, in this section, a {\em structural dichotomy} of the \ourprob\ problem for CQs. Interestingly, it turns out that the procedural and structural dichotomies do not have a one-one mapping; namely, distinct cases of the \isptime\ procedure map to same case in the structural characterization, and vice-versa. Our main theorem in this section is the following:
}
%However, when \isptime\ goes into the ``other'' bucket, the classification of CQs becomes very complex (as shown in Figure~\ref{fig:others}). Moreover, on the special case of resilience problem, the essence of \cite{FreireGIM15} is a  characterization for the ``triad'' structure, which precisely captures the hardness of \ourprob\ on  boolean CQ. Both motivate us to further investigate the common structures of all NP-hard queries.}
%since \isptime\ only rules out the existence of certain special structures (universal attributes, disconnected query, vacuum relation,  boolean query), and fails to capture the common structures shared by all NP-hard CQs.}
%In this section, we complement $\isptime$ with a structural dichotomy for the \ourprob\ problem. 
\begin{theorem}
	\label{THM:DICHOTOMY-STRUCTURE}
	For a CQ $Q$, $\ourprob(Q,k,D)$ is NP-hard if and only if one of the following happens:
	\begin{itemize}
	    \item $Q$ contains a ``triad-like'' structure,
	    \item $Q$ contains a ``strand'' structure, or
	    \item the head join of non-dominated relations is non-hierarchical.
	\end{itemize}
\end{theorem}

%Due to space limit, 
In the rest of this section, we explain the the three ``hard structures'' in Theorem~\ref{THM:DICHOTOMY-STRUCTURE} and give some intuition for why they make the \ourprob\ problem NP-hard. The proof of Theorem~\ref{THM:DICHOTOMY-STRUCTURE} is given in Appendix~\ref{appendix:dichotomy-structure}.

\subsection{Boolean CQ Revisited} 

As mentioned earlier, a complete characterization of boolean CQs for the \ourprob\ problem is known from previous work: 
\begin{theorem}[\cite{FreireGIM15}]
\label{thm:boolean-dichotomy}
	On a boolean CQ $Q$ without self-joins, the problem $\ourprob(Q,D,1)$ is poly-time solvable if there is no triad structure, and NP-hard otherwise.
\end{theorem}

To explain this result, we introduce some new terminology. In a CQ $Q$, a relation $R_j \in \rel(Q)$ is {\em exogenous} if there exists another relation $R_i  \neq R_j \in \rel(Q)$ such that $\attr(R_i) \subsetneq \attr(R_j)$, and {\em endogenous} otherwise. If there is more than one relation defined on the same set of attributes, we just consider any one of them as {\em endogenous} and the remaining ones as {\em exogenous}. For example, in the boolean CQ $Q:- R_1(A), R_2(A,B), R_3(B,C), R_4(B,C), R_5(B,C)$, there are two endogenous relations: $R_1$ and any one of $R_3$, $R_4$, $R_5$. Next, we define a {\em path} between a pair of relations $R_i, R_j \in \rel(Q)$ as a path between any pair of attributes $A,B$ for $A \in \attr(R_i)$ and $B \in \attr(R_j)$. This brings us to the definition of the {\em triad} structure:

\begin{definition}[triad]
	\label{def:triad}
	A triad is a triple of endogenous relations $R_1, R_2, R_3$ such that for each pair of relations, say $R_1, R_2$, there is a path from $R_1$ to $R_2$ only using any attributes in $\attr(Q) - \attr(R_3)$. 
	\end{definition}

Two examples of boolean CQs containing a triad structure are $Q_\triangle:- R_1(A,B), R_2(B,C), R_3(C,A)$ and $Q_T:- R_1(A,B,C)$, $R_2(A)$, $R_3(B)$, $R_4(C)$, on which the \ourprob\ problem is NP-hard.%% by Theorem~\ref{thm:boolean-dichotomy}.

\subsection{Hard Structures for General CQs} 

A natural question for general CQs is how the existence of output attributes changes the hardness of \ourprob\ problem. We will explore this question %step by step by identifying 
starting with three hard structures. 

\subsubsection{Triad-like}

We observe that adding output attributes to a hard boolean CQ maintains the NP-hardness of the \ourprob\ problem. For example, the CQ $Q(E,F,G):-R_1(A,B,E), R_2(B, C,F), R_3(C,A,G)$ is NP-hard (since \isptime\ returns \false), which contains the $Q_\triangle$. \revc{We extend the notion of triad to capture this class of hard queries:} %We generalize the notion of triad to characterize this class of hard queries:
\begin{definition}[triad-like]
	\label{def:triad-like}
	A triad-like structure is a triple of endogenous relations $R_1, R_2, R_3$ such that for each pair of relations, say $R_1, R_2$, there is a path from $R_1$ to $R_2$ only using attributes in $\attr(Q) - (\head(Q) \cup \attr(R_3))$. %\alert{I changed the expression, please check.}\xiao{Correct.}
\end{definition}
This takes care of our first case: if there is a triad-like structure (in the non-output attributes), the CQ is NP-hard.

\subsubsection{Non-hierarchical Join}
The situation becomes more complicated when we add output attributes to a poly-time solvable boolean CQ. For example, on a boolean CQ $Q:- R_1(C,E), R_2(E,F), R_3(F,H)$, adding a universal attribute $A$ leads to a poly-time solvable query $Q(A):- R_1(A,C,E), R_2(A,E,F), R_3(A,F,H)$, but adding attributes $A,B$ selectively to some of the relations (e.g., $Q(A,B):- R_1(A,C,E), R_2(A,B,E,F), R_3(B,F,H)$) can result in an NP-hard query. So, our goal is to understand how the addition of output attributes changes the complexity of the \ourprob\ problem. For simplicity, the {\em head join} for a CQ $Q$ denotes the residual query after removing all non-output attributes from all relations in $Q$. 
%For simplicity and intuition, w
We start with the class of full CQs, i.e., without non-output attributes. A nice connection between {\em hierarchical join} and our previously defined procedure \isptime\ can be observed. 

\begin{definition}[Hierarchical Join]
	A full CQ $Q$ is hierarchical if for each pair of attributes $A,B \in \attr(Q)$, $\rel(A) \subseteq \rel(B)$, $\rel(B) \subseteq \rel(A)$, or $\rel(A) \cap \rel(B) = \emptyset$, and non-hierarchical otherwise.
\end{definition}

\begin{figure}
	\centering
	\includegraphics[scale=0.8]{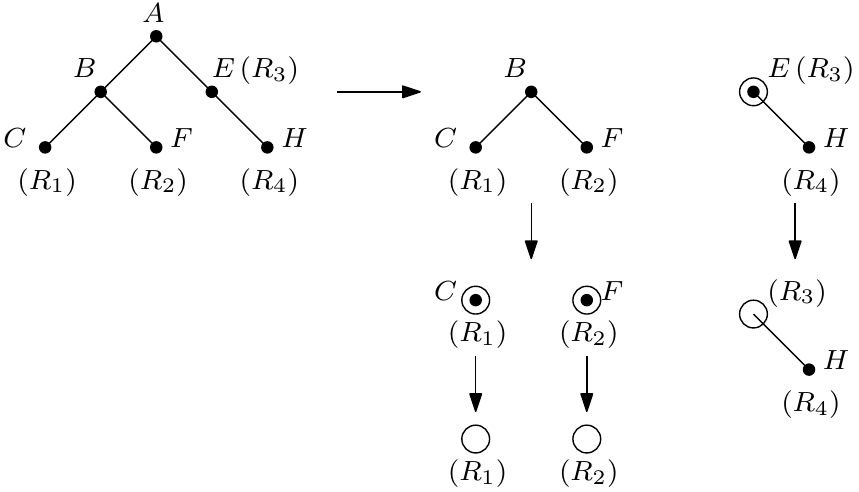}
	\caption{An example of hierarchical join $Q(A,B,C,E,F,H):- R_1(A,B,C), R_2(A,B,F),$ $R_3(A,E),$ $R_4(A,E,H)$, and % The left is the attribute tree and the right is 
	an illustration of applying procedure \isptime\ on it.}
	\label{fig:hierarchical}
\end{figure}

Note that a hierarchical CQ can be organized into a tree structure, where each relation is a root-to-node path. An example is given in Figure~\ref{fig:hierarchical}. Moreover, each relation ends up vacuum by alternately applying the two simplification steps in \isptime\ on this tree. In this way, if $Q$ is hierarchical, $\isptime(Q)$ always returns \true. However, the converse is not necessarily true. For example, $Q(A, B, E):-R_1(A,E), R_2(A,B,E), R_3(B,E), R_4(E)$ is non-hierarchical but $\isptime(Q)$ returns \true\ (after removing the universal attribute $E$, relation $R_4$ becomes vacuum).  We focus on non-hierarchical 
CQs in the rest of this discussion.

The previous result on boolean CQs only considers endogenous relations. Unfortunately, this is insufficient for a full CQ in general; for example, removing the exogenous relation $R_2$ would make $\pathquery(A,B):-R_1(A), R_2(A,B), R_3(B)$ poly-time solvable. So, we need a more fine-grained notion than exogenous/endogenous relations in characterizing the complexity of non-boolean CQs. 

\begin{definition}[Dominated Relation in Full CQs]
	\label{def:dominated-full}
	In a full CQ $Q$, relation $R_j$ is dominated by relation $R_i$ if
	(1) $\attr(R_i) \subseteq \attr(R_j)$; \revm{and} (2) for any relation $R_k$ with $\attr(R_i) - \attr(R_k) \neq \emptyset$, $\attr(R_j) \cap \attr(R_k) \subseteq \attr(R_i)$.
\end{definition}

We say that a relation is {\em dominated} if it is dominated by any other relation, and {\em non-dominated} otherwise. Note that a dominated relation must be exogenous, but all exogenous relations may not be dominated.  A structural dichotomy for full CQs based on dominated relations is given by: %Lemma~\ref{cor:full}. 

\begin{lemma}
    \label{lem:full}
	For a full CQ $Q$, the $\ourprob(Q,D,k)$ problem is NP-hard if and only if the non-dominated relations are non-hierarchical.
\end{lemma}

Note that full CQs do not have any non-output attributes. But, fortunately, the above hardness continues to hold even on adding output attributes. To make this formal, we need to extend the notion of dominated relations to general CQs. 

\begin{definition}[Dominated Relation in CQs]
	\label{def:dominated}
	In a CQ $Q$, relation $R_j$ is dominated by relation $R_i$ if (1) $\attr(R_i) \subseteq \attr(R_j)$; (2) for any relation $R_k$ with $\attr(R_i) - \attr(R_k)  \neq \emptyset$, $\attr(R_j) \cap \attr(R_k) \subseteq \attr(R_i) \cap \head(Q)$; (3) $\attr(R_i) \subseteq \head(Q)$ or $\head(Q) \subseteq \attr(R_i)$.
\end{definition}

If there is more than one relation defined on the same attributes, i.e., $\attr(R_i) = \attr(R_j)$,  then we just consider any one of them as {\em  non-dominated} and the remaining ones as {\em dominated}. We can now use this extended definition to claim our second hard case: if the head join of non-dominated relations is non-hierarchical, then the CQ is NP-hard. \revc{Note that these definitions of ``domination'' are different from~\cite{FreireGIM15}, as we need a more fine-grained characterization of exogenous relations for \ourprob.
Moreover, Lemma~\ref{lem:vacuum} can be easily interpreted as follows: If there is a vacuum relation $R_i$ in a CQ $Q$, then every remaining relation must be dominated by $R_i$, therefore $\ourprob(Q,D,k)$ is poly-time solvable by Theorem~\ref{THM:DICHOTOMY-STRUCTURE}. }

\subsubsection{Strand}

The remaining case is one where on the output attributes, the non-dominated relations are hierarchical {\em and} on the non-output attributes, there is no triad-like structure. These two conditions guarantee poly-time solvability for full and  boolean CQs respectively. But, interestingly, when appearing together in a general CQ, they no longer guarantee poly-time solvability.  For example, the CQ $Q(A,B,C):-$ $R_1(A,B,E), R_2(A,C,E)$ is NP-hard while both $Q(A,B,C):-$ $R_1(A,B), R_2(A,C)$ and $Q():-R_1(E), R_2(E)$ are poly-time solvable. To characterize this class of queries, we introduce our third hard structure that we call a {\em strand}:

\begin{definition}[strand]
	\label{def:pyramid}
	A strand is a pair of non-dominated relations $R_i, R_j \in \rel(Q)$ such that (1) $\head(Q) \cap \attr(R_i) \neq \head(Q) \cap \attr(R_j)$; (2) $(\attr(R_i) \cap \attr(R_j)) - \head(Q) \neq \emptyset$.
\end{definition}

The reason why the strand structure makes the \ourprob\ problem hard can be explained by the procedure \isptime. Consider any CQ with such a strand structure with  $R_i, R_j$. After applying two simplification steps, $R_i, R_j$ will be in the same connected subquery $Q_0$, since attributes in $(\attr(R_i) \cap \attr(R_j)) - \head(Q)$ are not universal and therefore couldn't have been removed by \isptime. Moreover, $Q_0$ is non-boolean, since $\attr(R_i)\cap \head(Q) \ne \attr(R_j)\cap \head(Q)$ and therefore, there is at least one non-universal output attribute. Next, we prove that there is no vacuum relation in $Q_0$. Suppose $R_\ell$ becomes vacuum in $Q_0$. Observe that $\attr(R_\ell) \subseteq \head(Q)$ and $\attr(R_\ell) \subseteq \attr(R_h)$ for every relation $R_h \in \attr(Q_0)$. Since $R_i$ is not dominated by $R_\ell$, there must exist another relation $R_k \in \rel(Q) - \{R_i, R_j\}$ such that $\attr(R_\ell) - \attr(R_k) \neq \emptyset$ and $(\attr(R_i) \cap \attr(R_k)) - \attr(R_\ell) \neq \emptyset$. Note that $R_k$ is not in $Q_0$; otherwise, $\attr(R_\ell) - \attr(R_k) = \emptyset$. In this case, $(\attr(R_i) \cap \attr(R_k)) - \attr(R_\ell) = \emptyset$, coming to a contradiction. Therefore, the $\isptime$ algorithm will go to ``others'', and return \false\ for $Q_0$, as well as for $Q$. This allows us to claim our third hard case: if a strand exists, then CQ is NP-hard.

%%This completes a sketch of the proof of Theorem~\ref{THM:DICHOTOMY-STRUCTURE}. The details of the proof are in the full version~\cite{hu2020ADP}.

\subsection{Sketch of Proof of Theorem~\ref{THM:DICHOTOMY-STRUCTURE}}
\label{SEC:EQUIVALENCE}

So far, we have defined three hard structures for general CQs, any one of which makes the \ourprob\ problem NP-hard. We now sketch the main ideas in the proof of Theorem~\ref{THM:DICHOTOMY-STRUCTURE}; the detailed proof is in Appendix. This proof uses  Theorem~\ref{thm:dichotomy} by mapping each of the NP-hard cases in Theorem~\ref{thm:dichotomy} to the existence of a hard structure as defined by Theorem~\ref{THM:DICHOTOMY-STRUCTURE}, and vice-versa. But, interestingly, this mapping is not one-one in the sense that multiple cases in the procedural dichotomy established by Theorem~\ref{thm:dichotomy} map to same case in the structural dichotomy of Theorem~\ref{THM:DICHOTOMY-STRUCTURE}, and vice-versa. This lends further credence to our assertion that the procedural dichotomy of the previous section is not sufficient by itself to explain the structural reasons behind the NP-hardness or poly-time solvability of the \ourprob\ problem for individual CQs.

We first point out that the two simplification steps in the \isptime\ procedure preserve the existence of hard structures.
\begin{lemma}
\label{LEM:COMMON-HARD-STRUCTURE}
Let $A$ be a universal attribute in $Q$. Then, there is a hard structure in $Q$ if and only if there is a hard structure in $Q_{-A}$. 
\end{lemma}

\begin{lemma}
\label{LEM:DECOMPOSE-HARD-STRUCTURE}
Let $Q_1, Q_2, \cdots, Q_s$ be the connected subqueries of $Q$. Then, there is a hard structure in $Q$ if and only if there is a hard structure in $Q_i$ for some $i \in \{1,2,\cdots,s\}$. 
\end{lemma}

When neither of the simplification steps can be applied, \isptime$(Q)$\ ends up with three cases. If there is a vacuum relation in $Q$, say $R_i$, $\isptime(Q)$ returns \true. In this case, $Q$ does not contain any hard structure as $R_i$ is the only endogenous and non-dominated relation. If $Q$ is boolean, $\isptime(Q)$ returns \false\ if and only if it contains a triad. Then, we are left with the case when $\isptime(Q)$ goes into the ``Others'' bucket. Each core query shown in Section~\ref{sec:hardest} contains hard structure; more specifically, the head join of non-dominated relations in $\pathquery$ is non-hierarchical, and both $\swingquery$ and $\seesawquery$ contain a strand. In general, we can show the existence of hard structures for $Q$ falling into one of the three cases in Figure~\ref{fig:others}. The correspondence between different cases of the procedural and structural characterizations are shown in Figure~\ref{fig:mapping-hard-structure}.

\begin{figure}
    \centering
    \includegraphics[scale=0.8]{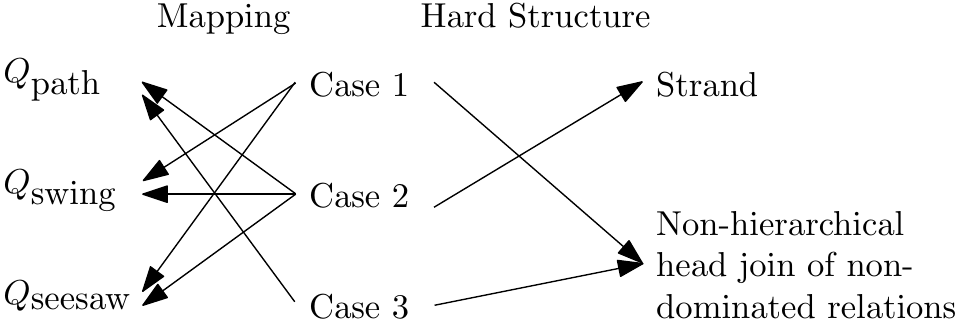}
    \caption{\revm{Correspondence between the three cases of CQs on which \isptime\ falling into ``other'' bucket in Figure~\ref{fig:others}, the core query it maps to (the left) and the hard structure it contains (the right).}}
    \label{fig:mapping-hard-structure}
\end{figure}

%% file: sections/approximation.tex
\revm{\section{Approximations}
\label{SEC:APPROX}
In this section, we discuss approximations for %optimal solutions to 
the $\ourprob(Q, D,k)$ problem when it is NP-hard.

\subsection{Full CQs}\label{sec:approx-full-CQ}

We first consider full CQs, on which \ourprob\ problem can be related to the \emph{Partial Set Cover problem} (PSC).

\begin{definition}
\label{def:kPSC}
    Given a set of elements $\mathcal{U}$, a family of subsets $\mathcal{S}\subseteq 2^U$, %a cost function on sets $c:\mathcal{S}\rightarrow \mathcal{Q}^+$ 
    and a positive integer $k'$, the goal of the Partial Set Cover problem is to pick a minimum collection of sets from $\mathcal{S}$ that covers at least $k'$ elements in $\mathcal{U}$.  %pick a minimum cost collection of sets from $\mathcal{S}$ that covers at least $k'$ elements in $\mathcal{U}$.
\end{definition}

Observe that $\ourprob(Q, D, k)$, where the goal is to pick the smallest number of input tuples that intervene on at least $k$ output tuples, can be modeled as a PSC problem as follows. Sets correspond to input tuples from relations in the body of $Q$ and elements to output tuples in $Q(D)$. The set corresponding to an input tuple comprises all elements corresponding to output tuples that are deleted on the deletion of the input tuple. Also, $k'=k$. %Additionally, if there is a cost associated with deleting a specific input tuple, the cost function $c$ in the partial set cover instance can be used to reflect this. 
Additionally, if there are $p$ relations in $Q$, then every element belongs to at most $p$ sets.
It is known that the PSC problem admits greedy and primal-dual algorithms with approximation factors of $O(\log k)$ and $p$ respectively~\cite{GKS04}. Hence, we get the same results for the $\ourprob$ problem.

\begin{theorem}\label{THM:FULL-APPROXIMATE}
    For a full CQ $Q$ with $p$ relations, any instance $D$ and integer $k$, $\ourprob(Q, k, D)$ admits $O(\log k)$ and $p$-approximations.
\end{theorem}
\begin{proof}
	We prove that the reduction preserves the approximation guarantee in two steps: 1) given an instance of $\ourprob(Q, k, D)$, how to construct an instance of $k'$-PSC, and 2) given a solution to $k'$-PSC, how to recover a solution to $\ourprob(Q, k, D)$.
	
	Given the full CQ $Q$ containing $p$ relations in its body, namely $R_1, R_2, \cdots, R_p$, we create a set per input tuple in the $p$ relations, and an element per output tuple in $Q(D)$. Each set contains elements that correspond to the output tuples resulting from the join between the associated input tuple and tuples from other relations in $Q$. It is well-known that the natural join on $R_1, R_2, \cdots, R_p$ can be computed in poly-time. Moreover, exactly one tuple in each of the $p$ relations participates in the join operation that produces a particular output tuple. Therefore, each element in the $k'$-PSC instance belongs to exactly $p$ sets. As a result, the size of the $k'$-PSC instance that we create is polynomial in the data complexity of $\ourprob(Q, k, D)$. Moreover, there is a one-on-one correspondence between instances of the two problems.
	
	Lastly, given a $p$-approximate solution to $k'$-PSC, we recover a solution to $\ourprob(Q, D, k)$ by picking the tuples associated with the sets in the solution, say $I$. Observe that the sets in $I$ cover $k' = k$ elements in $U$. Thus, removing the corresponding input tuples from $\ourprob(Q, D, k)$ will intervene on at least $k$ output tuples.

	Note that this implies that if the query has constant size, i.e., $p$ is a constant, full CQs admit a constant-factor approximation for the \ourprob\ problem.
\end{proof}

This implies that if the query has constant size, i.e., $p$ is a constant, full CQs admit a constant-factor approximation for the \ourprob\ problem.

\subsection{Inapproximability of General CQs}

The situation, however, is quite different for general CQs. We first observe that obtaining even sub-polynomial approximations for the \ourprob\ problem in general is unlikely. In particular, on $\swingquery(A): R_2(A,B), \\R_3(B)$, which is the core hard query in Section~\ref{sec:hardest}, we show the following hardness:

\begin{lemma}
\label{lem:swing-hard-approximate}
Under some mild cryptographic assumptions, the $\ourprob(\swingquery, D, k)$ problem with $|D| = n$ is hard to approximate within $\Omega(n^{\epsilon})$ factor for some constant $\epsilon>0$.
\end{lemma}

Recall that we established NP-hardness of $\ourprob(\swingquery, D, k)$ via a reduction from the {\em k-minimum coverage} (\kmc) problem.
%: given a universe $\mathcal{U}$, a family $\mathcal{S}$ of subsets of $\mathcal{U}$ and an integer $k$, it asks to find $k$ subsets from $\mathcal{S}$ such that the size of their union is minimized. 
As shown in Appendix~\ref{APPENDIX:HARD-CORE},
his reduction is also approximation-preserving, which implies the above lemma via known hardness results for the \kmc\ problem~\cite{applebaum2013pseudorandom,chlamtavc2017minimizing, chlamtac2018densest}.
While this rules out the possibility of approximation algorithms in general for the \ourprob\ problem, there are several query classes on which we had shown NP-hardness of the problem but their approximability is still open. This includes simple CQs such as $\seesawquery(A):R_1(A), R_2(A,B), R_3(B)$. We leave the precise classification of query classes according to approximability of the \ourprob\ problem as an interesting direction for future work.
}

%% file: sections/algorithms.tex
\section{Algorithms and Optimizations}
\label{SEC:ALGORITHMS}
%\cut{
%We now design a practical, optimal algorithm for CQs on which the \ourprob\ problem is poly-time solvable. Using this algorithm, we can also support a large class of {\em CQs with selection}. Finally, for NP-hard CQs, we give a heuristic algorithm that computes a feasible (but possibly suboptimal) solution. In the next section, we perform experiments to evaluate the performance of these algorithms in practice.
%}

%\subsection{A practical algorithm}

\reva{The framework of our poly-time algorithm, which returns the exact solution for ``easy'' queries and a heuristic for hard queries, is described as \computeADP\ in Algorithm~\ref{algo:computeadp}. It builds upon the algorithm for the Resilience problem \cite{FreireGIM15}, which is a special case of the \ourprob\ problem. Our algorithm recursively calls itself through {\sc Universal} and {\sc Decompose} procedures. For poly-time solvable CQs, it only uses the first four cases: this follows the proof of Theorem~\ref{thm:dichotomy} by applying the two simplifications repeatedly until it becomes a boolean query or contains a vacuum relation. Our first optimization is to include a new base case that we call {\em singleton}. If the conditions of this case (we describe them below) are satisfied, then a simple algorithm {\sc Singleton} is directly applied instead of continuing to apply the two simplification steps. In addition to computing the optimal solution for poly-time solvable CQs, Algorithm~\ref{algo:computeadp} also generates a feasible solution for NP-hard CQs. In this case, it alternately applies these two simplification steps until it becomes boolean or goes to the ``others'' category in Figure~\ref{fig:isptime}. We eventually invoke an approximate procedure {\sc GreedyForCQ} on the non-boolean CQ when neither simplification step can be applied any more. Our second optimization is a smarter way of solving the recurrent formula for these two simplification steps, as shown in
\universe$(Q, D, k)$ and \decompose$(Q, D, k)$. Note that the simplification steps involve large dynamic programs; so, this optimization provides significant scalability in practice. Both poly-time solvable and NP-hard queries benefit from the improvement of two simplification steps.}

\begin{algorithm}[t]
	\caption{$\textrm{\sc ComputeADP}(Q, D, k)$}
	\label{algo:computeadp}
	\textbf{If} $Q$ is Boolean
	\Return \medskip \noindent {\sc Boolean}$(Q,D,k)$\;
	
	\textbf{ElseIf} $Q$ is a singleton 
	\Return \medskip \noindent {\sc Singleton}$(Q,D,k)$\;
		
	\textbf{ElseIf} $Q$ has universal attribute \textbf{then} \medskip \noindent \universe$(Q,D,k)$\;
	
	\textbf{ElseIf} $Q$ is disconnected
    \textbf{then} \medskip \noindent \decompose$(Q,D,k)$\;
	
	\textbf{Else} \Return \medskip \noindent {\sc GreedyForCQ}$(Q,D,k)$\;
\end{algorithm}
In the recursion tree of \computeADP, each leaf node (\bool,  \singleton\ and {\sc GreedyForCQ}) can be computed in poly-time and internal node (\universe\ and \decompose) can be built upon its children in poly-time.  Also, there are $O(1)$ nodes in this recursion tree, since the query size (in terms of number of attributes and relations) is constant and each recursive call decreases the query by at least one relation or attribute. Hence, we get an poly-time algorithm overall.

\subsection{Boolean}

In~\cite{FreireGIM15}, a poly-time algorithm was proposed for boolean CQs without a triad structure. A boolean query is {\em linear} if its relations may be arranged in linear order such that each attribute occurs in a contiguous sequence of atoms. It is proved that every boolean query without a triad structure can be transformed into a query of equivalent complexity that is linear. Thus, we only provide the algorithm for computing the \ourprob\ problem on an arbitrary linear query.

\medskip \noindent {\bf \bool}$(Q,D,k)$. We first label relations in linear ordering $R_1, R_2,\cdots, R_p$ and then build a network construct a network $G$ as follows. Note that $G$ is an $(p+1)$-partite graph consists of vertices $V = \{x\} \cup V_1 \cup V_2 \cup \cdots \cup V_{p-1} \cup \{y\}$, where $V_i = \attr(R_i) \cap \attr(R_{i+1})$. There is an edge $e = (u,v)$ for $u \in V_i, v \in V_{i+1}$ if there exists a tuple $t \in R_{i+1}$ with $\pi_{V_i} t = u$ and $\pi_{V_{i+1}} t = v$. Moreover, there is an edge between every vertex in $V_1$ and $x$, and every vertex in $V_{p-1}$ and $y$. Each edge has weight $1$.

A minimum cut of $G$ is exactly the solution for $\ourprob(Q, D, k)$, which can be computed using the standard Edmonds–Karp algorithm with time complexity $O(|D|^3)$.

\subsection{Singleton}
We first lay out the conditions of this new base case for a poly-time solvable CQ:
\begin{definition}[Singleton]
	\label{def:singleton}
	A CQ $Q$ is singleton, if there exists a relation $R_i \in \rel(Q)$ such that \revc{(1) $\attr(R_i) \subseteq \attr(R_j)$ holds for every other relation $R_j \in \rel(Q)$;} and (2) either $\attr(R_i) \subseteq \head(Q)$ or $\head(Q) \subseteq \attr(R_i)$.
\end{definition}

Note that the execution of \isptime\ can also be modeled as {\em recursion tree}, where each leaf node is either a Boolean query or contains vacuum relation, and each internal node corresponds to one simplification step. %(1) $Q$ produces one child node $Q_{-A}$ if an universal attribute $A$ is removed; (2)  $Q$ is a leaf node if $Q$ is Boolean or contains a vacuum relation; (3) $Q$ produces $s$ children nodes $Q_1, Q_2, \cdots , Q_s$ if $Q$ is disconnected with $s$ connected subqueries; and (4) $Q$ is a leaf node otherwise. For a CQ $Q$ on which $\isptime(Q)$ returns \true, it has a recursion tree such that each leaf node is either a Boolean query without triad, or has a vacuum relation. 
On this recursion tree, we point out an important property for singleton structure, as stated in Lemma~\ref{LEM:SINGLETON}. %state below a useful property of singleton queries (using $\isptime$, it is easy to verify that all singleton queries are poly-time solvable) -- the proof is in tlhe full version~\cite{hu2020ADP}.

\begin{lemma}
	\label{LEM:SINGLETON}
	For a CQ $Q$ on which $\isptime(Q)$ returns \true, each leaf (not root) node containing a vacuum relation must have an ancestor that is a singleton query.
\end{lemma}
\begin{proof}
	Note that each node $v$ in the recursive tree is associated with a query $Q_v$.
	Let $v$ be a leaf node in the recursion tree containing a vacuum relation $R_i$. Let $u$ be the parent node of $v$. Observe that $u$ doesn't contain a vacuum relation; otherwise, $u$ itself is a leaf. If $u$ generates $v$ by decomposing a disconnected, then $R_i$ is also a vacuum relation in $Q_u$, coming to a contradiction. If $u$ generates $v$ by removing an universal attribute $A$, $\attr(R_i) = \{A\}$ in query $Q_v$. As $A$ is an universal attribute in $Q_u$, $Q_u$ is a singleton by Definition~\ref{def:singleton}. 
\end{proof}

So, it suffices to replace the vacuum relation base case with the singleton.  %, and this is necessarily more efficient, which will be shown by our experiment later.
%The algorithm for computing the \ourprob\ problem for a singleton query is quite straightforward -- t
%The detailed proof, algorithm, and pseudocode are given in~\cite{hu2020ADP}.

\begin{algorithm}[h]
	\caption{\singleton$(Q, k, D)$}
	\label{algo:singleton}
	$R_i \gets \arg \min_{R_j \in \rel(Q)} |\attr(R_j)|$\;
	\If{$\attr(R_i) \subseteq \head(Q)$}{
		\ForEach{tuple $t \in R_i$}{
			$p_t \gets \pi_{\attr(R_i) = t} Q(D)$\; 
		}
		Sort all $p_t$'s in decreasing order as $p_1, p_2, \cdots, p_m$\;
		Find index $i$ such that $\sum_{j=1}^{i-1} p_j < k \le \sum_{j = 1}^i p_j$\;
		\Return $i$\;
	}
	\Else{
		Remove all dangling tuples in $R_i$\;
		\ForEach{$t \in Q(D)$}{
			$c_t \gets |\pi_{\head(Q)= t} R_i|$\;
		}
		Sort all $c_t$'s in increasing order as $c_1,c_2, \cdots, c_m$\;
		\Return $\sum_{j =1}^k c_j$\;
	}
	
\end{algorithm}

\medskip \noindent {\bf \singleton $(Q, D, k)$}. Let $R_i$ be the relation with 
the minimum number of attributes. By definition, either $\head(Q) \subseteq \attr(R_i)$ or $\attr(R_i) \subseteq \head(Q)$. 

Case 1: $\attr(R_i) \subseteq \head(Q)$.  We compute the number of output tuples that inherent attribute values from a tuple $t \in R_i$ and call it the ``profit'' of $t$, denoted as $p_t$. Then, we sort the tuples by their profits and choose greedily in decreasing order until their sum exceeds $k$. These chosen tuples form an optimal solution.

Case 2: $\head(Q) \subseteq \attr(R_i)$. We first remove all {\em dangling tuples}\footnote{A tuple is {\em dangling} if it doesn't participate in any full join result, and {\em non-dangling} otherwise. For $R_j \in \rel(Q)$, its non-dangling tuples can be obtained by projecting full join results on $\attr(R_i)$. This can be done in poly-time.} in $R_i$, i.e., those don't participate in the full join result of the body of $Q$. Then we count for each output tuple $t \in Q(D)$, the number of tuples in $R_i$ whose projection on attributes $\head(Q)$ is equivalent to $t$, and call it the ``cost'' of $t$, denoted by $c_t$. Finally, we sort the output tuples by cost and choose in increasing order the first $k$ tuples. The optimal solution is now obtained as the set of input tuples in $R_i$ whose removal deletes $k$ output tuples.

This algorithm takes $O(|D|^{|Q|})$ time since computing full join results dominates the complexity.

\subsection{Universe and Decompose}
We show some optimization for \decompose\ and \universe\ procedures respectively.

\begin{algorithm}[h]
	\caption{\universe$(Q, D, k)$}
	\label{algo:universal}
	
	$A \gets \head(Q) \cap \left(\bigcap_{R \in \rel(Q)} \attr(R)\right)$\;
	Label all possible combinations over $A$ as $\{a_1, a_2, \cdots, a_g\}$\; 
	\ForEach{$i \in \{1,2,\cdots, g\}$}{
		$D_i \gets \{\sigma_{\pi_A t = a_i} R_i: \forall R_i \in \rel(Q)\}$\;
	}
	\ForEach{$j \in \{1,2,\cdots,k\}$}{
		$\optcost[1][j] \gets \computeADP(Q, D_1, j)$\;
	}
	\ForEach{$i \in \{2, \cdots, g\}$}{			 
		\ForEach{$j \in \{1,2,\cdots,k\}$}{
			$\optcost[i][j] \gets \optcost[i-1][j]$\;
			\For{$m = 1$ to $j-1$}{
				$c_{i, m} \gets \computeADP(Q, D_i, m)$\;
				\If{$\optcost[i][j] > \optcost[i-1][j - m] + c_{i, m}$}{							 
					$\optcost[i][j] \gets \optcost[i-1][j - m] + c_{i, m}$\;
				}
			}
		}
	}
	\Return $\optcost[g][k]$\;
\end{algorithm}

\begin{algorithm}[h]
	\caption{{\decompose}$(Q,D,k)$}
	\label{algo:decompose}
	
	Let $Q_1, Q_2, \cdots, Q_s$ be the connected subquery of $Q$\;
	$Q_\alpha \gets Q_1$\;
	\ForEach{$j \in \{1,2,\cdots,k\}$}{
		$\optcost[1][j] \gets \computeADP(Q, D_1, j)$\;
	}
	\ForEach{$i \in \{2,3,\cdots, s\}$}{
		$m_1 \gets \prod_{\ell=1}^{i-1}|Q_\ell(D)|$, $m_2 \gets |Q_i(D)|$\;
		\ForEach{$j \in \{1,2,\cdots, k\}$}{
			$\optcost[i][j] \gets +\infty$\;
			\ForEach{$(k_1, k_2) \in \{0,1,\cdots,j\} \times \{0,1,\cdots, j\}$}{
				\If{$k_1m_2 + k_2m_1 - k_1k_2 \geq j$}{
					$c_{i, k_2} \gets \computeADP(Q_i, D, k_2)$\;
					\If{$\optcost[i][j] >  \optcost[i-1][k_1] + c_{i, k_2}$}{
						$\optcost[i][j] \gets \optcost[i-1][k_1] + c_{i, k_2}$\;
					}
				}
			}
		}
		$Q_{\alpha} \gets Q_{\alpha} \times Q_i$\;
	}
	\Return $\optcost[s][k]$\; 
\end{algorithm}

\medskip \noindent {\bf \decompose $(Q, D, k)$}. %The proof of Lemma~\ref{lem:decompose} already implies a recursive algorithm for handling a disconnected query. 
Assume $Q$ has $s$ connected subqueries, $Q_1, Q_2, \cdots, Q_s$. The divide-and-conquer strategy will first compute $\ourprob(Q_i, D, k_i)$ for each subquery $Q_i$ over $k_i$, and then find an optimal combination of $k_1, k_2, \cdots, k_s$ by enumeration over $\Theta(k^s)$ solutions, which becomes expensive for large $s$. We give an optimized algorithm.

Let $\optcost[i][j]$ denote the minimum number of input tuples to remove at least $j$ output tuples from subquery $\times_{j=1}^i Q_j(D)$. $\optcost[i][j]$ can be computed using the following dynamic program: 
\begin{equation*}
%\label{eq:decompose-2}
\optcost[i][j] =  \min_{k_1, k_2 \in K(i,j)}  \optcost[i-1][k_1] + \computeADP(Q_i, D, k_2)
\end{equation*}
where $K(i,j) = \{k_1, k_2: k_1 |Q_i(D)| + k_2  \prod_{\ell=1}^{i-1} |Q_\ell(D)| - k_1k_2 \geq j, k_1, k_2 \in \mathbb{Z}^+\}$ and Algorithm~\ref{algo:computeadp} is invoked for solving \ourprob$(Q_i, D, k_2)$. To remove at least $j$ output tuples from $\times_{j=1}^i Q_j(D)$, we remove $k_1$ output tuples from first $i-1$ queries and $k_2$ output tuples from $Q_i(D)$, the total number of results removed is $k_1 |Q_i(D)| + k_2 \prod_{\ell=1}^{i-1} |Q_\ell(D)| - k_1k_2$ since results across subqueries are joined by Cartesian product. Thus, after recursively computing the solution to $\ourprob(Q_i, D, k_2)$ for each subquery $Q_i$ over all values of $k_2$, the recurrence formula can be solved in $O(s \cdot k^3) = O(|Q| \cdot k^3)$ time since there are $O(sk)$ cells in the two-dimensional data structure $\optcost[i][j]$ and each can be computed in $O(k^2)$ time.

\medskip \noindent{\bf \universe $(Q, D, k)$}. %The proof of Lemma~\ref{lem:common} also implies a recursive algorithm for handling CQs containing universal attributes. 
Let $A$ be an universal attribute in $Q$. The input instance $D$ is partitioned into $D_1, D_2, \cdots, D_g$ corresponding to possible combinations of values $a_1, a_2, \cdots, a_g$ over $A$. In $D_i$, each tuple $t$ has $\pi_A t= a_i$. Note that the query result $Q(D)$ is a disjoint union of the subquery results $Q(D_1), Q(D_2),\cdots, Q(D_i)$. 

Let $\optcost[i][s]$ denote the minimum number of input tuples that have to be removed in order to remove at least $s$ output tuples from $Q(D)$, under the constraint that the input tuples can only be chosen from $D_1$ to $D_i$. Using this notation, we can now write the following dynamic program: 
\begin{equation*}
%\label{eq:universal}
\optcost[i][s] = \min_{m = 0}^{s}\Big \{ \optcost[i-1][s - m] + \computeADP(Q, D_i, m) \Big\}.
\end{equation*} %Moreover, this problem has an optimal sub-structure as shown in (\ref{eq:universal}) by computing $O(g k)$ intermediate problems in form of $\ourprob(Q, \bigcup_{\ell=1}^i D_\ell, j)$ over $1 \le i\le g$ and $1 \le j \le k$.  Our original problem is the case with $i = g$ and $j = k$. 
%The recurrence formula of (\ref{eq:universal}) can be rewritten as $\optcost[i][s] =$
%\begin{equation*}
%		\min_{m = 0}^{s} \big \{ \optcost[i-1][s - m] + \computeADP(Q, D_i, m) \big \}.
%	\end{equation*}
where Algorithm~\ref{algo:computeadp} is revoked for solving the \ourprob$(Q, D_i, m)$ over $1 \le i \le g$ and $0\le m \le s$. %Here, $m$ denotes the number of output tuples being removed from the subproblem on $D_i$ and \computeADP$(Q, D_i, m)$ finds the minimum number of input tuples that would remove at least $m$ output tuples from $Q(D_i)$. 

When there are more than one universal attributes, they should be removed as one ``combined'' attribute, instead of one by one. Let $A_1, A_2, \cdots, A_h$ be the universal attributes in $Q$. Assume all subproblems $\ourprob(Q, D_i, j)$ over $1 \le i\le g$ and $1 \le j \le k$ have been computed. Then, removing $A_1$, $A_2$, $\cdots$, $A_h$ one by one takes $O(k \cdot |\pi_{A_1, A_2, \cdots, A_h} Q(D)|)$ time while removing them as whole (say in index ordering) takes $O(k \cdot \sum_{\ell=1}^h|\pi_{A_1,\cdots, A_\ell} Q(D)|)$ time. 
Our experiments show this difference in practice.

\subsection{Greedy Heuristics}

	Clearly, we cannot hope for a poly-time algorithm on NP-hard CQs for all input instances $D$ and integers $k$. We provide the following greedy heuristics for computing a feasible solution to $\ourprob(Q,D,k)$ when it is NP-hard. 
	
	\medskip \noindent{\bf GreedyForCQ}$(Q, D, k)$: 
	For many simple queries, the \ourprob\ problem is NP-hard, and is even hard to approximate implied by the results in Section~\ref{SEC:APPROX}. The prime-dual approximation algorithm~\cite{GKS04} for full CQs mentioned in Section~\ref{sec:approx-full-CQ} is not scalable since the size of linear programming would become very large, and  not applicable to CQs with projections. So, 
	%in Algorithm~\ref{algo:greedy}, 
	we give a greedy heuristic for handling all NP-hard CQs when neither simplification steps can be applied.%(pseudocode is in \cite{hu2020ADP}). 
	It greedily chooses a tuple which removes the maximum number of output tuples among the remaining ones in every step (like the approximation algorithm for the set cover problem). Moreover, we can narrow our scope to tuples in endogenous relations in the greedy algorithm. Note that {\sc GreedyForCQ} achieves $O(\log k)$-approximation for full CQs, but there is no theoretical guarantees on the approximation ratio when projection exists. %Our experiments demonstrate its quality by comparing to a baseline algorithm.

\begin{algorithm}[h]
	\caption{$\textrm{\sc GreedyForCQ}(Q, D, k)$}
	\label{algo:greedy}
	$S \gets \emptyset$\;
	\While{$k > 0$}{
		$t' \gets null$, $p(t') \gets 0$\;
		\ForEach{tuple $t$ from an endogenous relation}{
			$p(t) = |Q(D -S )| - |Q(D- S - t)|$\;
			\If{$p(t) \ge p(t')$}{
				$t' \gets t$, $p(t') \gets p(t)$\;
			}
		}
		$S \gets S \cup \{t'\}$, $k \gets k - p(t')$\;
	}
	\Return $S$\;
\end{algorithm}

	\medskip \noindent{\bf DrasticGreedyForFullCQ}$(Q, D, k)$: In the heuristic above,
	however, computing the ``profit'' for all input tuples from endogenous relations after every one input tuple is removed is expensive in practice. 
	For full CQs, we propose a more `drastic' greedy solution where we remove input tuples only from one endogenous relation (goes over all endogenous relations and picks the one giving smallest cost%, pseudocode in \cite{hu2020ADP}
	). This  significantly improves the efficiency in our experiments, since the profits are computed for all input tuples only once (since different tuples in the same relation remove disjoint full join results), but theoretically the approximation ratio is no longer guaranteed. Moreover, this strategy fails on CQs with projection. The reason is that input tuples from the same relation do not necessarily remove distinct query results, thus adding their individual profits together is not equivalent to the profit of their union.

\begin{algorithm}[h]
	\caption{$\textrm{\sc DrasticGreedyForFullCQ}(Q, D, k)$}
	\label{algo:greedy-full}
	$S \gets \emptyset$\;
	\ForEach{endogenous relation $R(e)$}{
		\ForEach{$t \in R(e)$}{
			$p(t) = |Q(D)| - |Q(D -t)|$\;
		}
		Sort $R(e)$ by $p(t)$ decreasingly, as $t_1, t_2, \cdots, t_{|R(e)|}$\;
		Find the smallest $i$ such that $\sum_{j=1}^i p(t_j) \ge k$\;
		\If{$i \le |S|$}{
			$S \gets \{t_j \in R(e): j \le i\}$\;
		}
	}
	\Return $S$\;
\end{algorithm}

\subsection{Supporting Selection Operator}\label{sec:selection}
So far, we focused on the class of CQs only with {\em project} and {\em join} operators. In fact, our algorithm also supports a larger class of CQs involving selection operator (\revc{when the domain of some of the attributes is restricted to be constant}). The class of {\em conjunctive queries with selections} can be described as
\[Q(\bm{A}): -\sigma_{\theta_1} R_1(\mathbb{A}_1), \sigma_{\theta_2} R_2(\mathbb{A}_2), \cdots, \sigma_{\theta_p} R_p(\mathbb{A}_p)\]
where $\theta_i$ is a set of predicates each in form of $A = a$ for some attribute $A \in \mathbb{A}$ and value $a$. The result of $\sigma_{\theta_i} R_i(\mathbb{A}_i)$ is the set of tuples in $R_i$ satisfying all predicates in $\theta_i$. Note that we do not have any selection in the head, since any selection in the head can be pushed down to relations in the query body. An attribute is {\em selected} if it appears in any selection; and {\em unselected} otherwise. Let $\mathbb{A}_\theta \subseteq \mathbb{A}$ be the set of {\em selected attributes} in $Q$. Here, we also don't include any self-joins, i.e., each $R_i$ in $Q$ is distinct.

Interestingly, for the \ourprob\ problem, the polynomial solvability of a CQ with selections is equivalent to that of the residual query on the unselected attributes. This is formally stated in Lemma~\ref{LEM:SELECTION}.%(proved in the full version~\cite{hu2020ADP}).

\begin{lemma}
	\label{LEM:SELECTION}
	For a CQ $Q$ with selection predicates $\theta$, the $\ourprob(Q,D,k)$ is NP-hard if and only if $\ourprob(Q_{-{\mathbb{A}_{\theta}}}, D, k)$ is NP-hard, where $Q_{-{\mathbb{A}_{\theta}}}$ is the residual query after removing selected attributes $\mathbb{A}_{\theta}$ from $Q$.
\end{lemma}

\begin{proof}
	We will first show that if $\ourprob(Q_{-{\mathbb{A}_{\theta}}}, D, k)$ is NP-hard, then $\ourprob(Q,D,k)$ is also NP-hard. For an arbitrary instance $D_\theta$ for $Q_{-{\mathbb{A}_{\theta}}}$, we construct another instance $D$ for $Q$ by setting a single value $*$ in the domain of every attribute $A \in \mathbb{A}_{\theta}$ and the related predicate as $A = *$. It can be easily checked that any solution for $\ourprob(Q,D,k)$ with selections is also a solution for $\ourprob(Q_{-{\mathbb{A}_{\theta}}}, D, k)$. If there is an poly-time algorithm for $\ourprob(Q,D,k)$, $\ourprob(Q_{-{\mathbb{A}_{\theta}}}, D, k)$ is also poly-time solvable, coming to a contradiction. Thus, the problem $\ourprob(Q,D,k)$ is NP-hard.
	
	Next we show that if there is a poly-time algorithm $\mathcal{A}$ for  $\ourprob(Q_{-{\mathbb{A}_{\theta}}}, D, k)$ over all instances $D$ and integer $k$, then there is also an poly-time algorithm $\mathcal{A}_\theta$ for $\ourprob(Q, D, k)$. Consider an arbitrary instance $D$ for query $Q$. Let $D'$ be the residual instance of applying predicates to $D$. Observe that the solution for $\ourprob(Q, D', k)$ is exactly that for $\ourprob(Q, D, k)$ since tuples in $D$ violating any predicate will not be removed. Moreover, tuples in $D'$ have the same value on every attribute $A \in \mathbb{A}_{\theta}$. Let $D''$ be the instance of removing attributes $\mathbb{A}_{\theta}$ from $D'$. The solution for $\ourprob(Q_{-{\mathbb{A}_{\theta}}}, D'', k)$ is also the solution for $\ourprob(Q, D', k)$, and can be computed in poly-time. Thus, $\ourprob(Q, D, k)$ is also poly-time solvable for any instance $D$ and integer $k$.
\end{proof}

%% file: sections/experiment.tex
\begin{figure*}
\begin{center}
	\minipage{0.31\textwidth}
	\includegraphics[width=\linewidth]{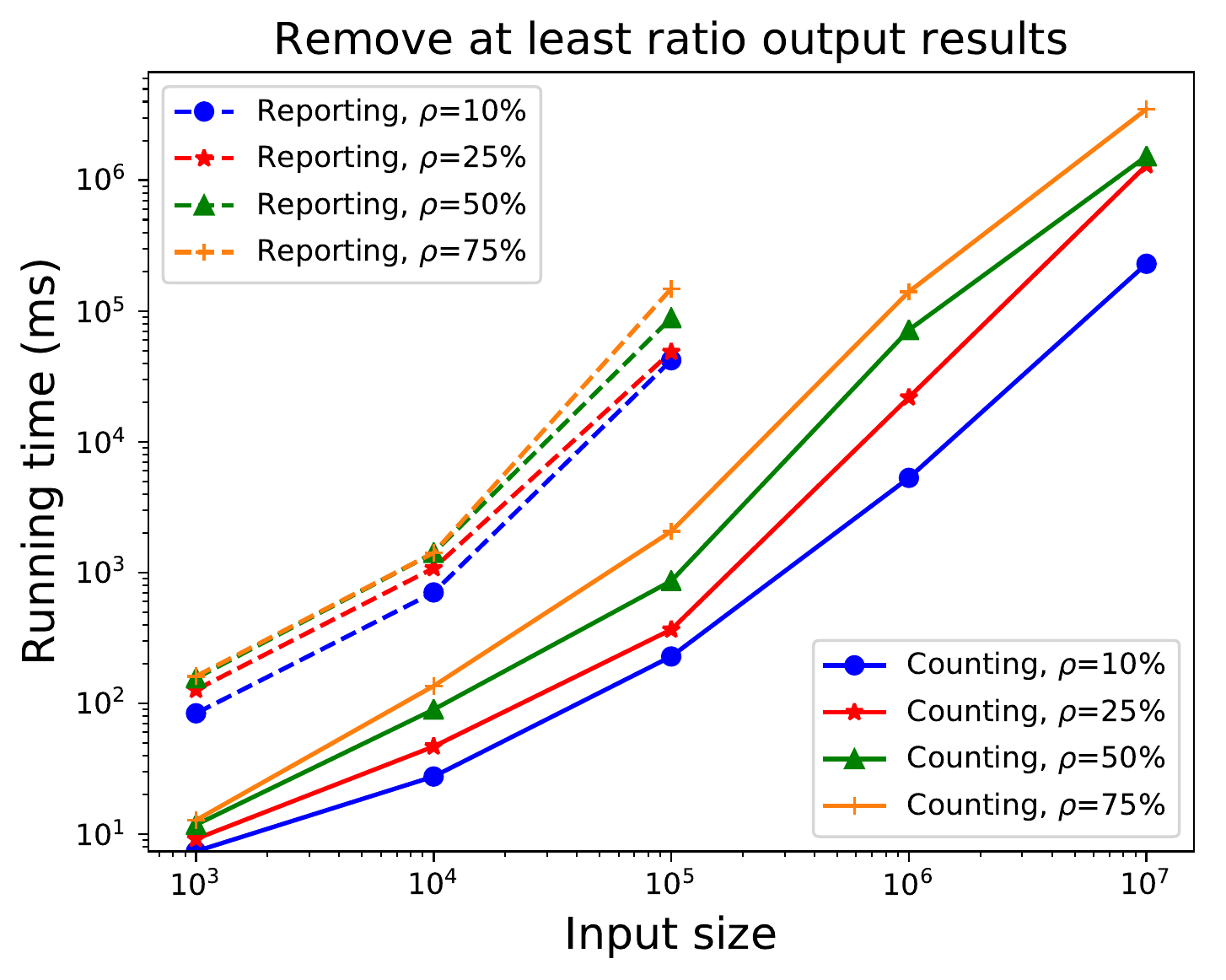}
	\caption{ \revm{Running Time: %counting/reporting 
	$\sigma_{\theta} Q_1$ (easy) exactly (count/report).}}
	\label{fig:Q1Selection}
	\endminipage \hfill
	\minipage{0.31\textwidth}
	\includegraphics[width=\linewidth]{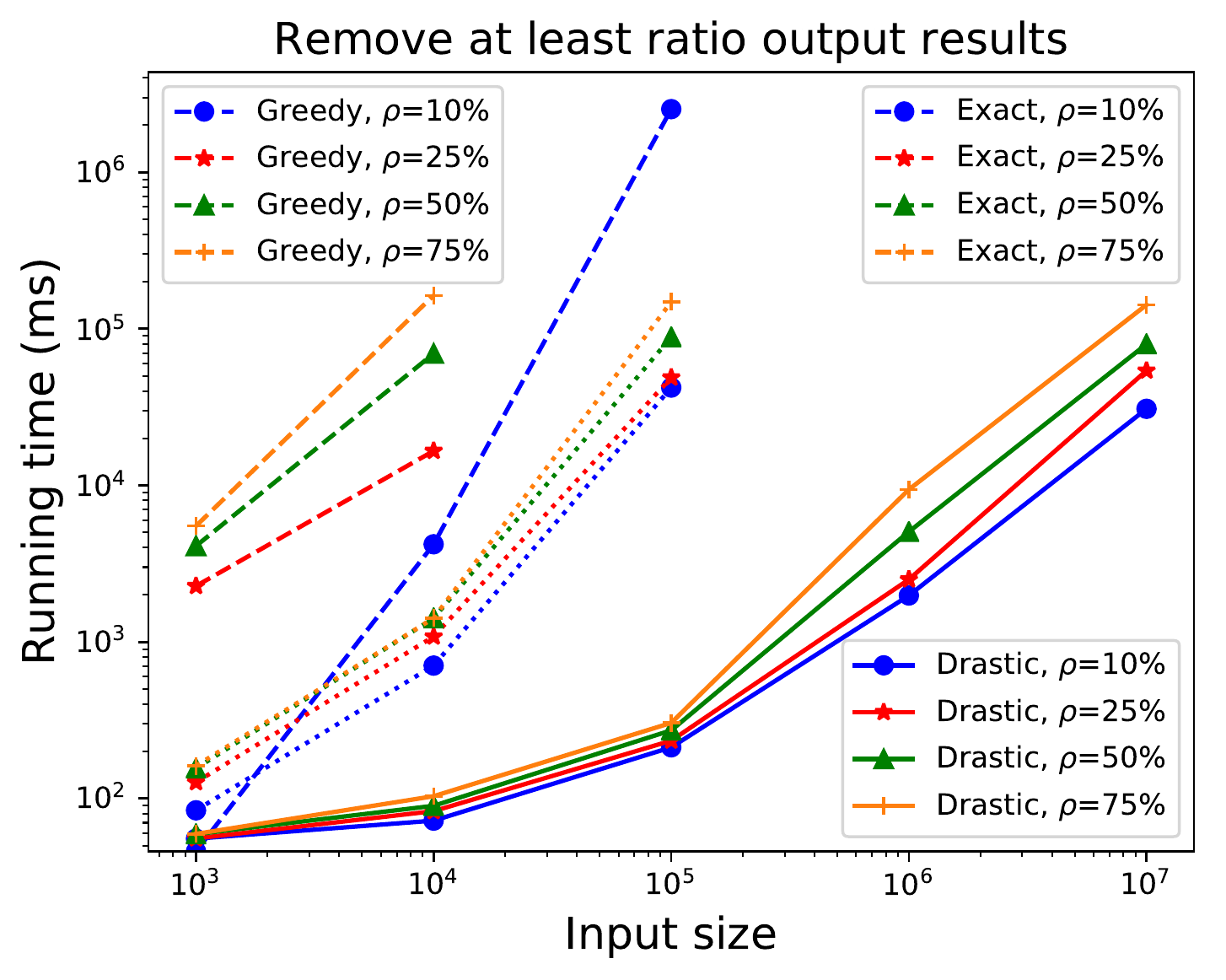}
	\caption{ \revm{Running Time: reporting $\sigma_{\theta}Q_1$ (easy) by heuristics. }}
	\label{fig:heuristicsEasyTime}
	\endminipage \hfill
	\minipage{0.31\textwidth}
	\includegraphics[width=\linewidth]{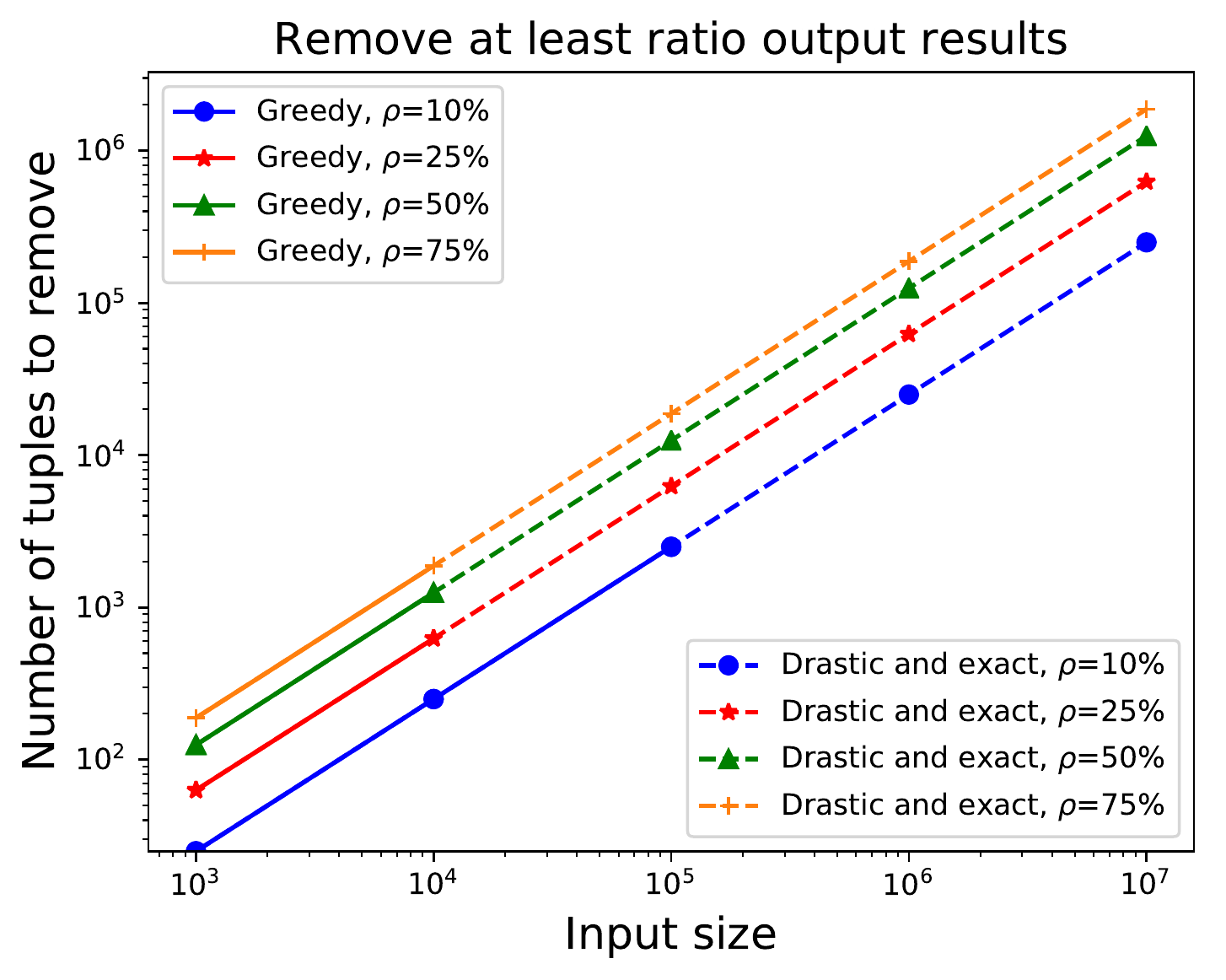}
	\caption{ \revm{Quality:  $\sigma_{\theta} Q_1$ (easy) by heuristics.}}
	\label{fig:Q1SelectionQuality}
	\endminipage \\
    \minipage{0.31\textwidth}
	\includegraphics[width=\linewidth]{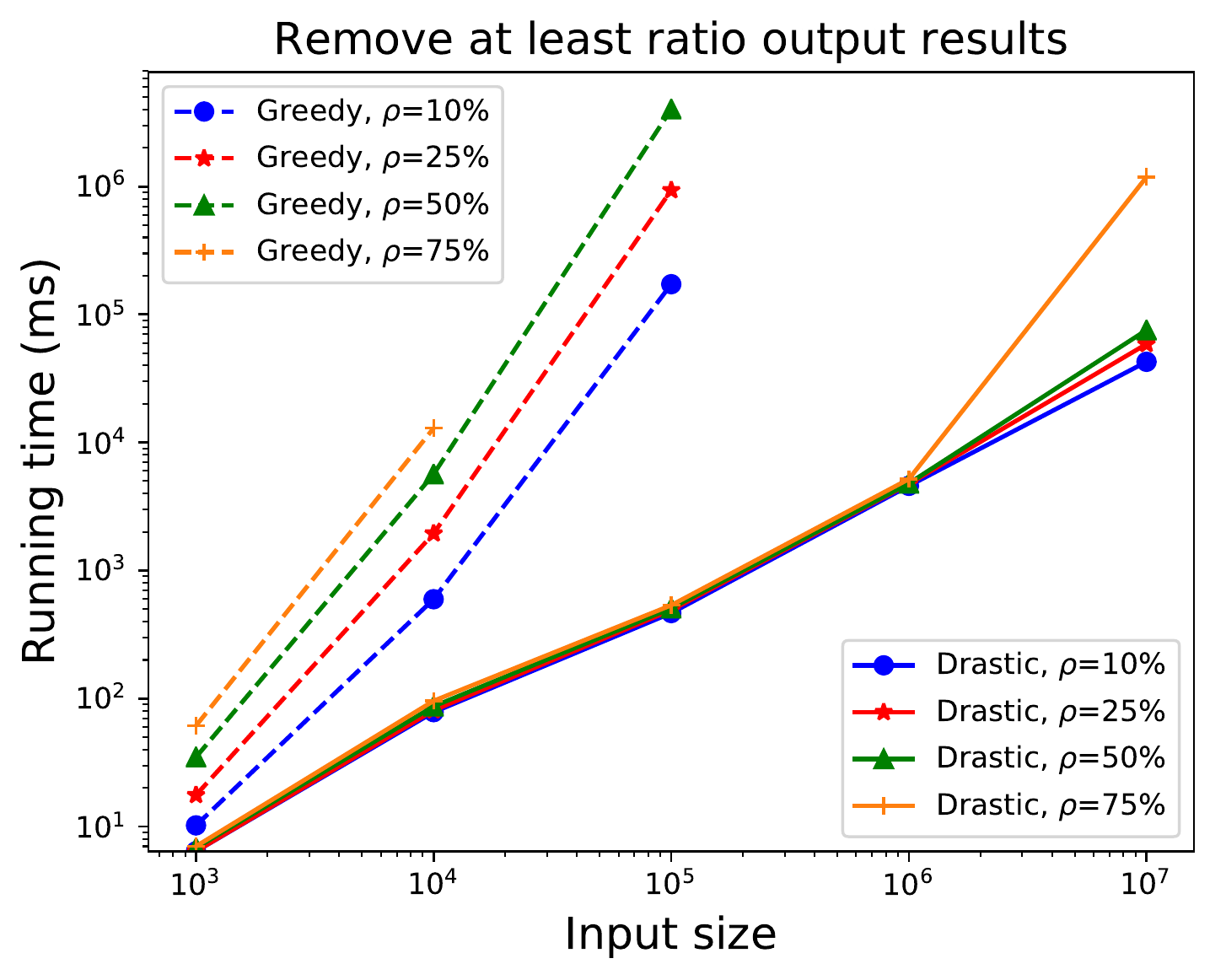}
	\caption{ \revm{Running Time: reporting $Q_1$ (hard) by heuristics.} }
	\label{fig:hard-time}
	\endminipage \hfill
	\minipage{0.31\textwidth}
	\includegraphics[width=\linewidth]{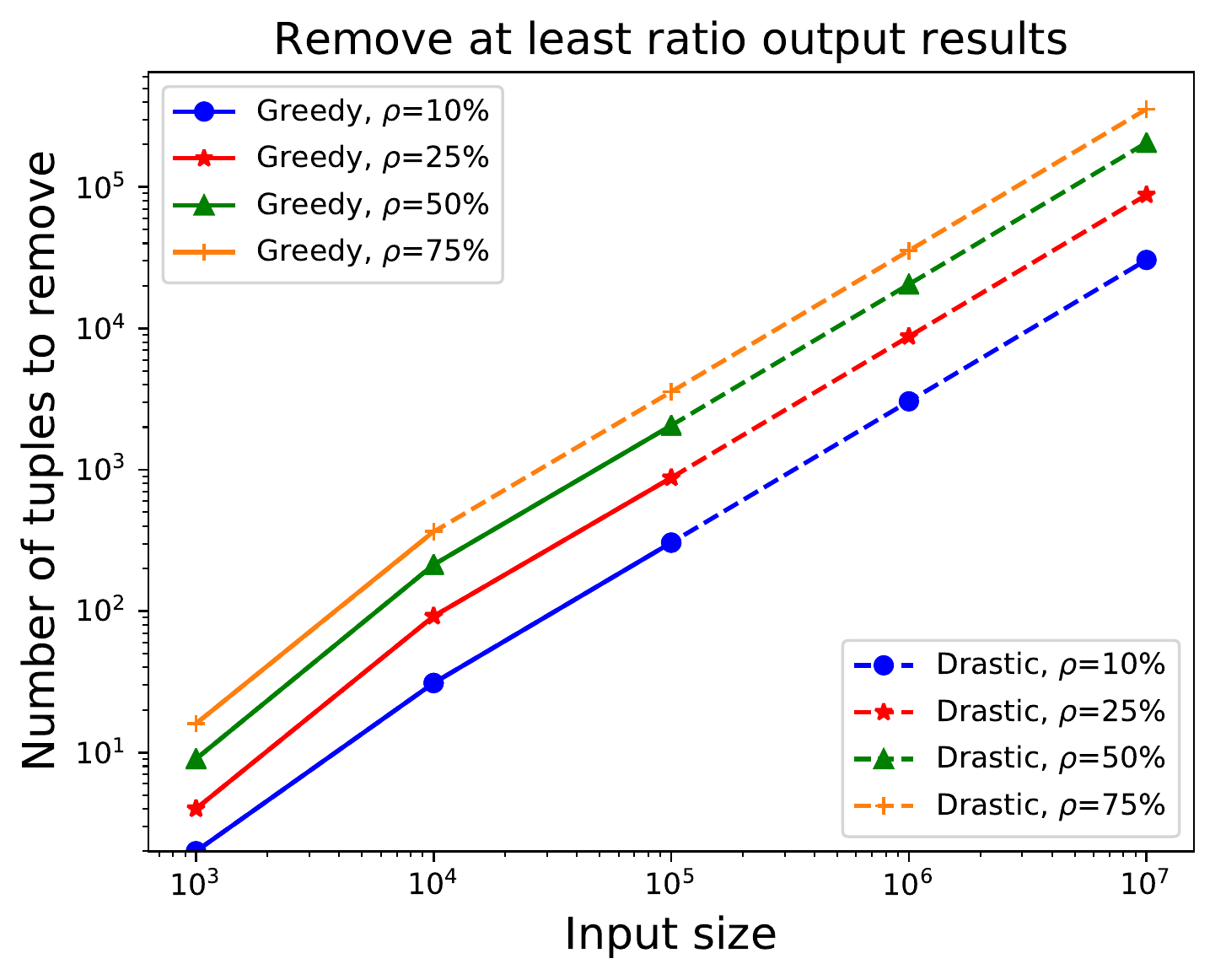}
	\caption{ \revm{Quality: $Q_1$ (hard) by heuristics.}}
	\label{fig:hard-quality}
	\endminipage \hfill
	\minipage{0.32\textwidth}
	\includegraphics[width=\linewidth]{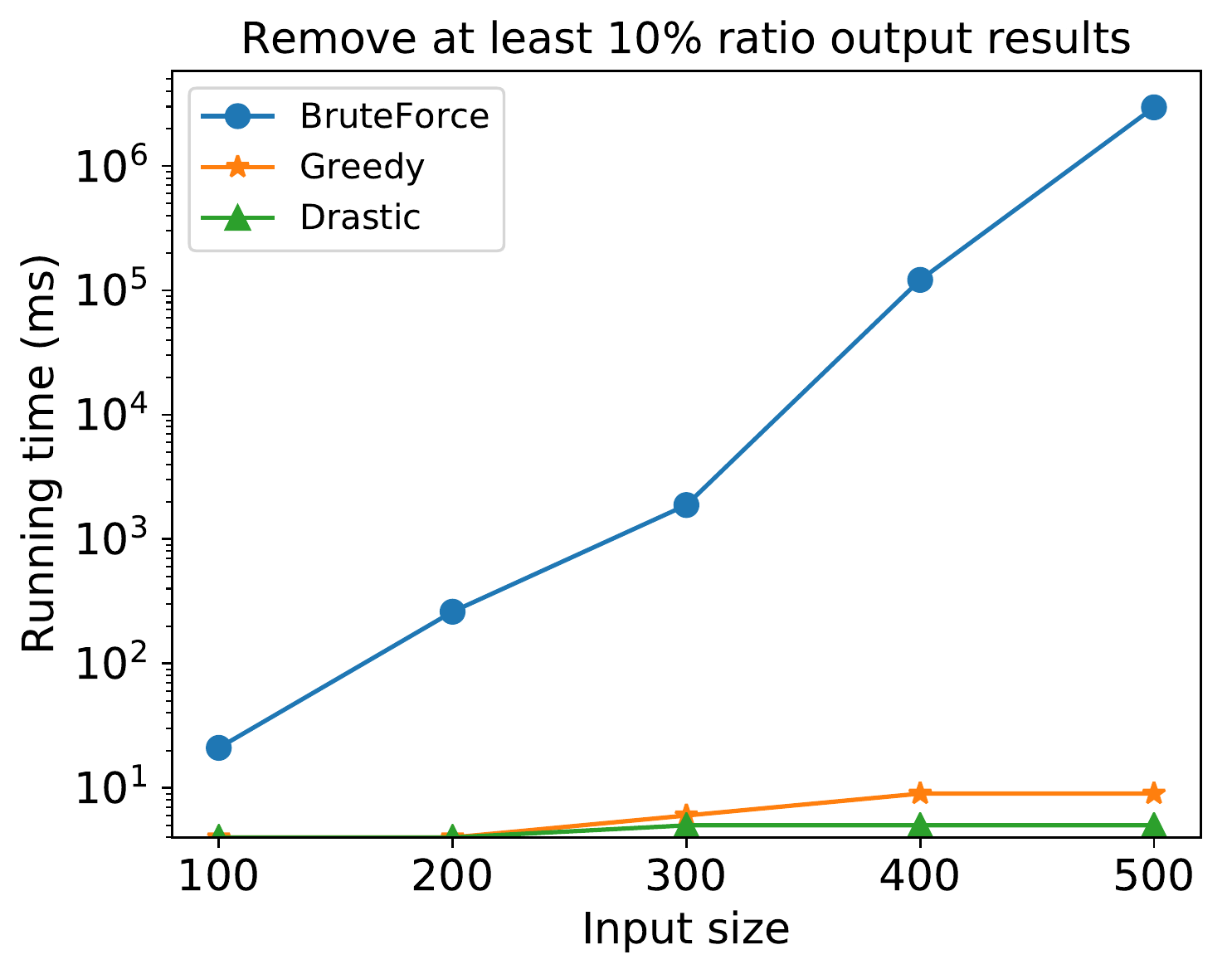}
	\caption{ \revm{Running Time:  brute-force v.s. heuristics for $Q_1$ (hard).}}
	\label{fig:bruteforceTime}
	\endminipage \\
    \minipage{0.32\textwidth}
	\includegraphics[width=1\linewidth]{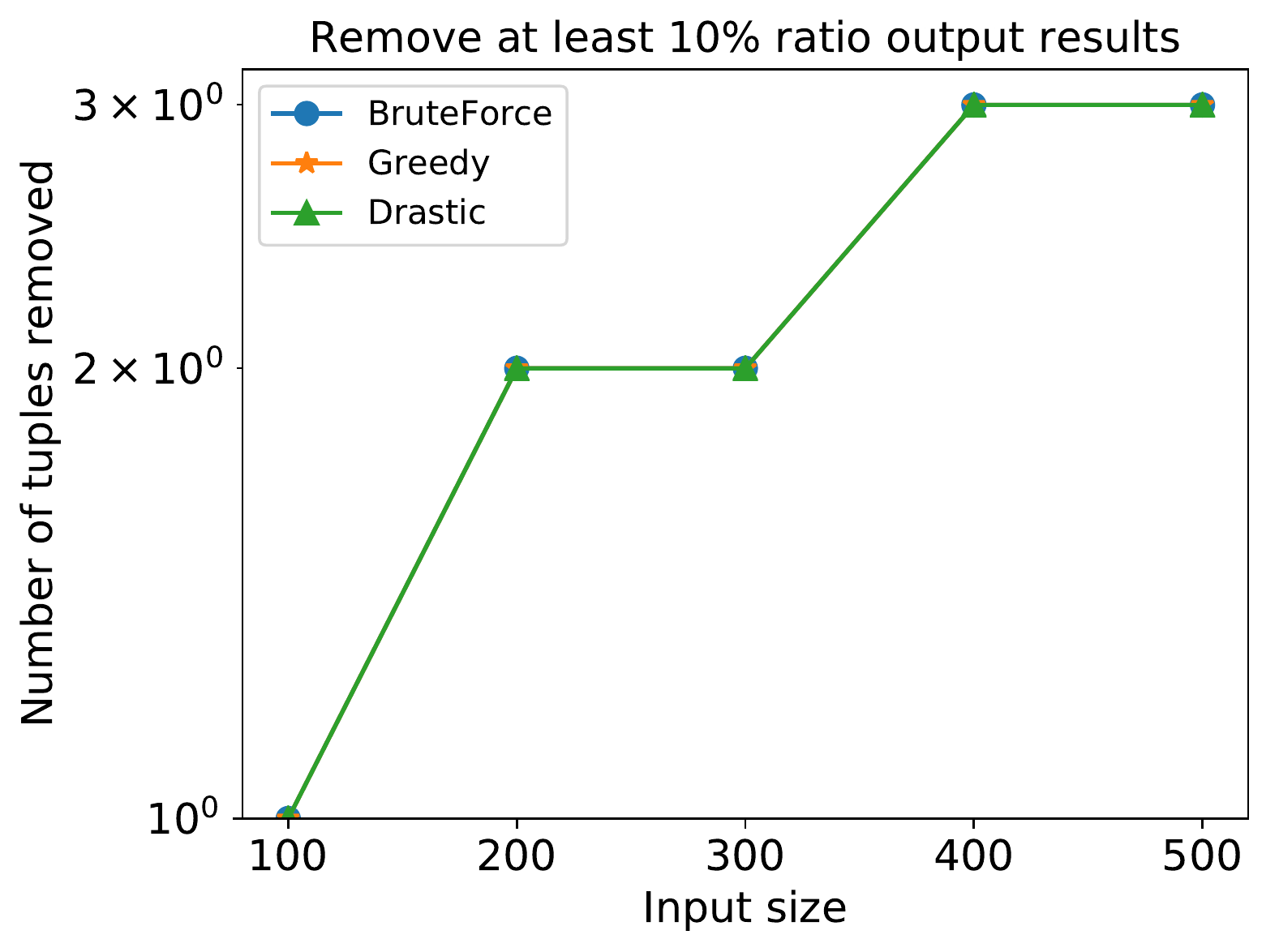}
	\caption{  \revm{Quality: brute-force v.s. heuristics for $Q_1$ (hard).}}
	\label{fig:bruteforceQuality}
	\endminipage \hfill
	\minipage{0.32\textwidth} 
	\includegraphics[width=\linewidth]{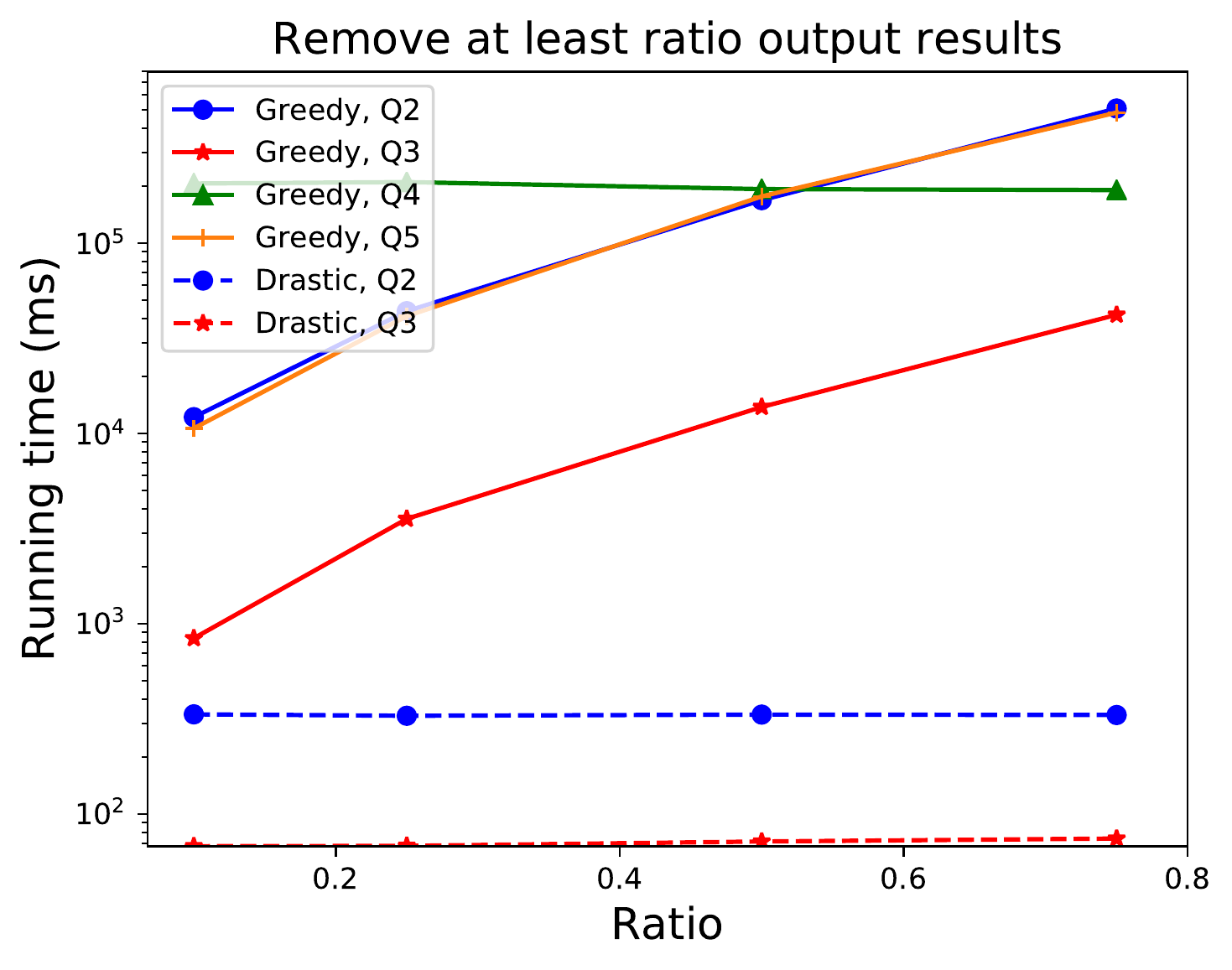}
	\caption{ \revm{Running Time:  %reporting 
	$Q_2$, $Q_3$, $Q_4$, $Q_5$ (hard) by heuristics.}}
	\label{fig:facebookTime}
	\endminipage \hfill
	\minipage{0.32\textwidth}
	\includegraphics[width=\linewidth]{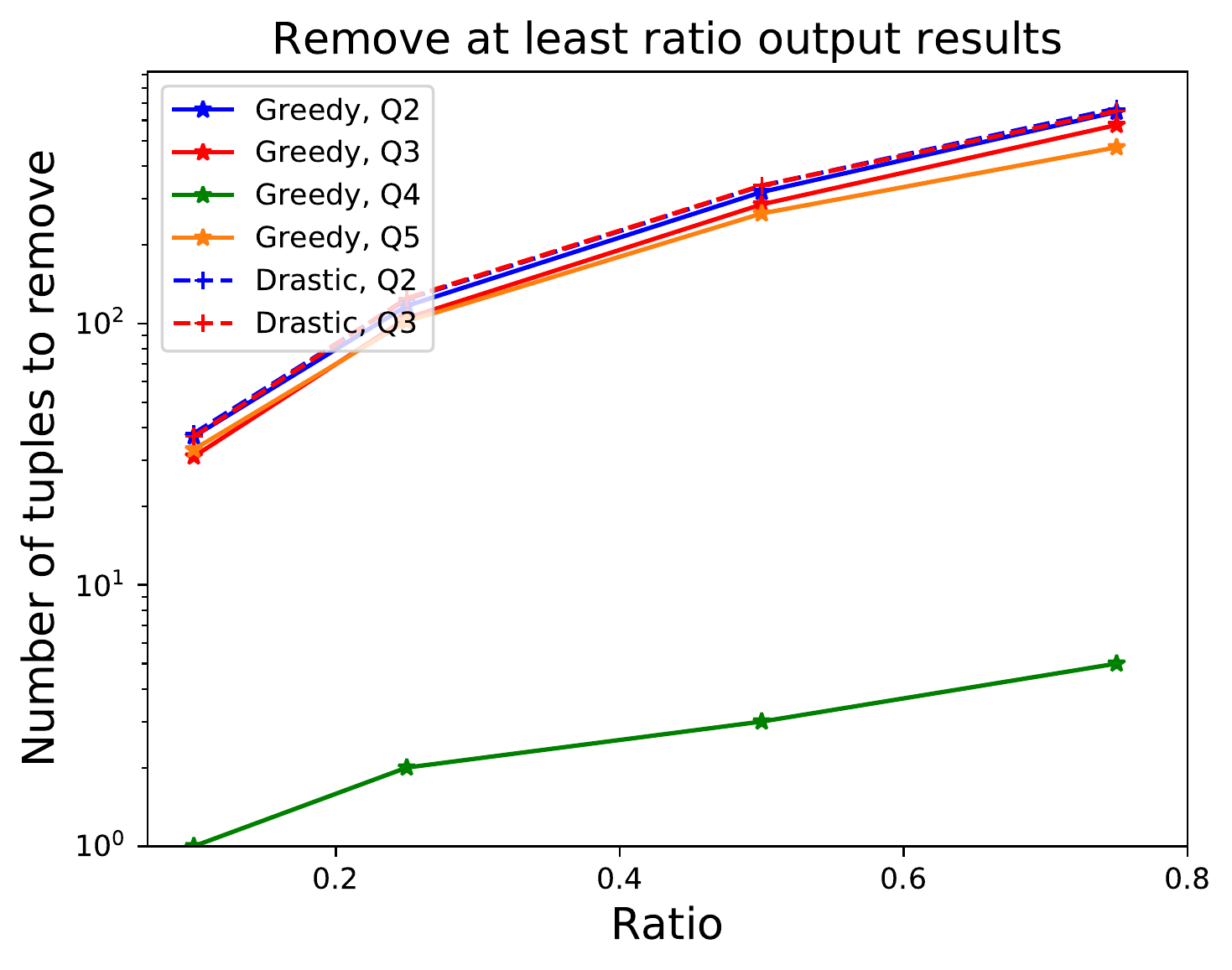}
	\caption{ \revm{Quality: $Q_2$, $Q_3$, $Q_4$, $Q_5$ (hard) by heuristics.}}
	\label{fig:facebookQuality}
	\endminipage
\end{center}
\end{figure*}

\section{Experiments}
\label{SEC:EXPERIMENT}

\revm{In this section, we evaluate the running time, scalability, and quality of \computeADP\ algorithm, and compare it with other baselines.}
\par
\revm{\textbf{Algorithms:}} In our plots, we call the exact algorithm using \computeADP\ for easy (poly-time) queries as ``{\sc Exact}''. For hard queries, and also for easy queries for scalability, we have implemented two versions of \computeADP\ embedded with {\sc GreedyForCQ} and {\sc DrasticGreedyForFullCQ} separately, shorted as ``\greedy'' and ``\drastic''. %We also  implementing greedy heuristics, we frequently invoke two SQL queries to the underlying database as primitives: (1) Given any subset of input tuples $\Delta \subseteq D$, how many query results can be removed by removing $\Delta$ from $D$?
%This can be answered by issuing two SQL queries $Q(D),Q(D-\Delta)$ and compute the difference between these two quantities. 
%(2) Given a relation $R_i$, how many full join results can be removed by removing each separate tuple from $R_i$? %This can be answered by issuing a SQL query with an aggregate on $R_i$.
%We invoke (1) for Algorithm~\ref{algo:greedy}, %and {\sc BruteForce}, and (2) for Algorithm~\ref{algo:greedy-full}. 
%
%Besides that, 
We also implemented a baseline brute-force algorithm called ``{\sc BruteForce}'', which enumerates all subsets of input tuples, computes the number of query results that can be removed by each subset (by invoking a SQL query), and finds the minimum one among which removes at least $k$ results.  
\par
\revm{\textbf{Reporting vs. counting versions:} Wherever applicable and feasible, we report the running time for both \emph{counting version}, when the goal is to only count the minimum number of input tuples to remove to achieve the desired effect, and the \emph{reporting version}, which reports the actual input tuples in one such solution. Note that for some of our motivating examples, e.g., for understanding robustness, the counting version suffices.
}
\par
\textbf{\revm{Setup:}} We implemented our algorithms in JavaSE-1.8 with the database stored in PostgreSQL 10.12. The experiments were performed on MacOS, with 16GB of RAM and Intel Core i7 2.9 GHz processor. %All codes are public at~\cite{shawn}.
We run the experiment 10 times and present the average results (metric) of the 10 runs.
%}

\subsection{\revm{Datasets and Queries}}\label{sec:expt-data}
\textbf{\revm{TPC-H dataset and queries:}} The TPC-H dataset has three relations: Supplier(S:NK, SK), PartSupp(PS:SK, PK), LineItem(L:OK, SK, PK).
%\end{small}
Consider the following two queries: %Due to the new trading law, r
\emph{(1) Remove least number of orders or suppliers so that at least $\rho$\% trading records can be restricted. (2) The same query but
for the specific PartKey = 13370.} %? 
They can be characterized by two problems $\ourprob(Q_1, D, k)$ and $\ourprob(\sigma_\theta Q_1, D, k_\theta)$  respectively, where
\begin{itemize}
	\item $Q_1$(NK, SK, PK, OK):-Supplier(S: NK, SK), PartSupp(PS: SK, PK), LineItem(L: OK, PK), $\theta: PK = 13370$, $k_\theta = \rho \cdot |\sigma_{\theta} Q(D)|$ and $k = \rho \cdot |Q(D)|$, where $\rho$ fraction of outputs are removed.
\end{itemize}
As shown in Lemma~\ref{LEM:SELECTION}, the $\ourprob(\sigma_\theta Q_1, D, k)$ is poly-time solvable with exact optimal solution returned, while the $\ourprob(Q_1, D, k)$ is NP-hard with only heuristic solution returned, by \computeADP.

\textbf{\revm{SNAP dataset and queries}:} We use the common ego-networks from SNAP (Stanford Network Analysis Project)~\cite{snap} for Facebook, where an ego-network of a user is a set of ``social circles'' formed by this user's friends~\cite{leskovec2012learning}. This dataset consists 10 ego-networks, 4233 circles, 4039 nodes, and 88234 edges. We choose the network around user 414 which consists of 7 circles, 150 nodes and 3386 edges. We further create tables $R_i(A,B)$ for $i \in [4]$ and insert $E_j$ into $R_i$ if the rank of $E_j \mod 4 = i$. All edges are bi-directed. We evaluate three different queries on this dataset as below:
\begin{itemize}
\itemsep0em
	\item $Q_2(A,B,C,D):- R_1(A,B), R_2(B,C), R_3(C,D)$
	\item $Q_3(A,B,C):- R_1(A,B), R_2(B,C), R_3(C,A)$
	\item $Q_4(A,C,E,G):- R_1(A,B), R_2(B,C), R_3(E,F), R_4(F,G)$. 
	\item \revb{$Q_5(A,B,C):- R_1(A,E), R_2(B,E), R_3(C,E)$}
\end{itemize} 
which are commonly used in community detection or friend recommendation over social networks. For instance, $Q_2$ finds a path of length three, $Q_3$ finds a triangle, $Q_4$ finds a pair of length-2 connection, and $Q_5$ captures a common friend. All of them are NP-hard, so \computeADP\ only returns heuristic results. %\revb{For instance, if three relations $R_1, R_2, R_3$ denote different types of connections (work, personal, or online acquaintances), $Q_2$ finds a path of length three, $Q_3$ finds a triangle, $Q_4$ finds a pair of length-2 connection, and $Q_5$ captures a common friend. Note that all these four queries are NP-hard by Theorem~\ref{thm:dichotomy}, so \computeADP\ only returns heuristic results for them. }

\subsection{Scalability} 

\textbf{\revm{Poly-time query:}} We evaluate $\ourprob(\sigma_\theta Q_1, D, k_\theta)$ on the TPC-H dataset with different input sizes \revm{$N = $1k, 10k, 100k, 1M, 10M}, which denotes the number of survived tuples after selection. \revm{We use different fractions $\rho = 0.1, 0.25, 0.5, 0.75$. Figure~\ref{fig:Q1Selection} display the results for both reporting and counting versions. The running time increases with increase of input data size and the $\rho$. Since the counting version only performs computation on numbers in dynamic programming, it uses much less memory and behaves much more scalable than the reporting version does. Moreover, as a remedy for reporting results when the data size becomes large, we also test the \greedy\ and \drastic\ on $\sigma_\theta Q_1$ (by %arbitrarily 
%forcing 
directly invoking Line 5
in Algorithm~\ref{algo:computeadp}), whose running time is much smaller than the exact algorithm as shown in Figure~\ref{fig:heuristicsEasyTime}. Meanwhile, we also show the quality of these three techniques in Figure~\ref{fig:Q1SelectionQuality}. All of them coincide due to the data distribution for $\sigma_{\theta}Q_1$, which implies that \greedy\ and \drastic\ 
also find optimal solutions. But \greedy\ is not as scalable as \drastic\ to larger dataset with input size 100K or more.

\textbf{Hard query:} We next evaluate $\ourprob(Q_1, D, k_\theta)$ on the TPC-H dataset with different input sizes $N = $1k, 10k, 100k, 1M, 10M and $\rho = 0.1, 0.25, 0.5, 0.75$ using \greedy\ and \drastic\ separately. Since \drastic\ only computes the ``profit'' for all input tuples through a SQL query once, while \greedy\ needs to update these statistics once an input tuple is removed. Thus, \drastic\ takes much less time than \greedy, as shown in Figure~\ref{fig:hard-time}.   
We also compare the quality of solutions returned by these two heuristics, as shown in Figure~\ref{fig:hard-quality}. Due to the data distribution (which is varied in Section~\ref{sec:expt-distribution}), \greedy\ and \drastic\ have the same quality when data size is smaller than 100K. However, \greedy\ is not scalable to larger dataset and quality results are only shown for \drastic\ in Figure~\ref{fig:hard-quality}.

\textbf{Comparison with brute-force:} Next, we  evaluate the {\sc BruteForce} algorithm on the TPC-H dataset for the NP-hard query $\ourprob(Q_1, D, k)$ with input size $N= 500$ and $\rho = 0.1$. The straightforward brute-force implementation does not work even on such a small dataset, since it iterates over all subsets of input tuples and issues as many as $2^{500}$ SQL queries in total. We use an optimization here by iterating all subsets in increasing order of their sizes, until a feasible solution (removing at least $k$ query results) is found. 

We compare the optimized {\sc BruteForce} with two heuristics. All three algorithms have their quality coinciding for this small dataset, as shown in Figure~\ref{fig:bruteforceQuality}. But heuristics significantly improve the running time of {\sc BruteForce}, as shown in Figure~\ref{fig:bruteforceTime}. The {\sc BruteForce} did not stop in several hours for $N = 1000$ or $\rho = 0.2$.

%{\bf ------ describe the dataset. Add a line that we used these optimizations for brute force...trivial one iterating over all subsets did not work...had to generate size k+1 from size k..  Other $\rho$? Any optimization for Brute-force? It did not stop in 2 hours for n = 1000...}
%{\bf Fig 8 HAS BOTH DR. AND EXACT.}

}

\eat{It should be noted that we count the number of tuples passing the selection as input size for $\ourprob(\sigma_\theta Q_1, D, k)$, while the number of tuples in total as input size for  $\ourprob(Q_1, D, k_\theta)$.
 Figures~\ref{fig:tpch-w} and \ref{fig:tpch-o} display the results. As expected, the running time increases with increasing input data size $|D|$ and the number of query results to be removed. For smaller dataset, the algorithms take a few seconds whereas for larger dataset it goes to more than few hundred seconds finishing in reasonable time.}
%(for the applications for \ourprob, interactive performances may not be critical).
\eat{
\begin{figure}[h]
	\minipage{0.23\textwidth}
	\includegraphics[width=\linewidth]{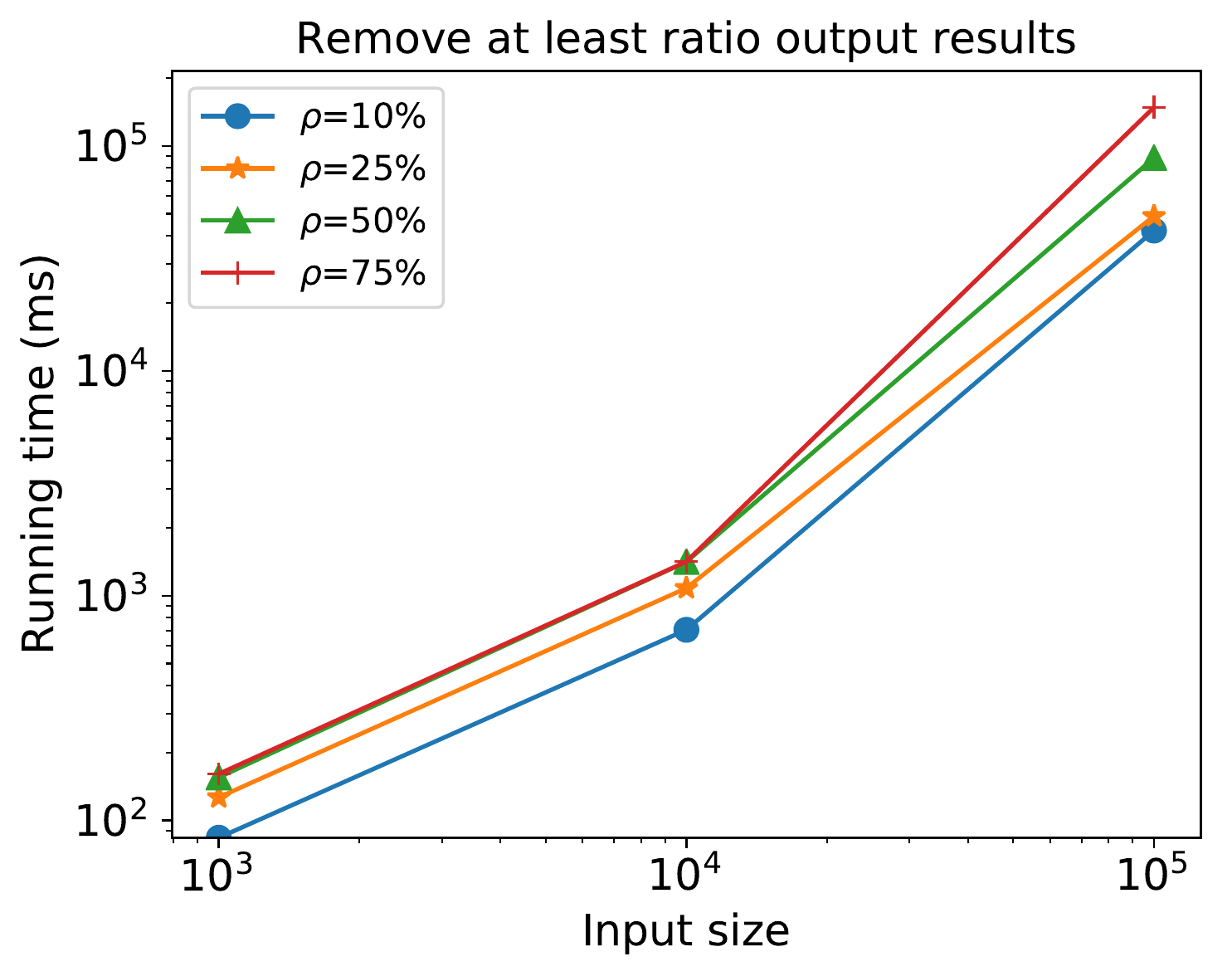}
	\caption{Query $\sigma_{\theta}Q_1$.}
	\label{fig:Q1SelectionReport}
	\endminipage\hfill
	\minipage{0.23\textwidth}
	\includegraphics[width=\linewidth]{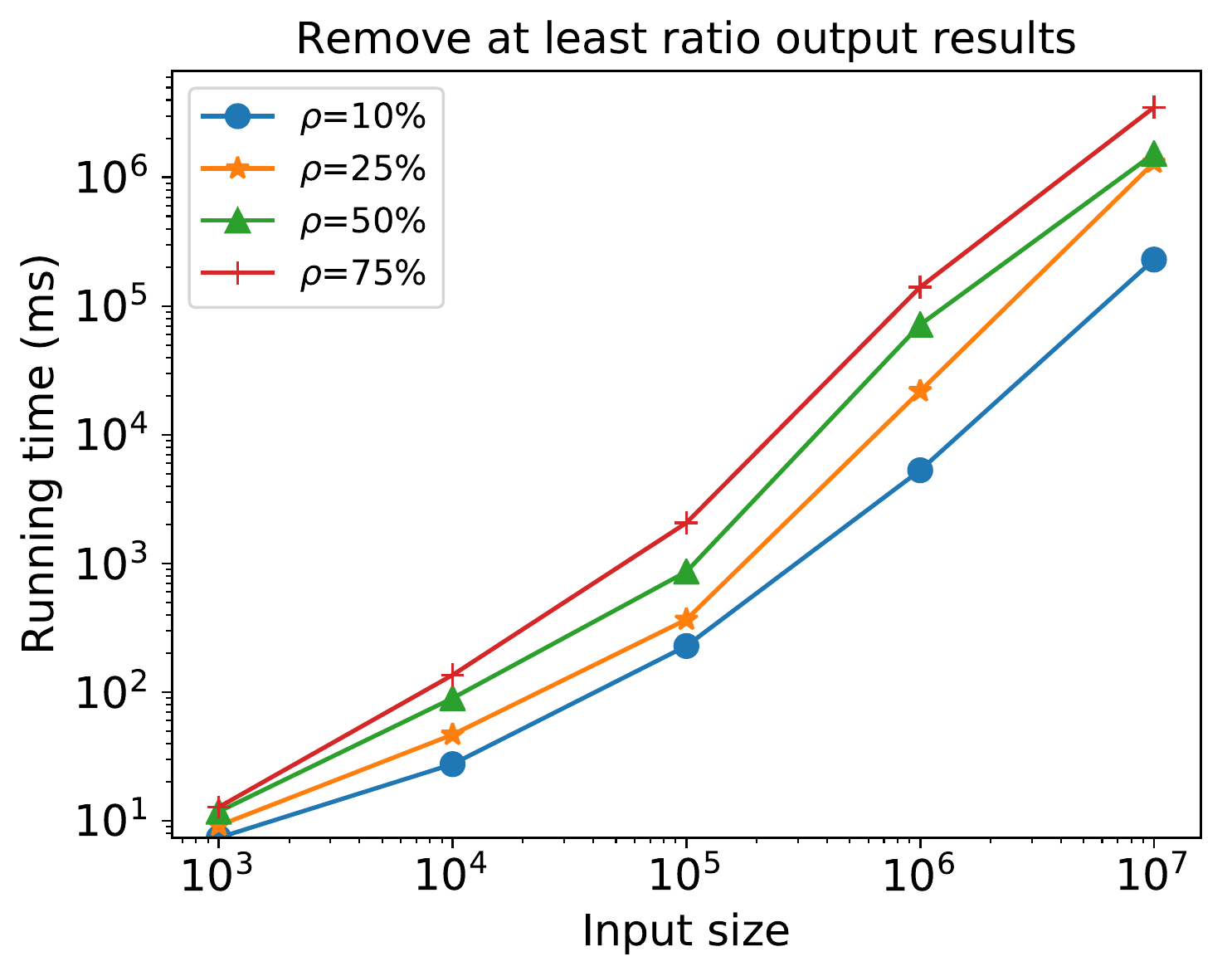}
	\caption{Query $Q_1$.}
	\label{fig:Q1SelectionCount}
	\endminipage\hfill
\end{figure}
}

\subsection{Complexity of Queries}

%{\bf DATA SIZE IS FIXED - WHY THESE QUERIES ARE TAKING DIFFERENT TIME.. ADD OLD Q4 -- VARY NUMBER AND REPORT..}

\revm{For each of $Q_2, Q_3, Q_4,Q_5$, we ran our experiments on the SNAP dataset and varied the fraction of query results to be removed (denoted as $\rho$) over $\{0.1, 0.25, 0.5, 0.75\}$. We evaluated \greedy\ and \drastic\ as follows. First, we invoked {\sc GreedyForCQ} directly on $Q_2, Q_3, Q_5$ since neither of the simplification steps can be applied to these queries. For $Q_4$, \greedy\ first decomposes it into two subqueries as $Q_{41}(A,C): -R_1(A,B), R_2(B,C)$ and $Q_{42}(E,G): -R_3(E,F), R_4(F,G)$ using \decompose, and handles them using {\sc GreedyForCQ} separately. Next, we invoked {\sc DrasticGreedyForFullCQ} on $Q_2, Q_3$ directly. All running times are displayed in Figure~\ref{fig:facebookTime}. As \drastic\ cannot be applied to $Q_4, Q_5$ with projection, these are not in Figure~\ref{fig:facebookTime}. The quality of these heuristics is displayed in Figure~\ref{fig:facebookQuality}.

%The running time of \drastic\ depends on the execution time of SQL queries in lines 3-5, sorting in line 6, and finding the first largest $i$ tuples to removed at least $k$ results in line 7. Note that $Q_2, Q_3$ are executed on the same dataset and the number of input tuples to be removed are almost the same (see Figure~\ref{fig:facebookQuality}). So Figure~\ref{fig:facebookTime} displays the time difference for executing the SQL queries for $Q_2, Q_3$ due to different query complexity.

The running time of \drastic\ depends on (i) the number of endogenous relations, (ii) computing the profits for all tuples in an endogenous relation by SQL queries, (iii) sorting the tuples by profit, and (iv) finding tuples with largest profits whose profits add up to at least $k$. Note that $Q_2, Q_3$ are executed on the same dataset and the number of input tuples to be removed are almost the same (see Figure~\ref{fig:facebookQuality}). So Figure~\ref{fig:facebookTime} displays the difference in runtimes for executing the SQL queries for $Q_2, Q_3$.

The running time of \greedy\ depends on (i) the number of %while 
iterations of the while loop, which is equal to the number of
input tuples to be removed, (ii) the number of SQL queries for each iteration of the while loop, which is the number of endogenous relations, and (iii) the time for executing one SQL query. On $Q_2, Q_3, Q_5$, \greedy\ removes almost the same number of tuples as shown in Figure~\ref{fig:facebookQuality}. So, Figure~\ref{fig:facebookTime} displays the difference in running time for executing SQL queries for $Q_2, Q_3, Q_5$ respectively. Note that \greedy\ needs to solve a dynamic program in \decompose\ as well as a large number of sub-problems for both $Q_{41}, Q_{42}$, which is only relevant to the sizes of their own query results, %which is much slower than running {\sc GreedyForCQ} for each connected component, 
so $Q_4$ has a larger and stable running time even though it removes much fewer input tuples. %\red{to be done: why $Q_4$ is stable.}
}%Moreover, observe that its solutions under different ratios removes less than $10$ input tuples from $Q_4$, which implies that solving the dynamic program takes almost the same time. Together with fact that solving subproblems for $Q_{41}, Q_{42}$ is irrelevant to the final input parameter $k$, so the running time for $Q_4$ is quite stable. }

%These four queries have different number of output results on the SNAP dataset, and removing the same ratio of query results requires removing different number of input tuples by heuristics, as shown in Figure~\ref{fig:facebookQuality}. Moreover, the running time of \greedy\ depends not only on the number of while iterations (Alogrithm 2), which is exactly the number of input tuples to be removed, but also the time for executing one SQL query inside one iteration of while loop in Algorithm~\ref{algo:computeadp}.  Thus, the running time and quality of \greedy\ on different queries have the following order ($Q_2$)
%
%On $Q_3$, \greedy\ is slightly better than \drastic\ and on both $Q_2, Q_4,Q5$ \greedy\ and \drastic\ have the same quality, which verifies the fact that \greedy\ has better (at least not worse) quality than \drastic\ on full CQs.}

\revm{
\subsection{Data Distribution}
\label{sec:expt-distribution}

%{\bf -- EASY QUERY R1(A) R2(A, B) -- COUNTING --- $\alpha = 0, 1$ --- EXACT ALGO (BOTH TIME AND QUALITY GRAPHS)}

We study the performance of \computeADP\ for a poly-time solvable singleton query $Q_6(A,B):-R_1(A),\\ R_2(A,B)$ and an NP-hard query $\pathquery(A,B):-R_1(A),R_2(A,B),R_3(B)$ on various data distributions, where the degrees of values from $A$ or $B$ in relation $R_2(A,B)$ is varied according to to obtain the different distributions. 
%generated according to some distributions. 
We used the Zipfian distribution, where the frequency of the $i$-th distinct key is proportional to $i^{-\alpha}$. The parameter $\alpha \ge 0$ controls the skewness of the distribution: larger $\alpha$ means larger skew. We fix the distribution of degrees for values in $B$ as uniform and vary the skewness of degrees of values in $A$ by varying $\alpha$. We evaluate both $Q_6$ and $\pathquery$ on our synthetic dataset with different input sizes $N = 1k, 10k, 100k, 1M$ and $0.2N$ distinct values in $A$ and $B$ separately. The results for $\pathquery$ are shown in Figure~\ref{fig:0skewTime}--\ref{fig:1skewQuality}, and those for $Q_6$ are shown in Figure~\ref{fig:0skewExactTime}--\ref{fig:1skewExactQuality}. 
We also tested other values of $\alpha$, which are reported in Figures~\ref{fig:0.25skewTime}, \ref{fig:0.25skewQuality}, \ref{fig:0.5skewTime}, \ref{fig:0.5skewQuality}.

For every fixed value of $\alpha$, the running time as well as the size of solutions returned by any algorithm increase with the input size and the value of $\rho$. If both the input size and $\rho$ are fixed, the size of the solution decreases with increasing $\alpha$. This is because on a skewed instance, the same number of output tuples 
%from $R_2$ can be done by removing 
can be removed by removing fewer input tuples. 
The running time for \drastic\ and %\singleton\
{\sc Exact}
stays almost the same since computing the profits for input tuples is the most costly step, independent of the size of the solution. However, the running time of \greedy\ decreases with the size of the solution, which is affected by $\alpha$.

\begin{figure*}
	\minipage{0.24\textwidth}
	\includegraphics[width=\linewidth]{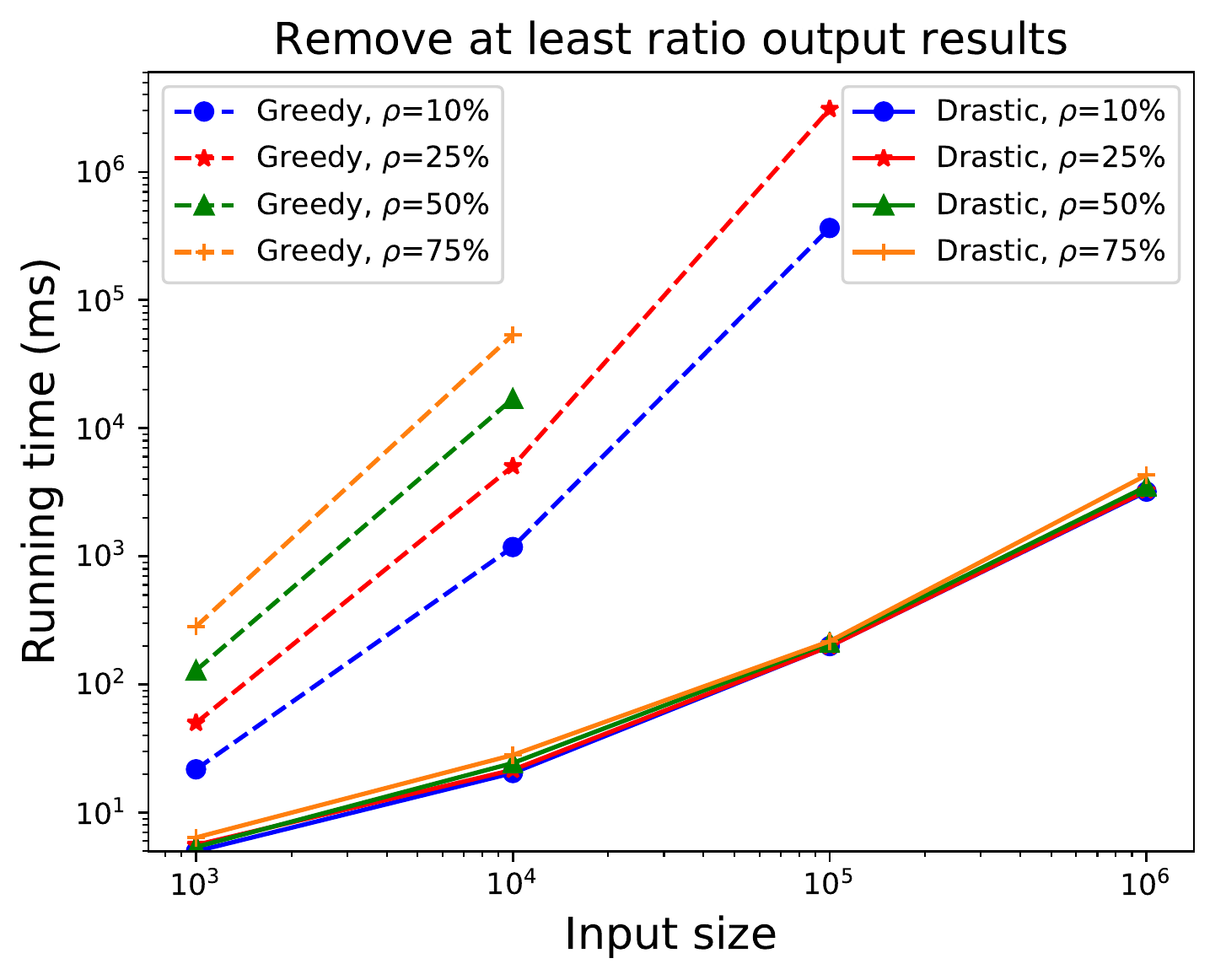}
	\caption{ {\revm{$\alpha = 0$ (hard)}}}
	\label{fig:0skewTime}
	\endminipage\hfill
	\minipage{0.24\textwidth}
	\includegraphics[width=\linewidth]{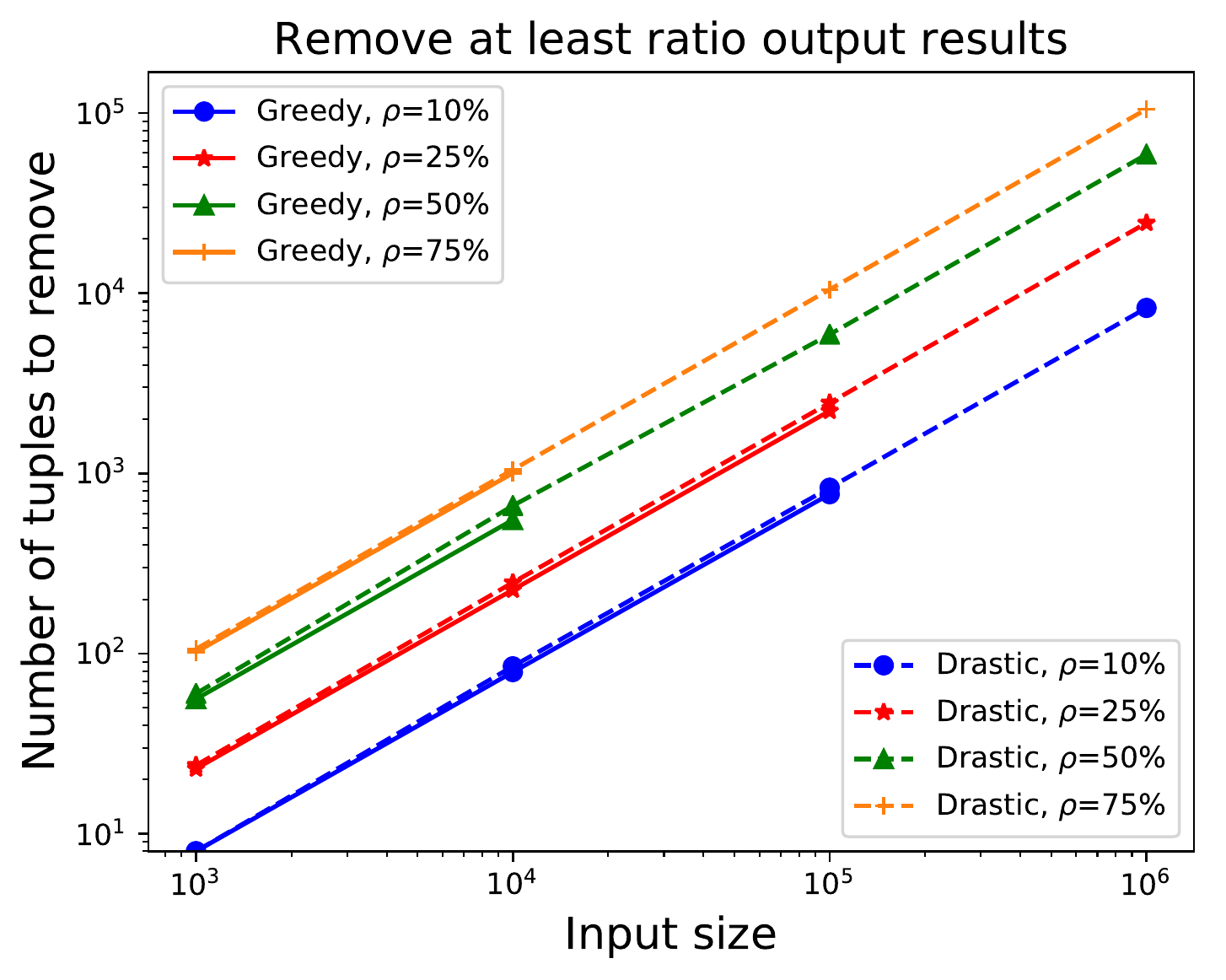}
	\caption{\revm{$\alpha = 0$ (hard)}}
	\label{fig:0skewQuality}
	\endminipage\hfill
	\minipage{0.24\textwidth}
	\includegraphics[width=\linewidth]{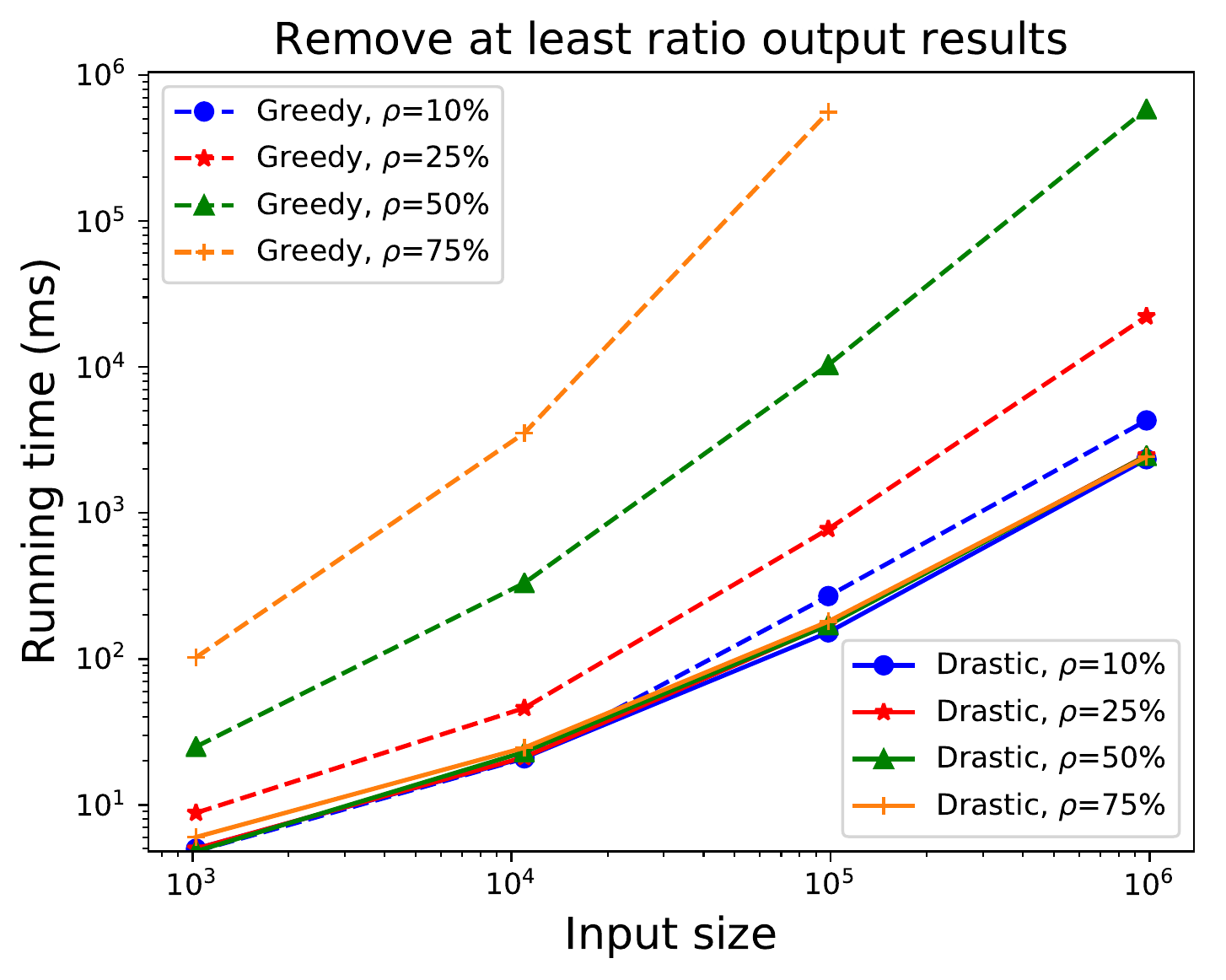}
	\caption{\revm{$\alpha = 1$ (hard)}}
	\label{fig:1skewTime}
	\endminipage \hspace{0.6em}
	\minipage{0.24\textwidth}
	\includegraphics[width=\linewidth]{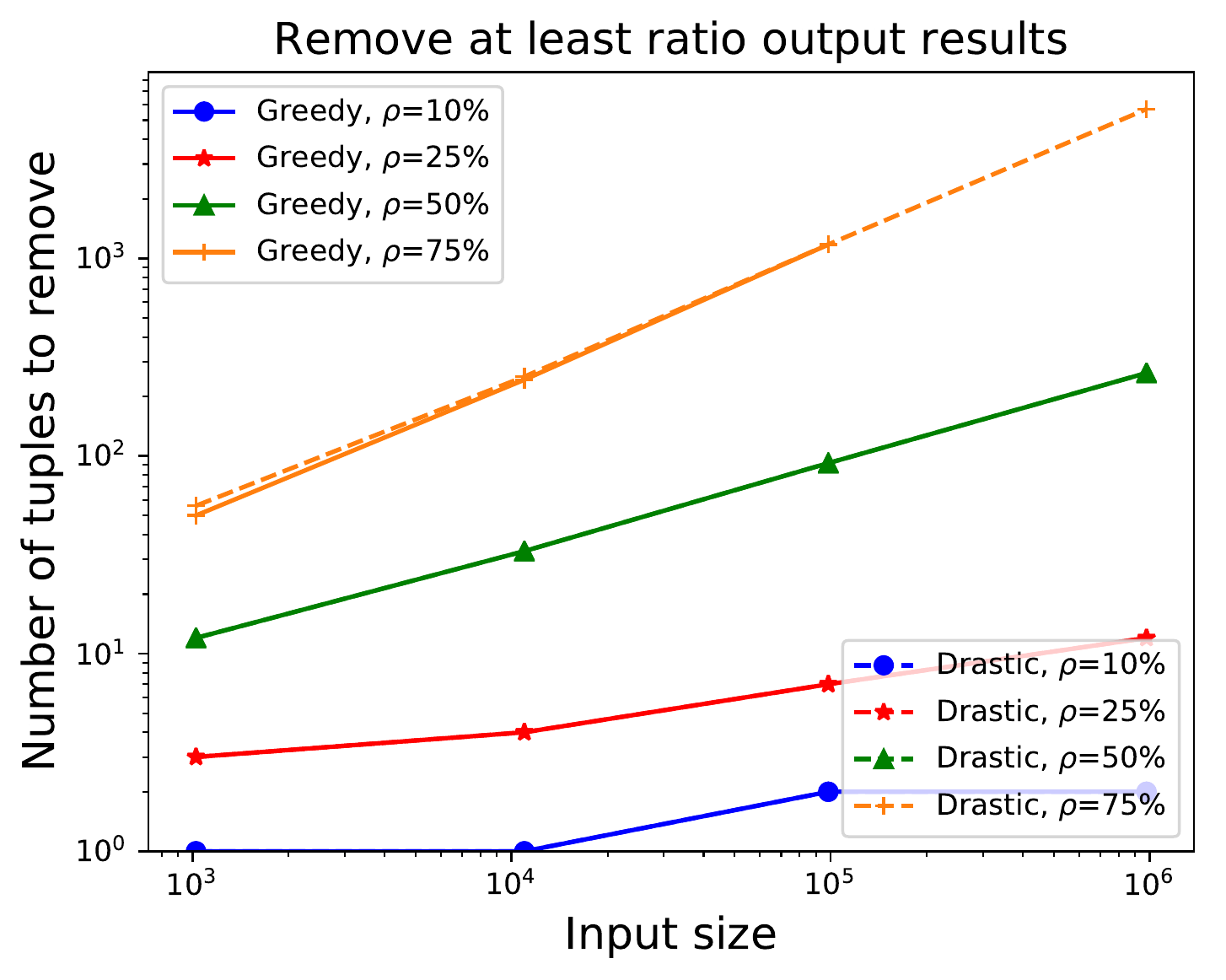}
	\caption{\revm{$\alpha = 1$ (hard)}}
	\label{fig:1skewQuality}
	\endminipage\hfill
	\minipage{0.24\textwidth}
	\includegraphics[width=\linewidth]{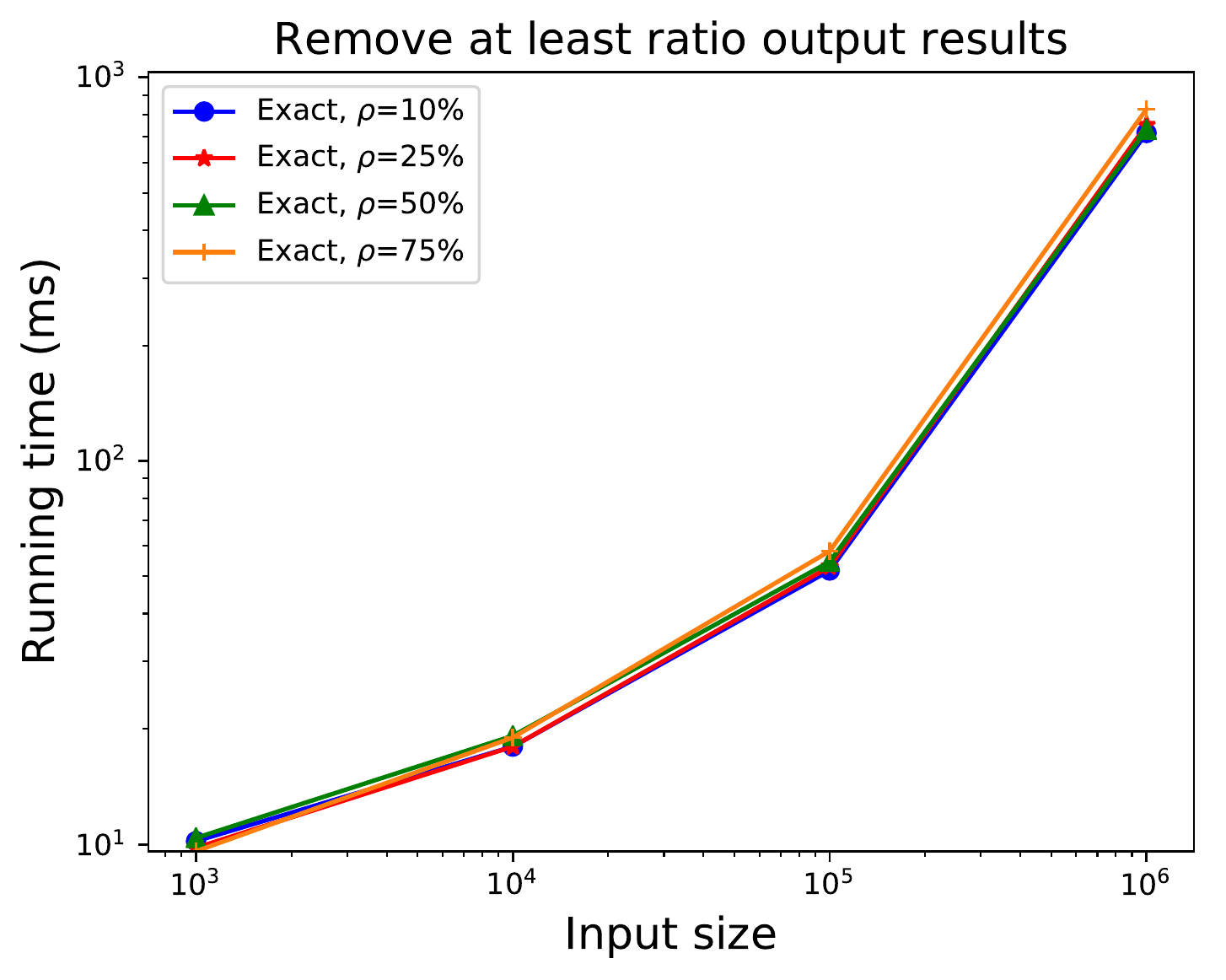}
	\caption{\revm{$\alpha = 0$ (easy)}}
	\label{fig:0skewExactTime}
	\endminipage\hfill
	\minipage{0.24\textwidth}
	\includegraphics[width=\linewidth]{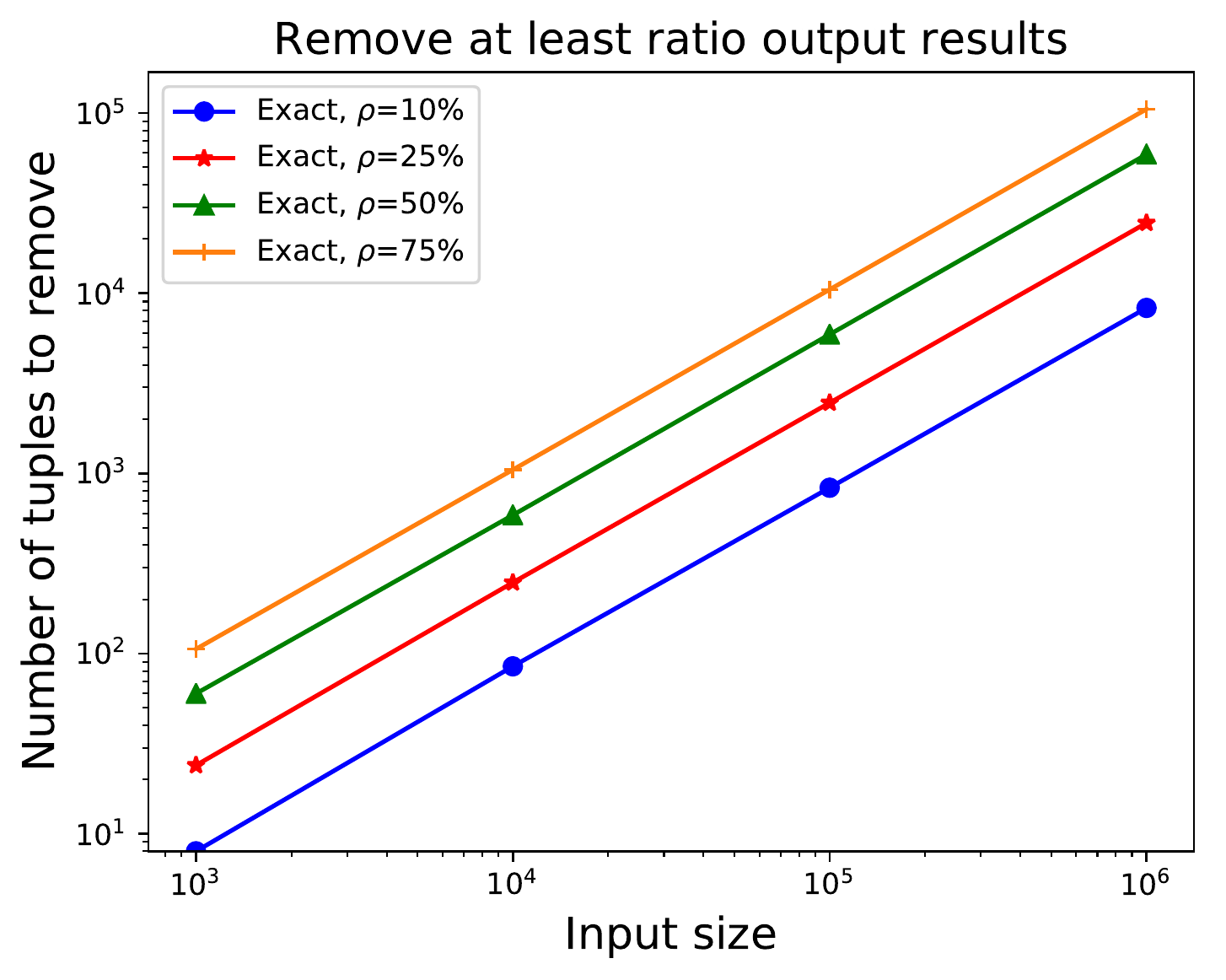}
	\caption{\revm{$\alpha = 0$ (easy)}}
	\label{fig:0skewExactQuality}
	\endminipage\hfill
	\minipage{0.24\textwidth}
	\includegraphics[width=\linewidth]{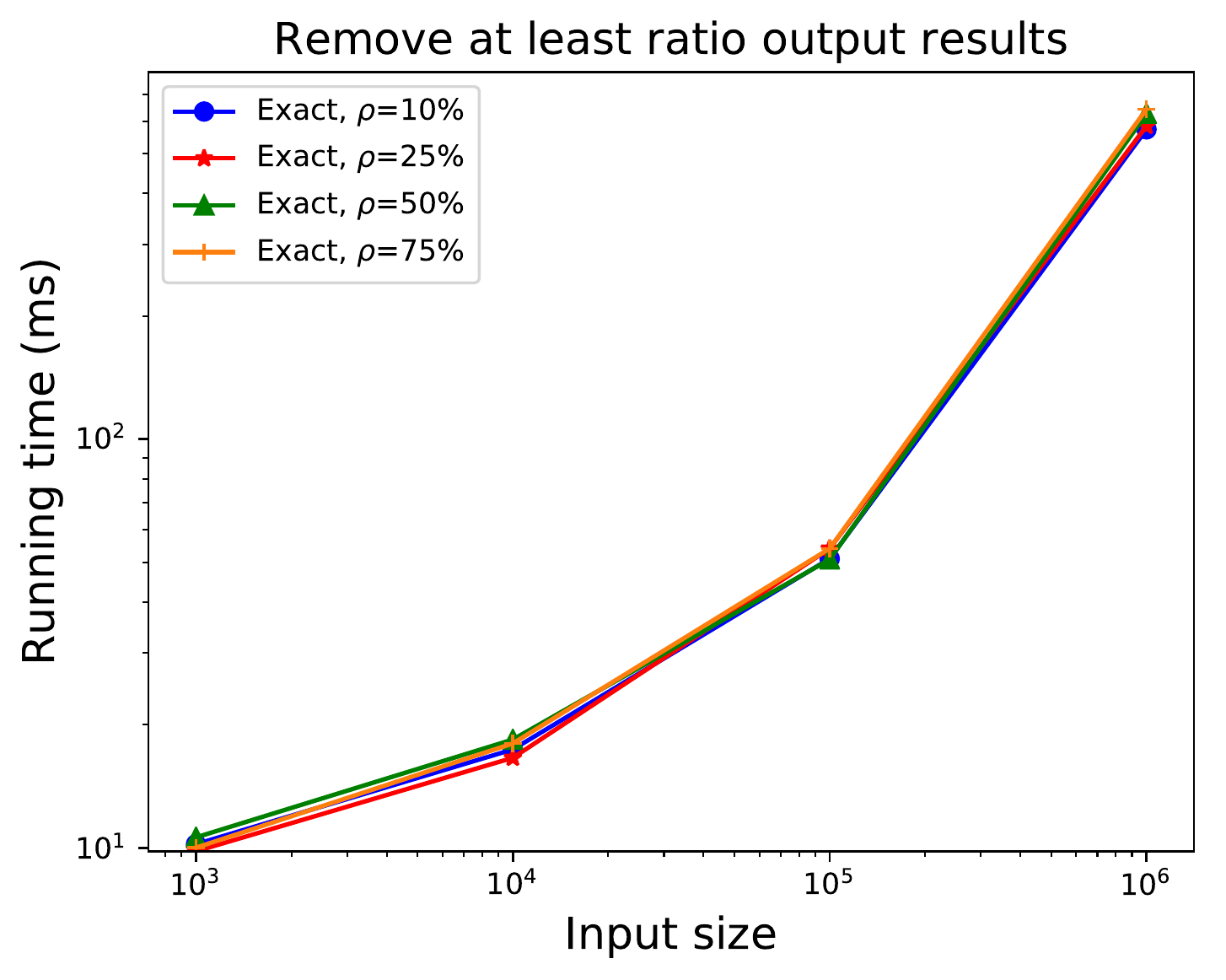}
	\caption{\revm{$\alpha = 1$ (easy)}}
	\label{fig:1skewExactTime}
	\endminipage \hspace{0.6em}
	\minipage{0.24\textwidth}
	\includegraphics[width=\linewidth]{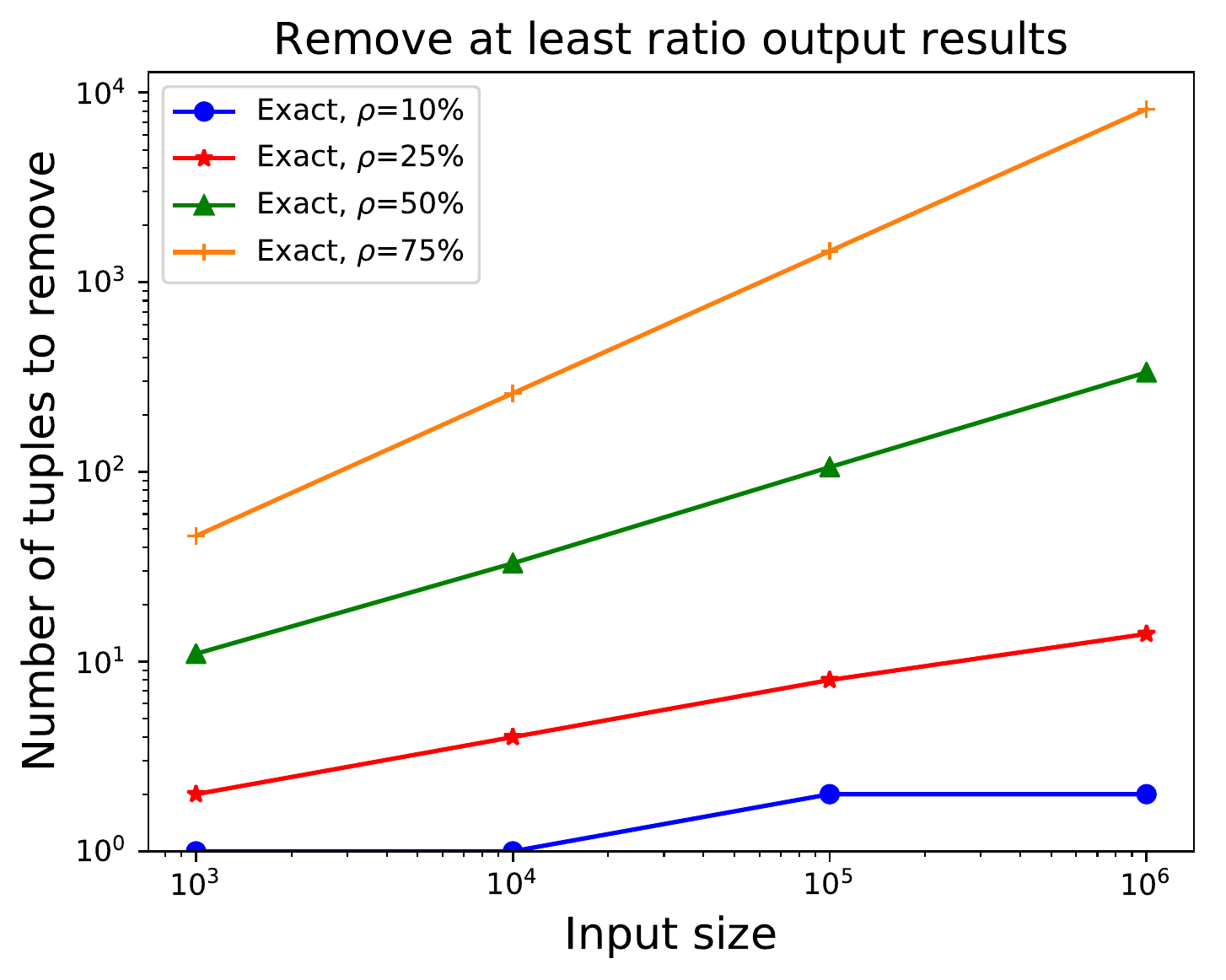}
	\caption{\revm{$\alpha = 1$ (easy)}}
	\label{fig:1skewExactQuality}
	\endminipage\hfill
	\minipage{0.25\textwidth}
	\includegraphics[width=\linewidth]{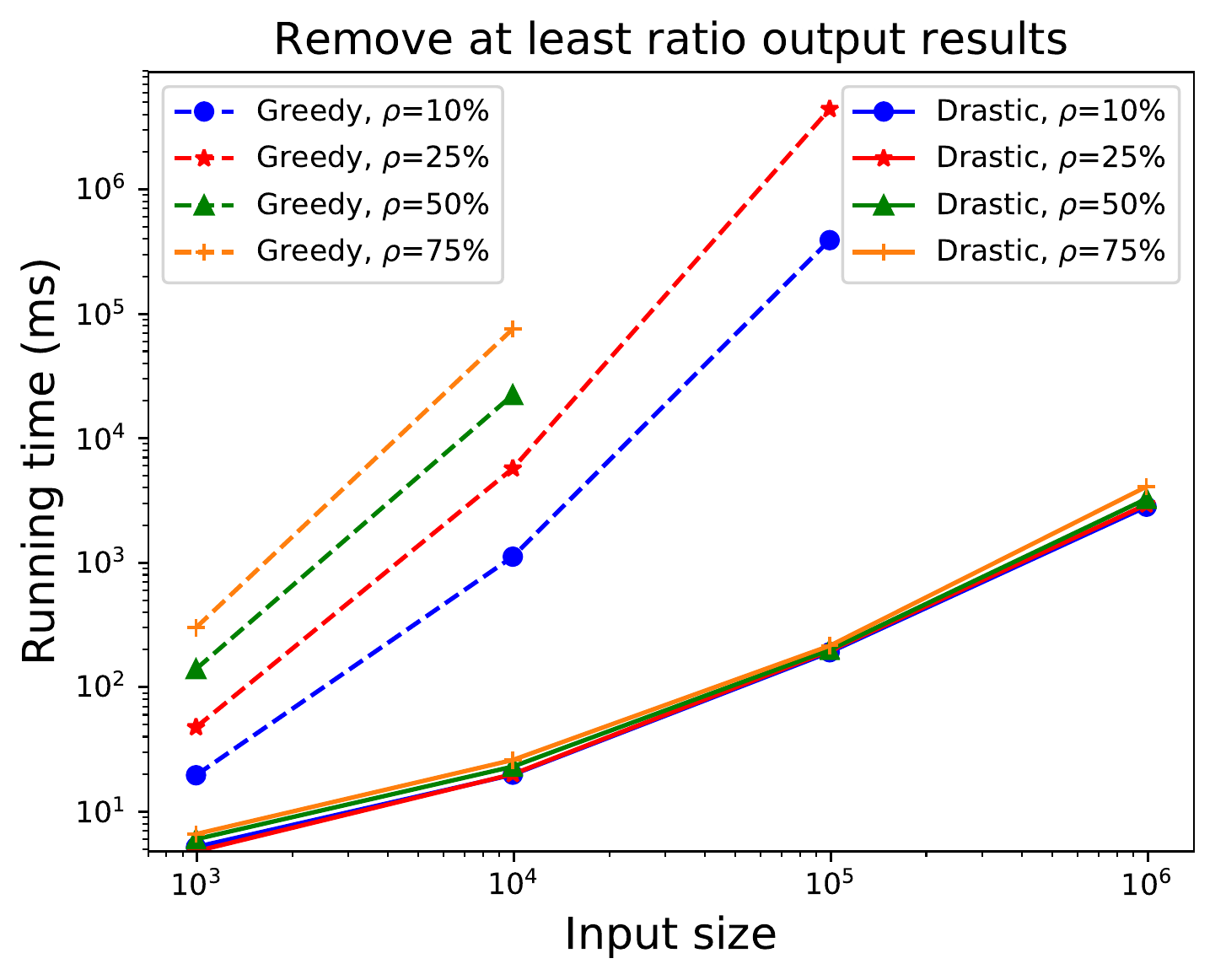}
	\caption{\revm{$\alpha = 0.25$ (hard)}}
	\label{fig:0.25skewTime}
	\endminipage\hfill
	\minipage{0.25\textwidth}
	\includegraphics[width=\linewidth]{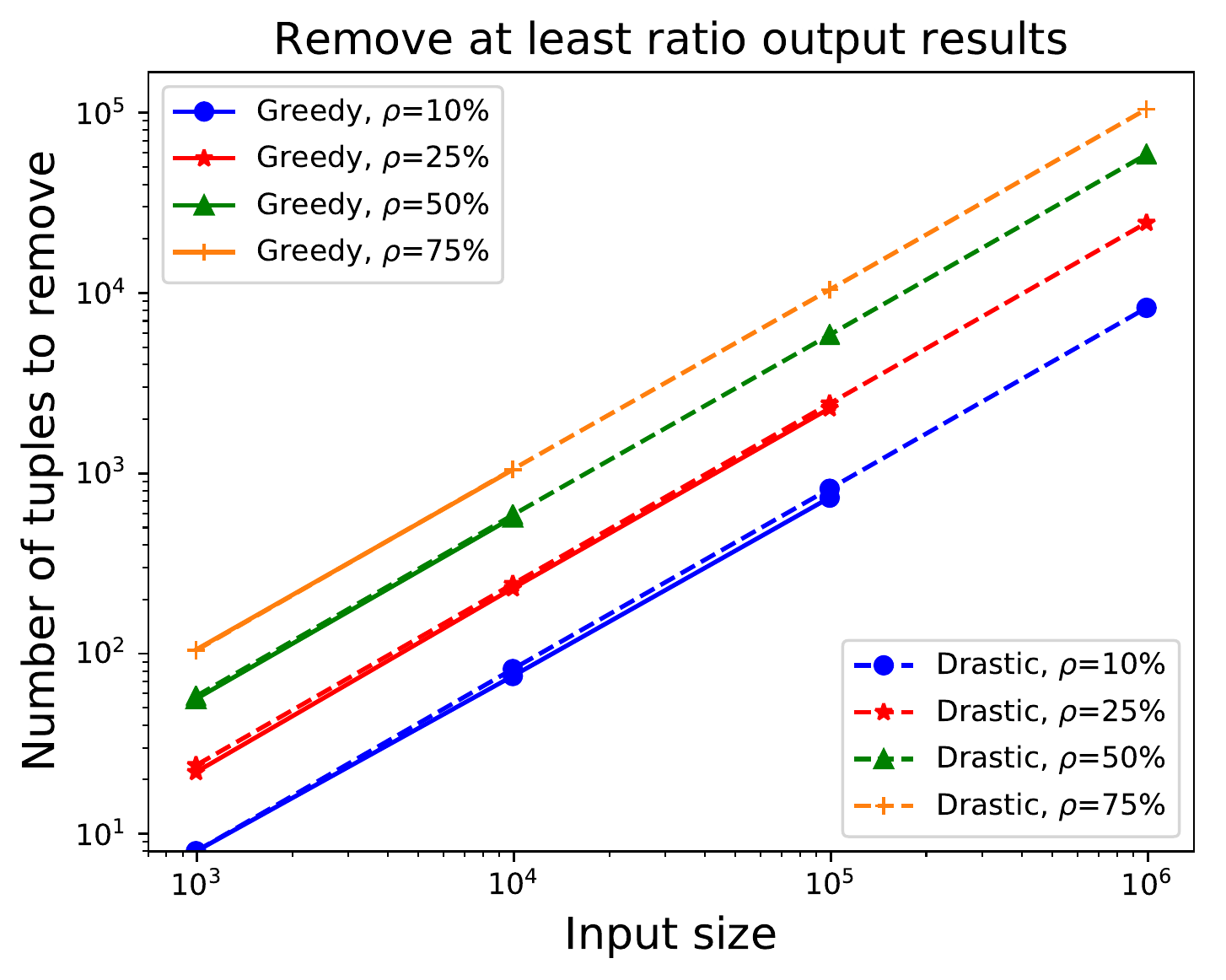}
	\caption{\revm{$\alpha = 0.25$ (hard)}}
	\label{fig:0.25skewQuality}
	\endminipage\hfill
	\minipage{0.24\textwidth}
	\includegraphics[width=\linewidth]{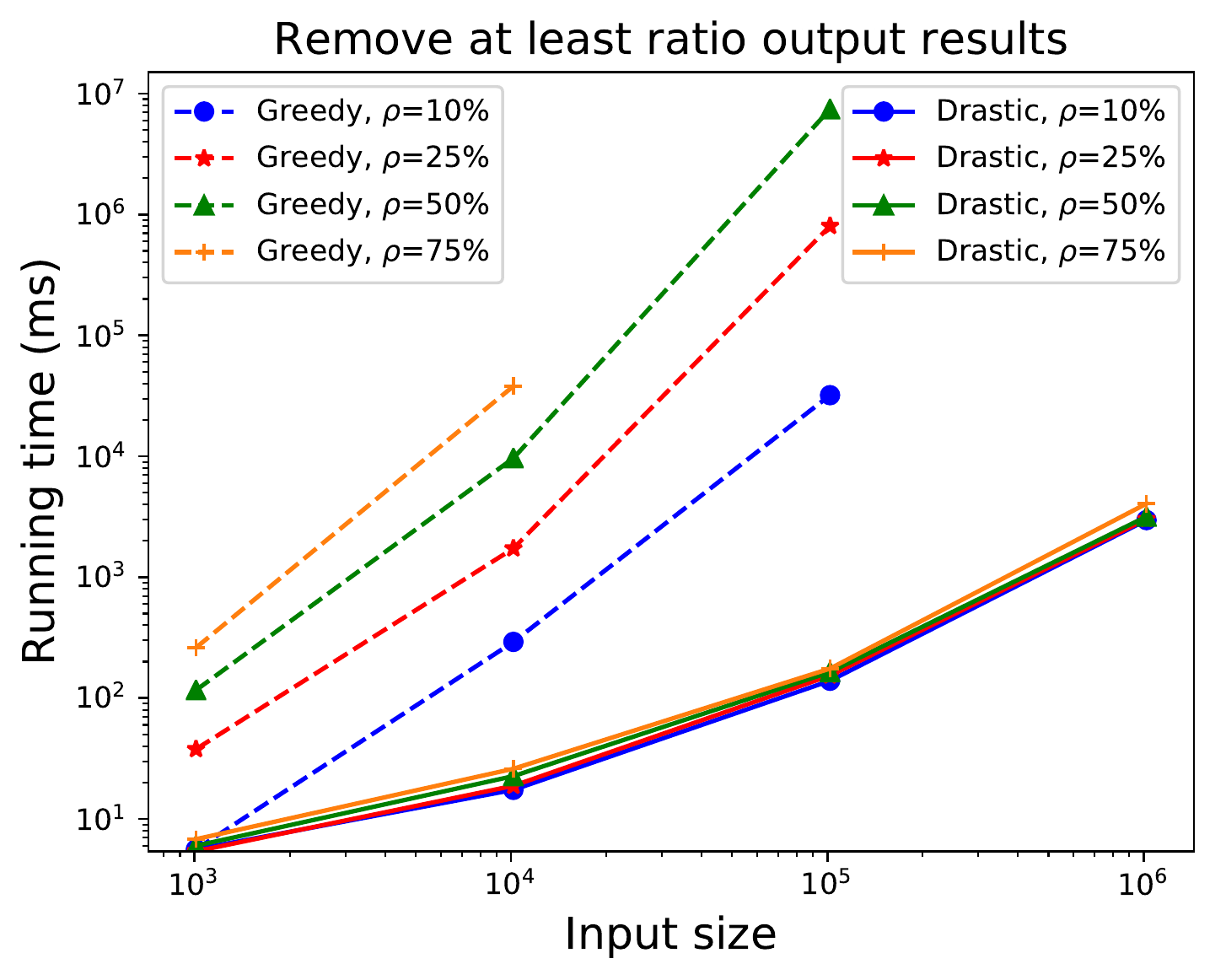}
	\caption{\revm{$\alpha = 0.5$ (hard)}}
	\label{fig:0.5skewTime}
	\endminipage\hfill
	\minipage{0.24\textwidth}
	\includegraphics[width=\linewidth]{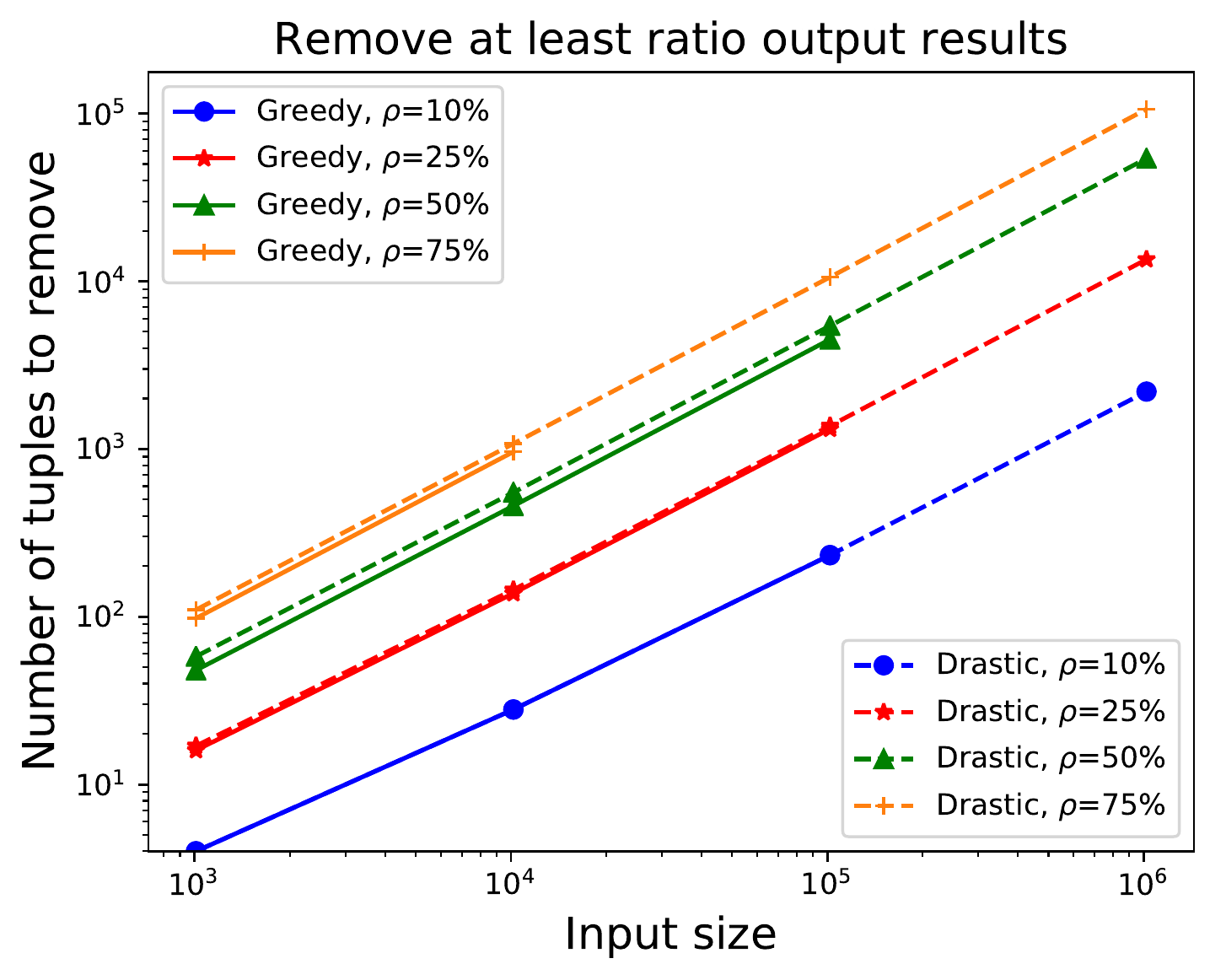}
	\caption{\revm{$\alpha = 0.5$ (hard)}}
	\label{fig:0.5skewQuality}
	\endminipage
\end{figure*}

}

\subsection{Optimizations}
\label{sec:expt-opt}
Next, we evaluate our optimizations on synthetic datasets. We use the following two queries: a singleton query $Q_5$ (attributes in $R_1$ are universal) and a disconnected query $Q_6$ (that can be  decomposed into three easy queries).
\begin{itemize}
	\item $Q_7(A,B,C,D,E,F,G):- R_1(A,B,C), R_2(A,B,C,D,E),R_3(A,B,C,D,G), R_4(A,B,C,F)$
	\item $Q_8(A_1,\cdots,C_3):- R_{11}(A_1), R_{12}(A_1, B_1), R_{21}(A_2), R_{22}(A_2,B_2), R_{31}(A_3), R_{32}(A_3, B_3)$
	%%\item $Q_7(A,B,C,E,G):- R_1(A,B), R2(B,C), R_4(E,F), R_5(F,G)$. 
\end{itemize}
\revm{
We generate relatively small synthetic datasets, as the non-optimized algorithm would take prohibitively long time on larger ones. For $Q_7$, each relation has $500$ input tuples and each tuple is randomly generated with a combination of integers between 1 and 100; for $Q_8$, $R_{11}, R_{21}, R_{31}$ each has $25$ input tuples and $R_{12}, R_{22}, R_{32}$ each has $50$. Each input tuple is randomly generated with a combination of integers between 1 and 100. }
%Even on this small dataset, we can demonstrate significant improvement of our optimization over the naive techniques and on larger dataset, it is impossible to finish the computation without optimizations. 
%
For $\ourprob(Q_7,D,k)$, we compare three different strategies:  (1) removing universal attributes $A,B,C$ one by one, (2) removing $A,B,C$ together, and (3) invoking procedure \singleton$(Q_7,D,k)$ based on sorting; the results are shown in Figure~\ref{fig:singleton}. For $\ourprob(Q_8,D,k)$, we  compare three different strategies: (1) decompose into 3 partitions at once, (2) decompose into $2$ partitions each time, and (3) improved dynamic programming; the results are shown in Figure~\ref{fig:decomposition}.  Note that all these strategies will compute all subproblems $\ourprob(Q_{i},D,k)$ for each subquery $Q_i(A_i, B_i, C_i):- R_{i1}(A_{i}, B_{i}), R_{i2}(A_{i}, B_{i})$, but only differ how the solutions for each subquery are used to construct the optimal solution for the $\ourprob(Q_7,D,k)$ problem. Figures~\ref{fig:singleton} and \ref{fig:decomposition} show that optimizations improve the  running time significantly.

\begin{figure}
	\centering
    \minipage{0.33\textwidth} 
	\includegraphics[width=\linewidth]{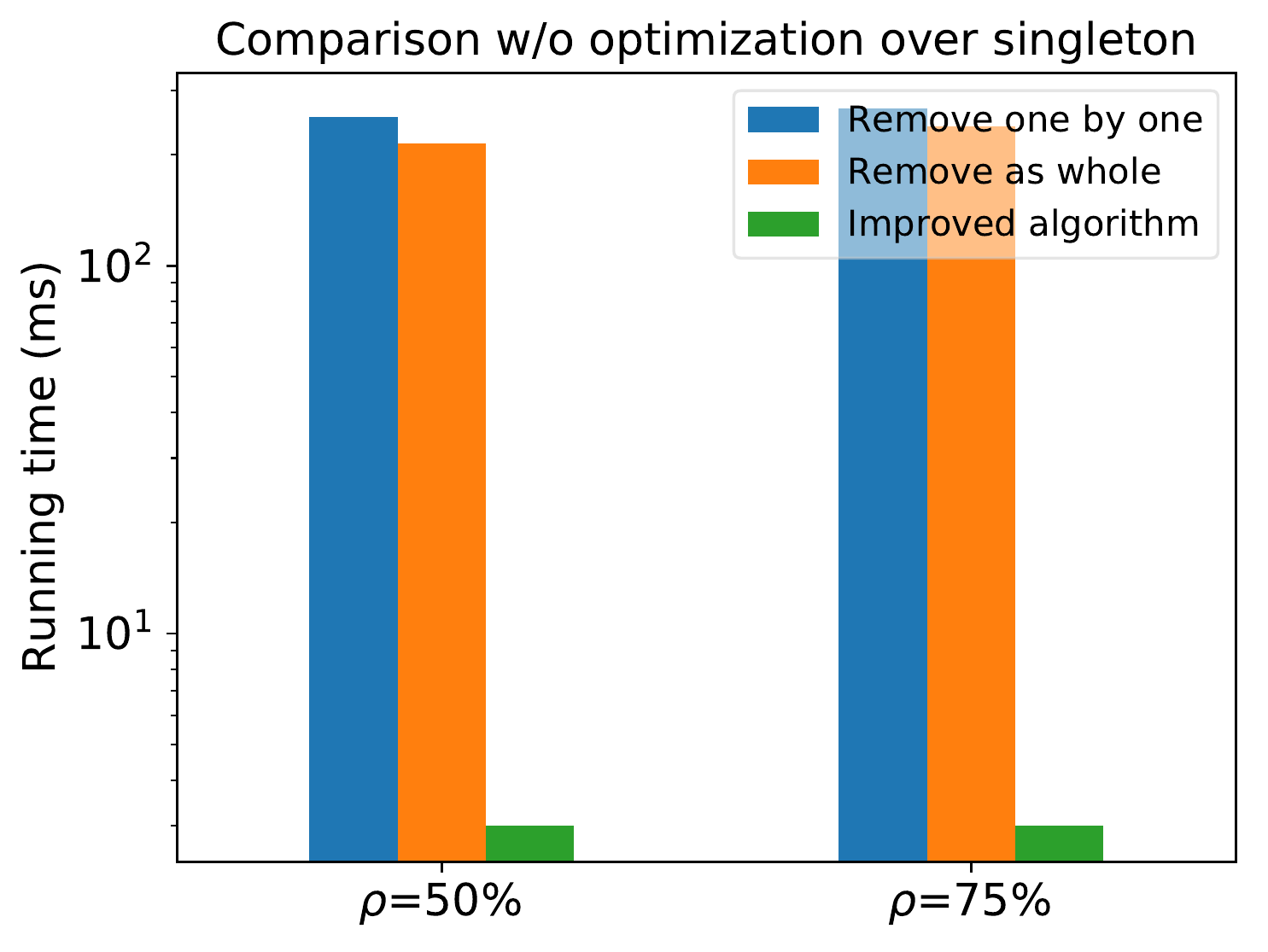}
	\caption{\revm{$Q_7$.}}
	\label{fig:singleton}
	\endminipage
	\minipage{0.33\textwidth}%
	\includegraphics[width=\linewidth]{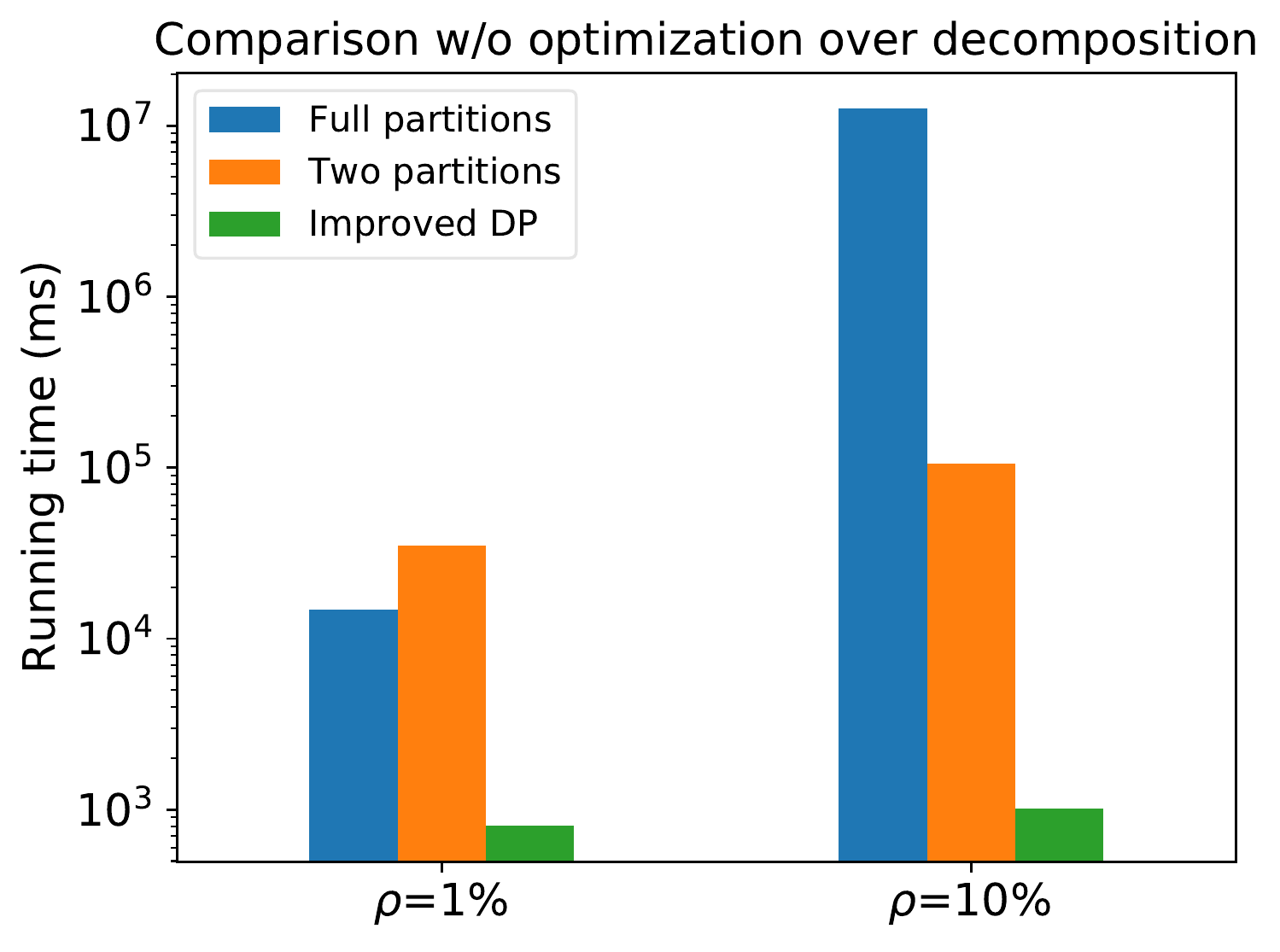}
	\caption{\revm{$Q_8$.}}
	\label{fig:decomposition}
	\endminipage
\end{figure}

%% file: sections/conclusions.tex
\section{Future Work}
\label{sec:conclusions}
\cut{
In this paper, we studied the aggregated deletion propagation (\ourprob) problem for the class of conjunctive query without self-joins, gave a dichotomy to decide whether \ourprob\ for a query is poly-time solvable for all $k$ and instance $D$ both algorithmically and structurally, and also gave approximation results.

Several open questions remain. First, it would be interesting to study the complexity of \ourprob for larger classes of queries involving self-joins, other classes of aggregates like {\em sum}, and also instances with arbitrary weights on the input and output tuples. 
}

\revm{
Several open questions remain. First, it would be interesting to study the \ourprob\ problem beyond CQs. In particular, many natural queries involve self-joins and/or aggregates like {\em sum}, for which the observations of this paper do not apply. It is also natural to consider scenarios where all input tuples are not equivalent in terms of the cost of removing them. 
As a first step, one might want to consider a scenario where only a subset of input tuples can be removed, and the remaining input tuples cannot be deleted.
\cut{It would be interesting to explore whether the results in this paper can be extended to this case by suitably modifying the definitions of concepts such as universal attributes, vacuum relations, etc. A more general question is to consider an arbitrary cost function on the input tuples, the goal being to remove at least $k$ output tuples by deleting a subset of input tuples of minimum cost. 
}
Investigating the approximability of the \ourprob\ problem is another interesting research direction. Although we showed some preliminary results in this context, obtaining an exact characterization of the approximability of this problem for individual queries, even for the special case of the Resilience problem, remains open. 
%In fact, obtaining approximation guarantees for the special case of the resilience problem would be a step toward achieving this goal. 
A related question is that of the parameterized complexity of \ourprob\ with respect to $k$ for full CQs. While we showed that \ourprob\ admits a poly-time algorithm for fixed $k$, obtaining an FPT algorithm for the problem remains open.} % , similar to the NP-hardness of this problem o  gneral CQs, iAnother direction is to understand the approximability of this problem: %even for the restricted class of CQs without self-join. 
 %whether a poly-time algorithm exists giving an absolute constant independent of the schema as the approximation factor is an interesting open problem. 

%%We showed that a $p$-approximation exists for full CQs where $p$ is the number of tables in the query, but our study of this problem suggests that this bound is not tight.

%% file: sections/appendix.tex
\appendix

\allowdisplaybreaks

\section{Endogenous relations}
\label{appendix:endogenous}

To compare our definitions with those from \cite{FreireGIM15}, we need to introduce the following terminologies. In a CQ $Q$, relation $R_j \in \rel(Q)$ is {\em exogenous} if there exists another relation $R_i  \neq R_j \in \rel(Q)$ such that $\attr(R_i) \subset \attr(R_j)$, and {\em endogenous} otherwise. It should be noted that if there are more than one relation defining on the same attributes, i.e., $\attr(R_i) = \attr(R_j)$, then we just consider arbitrary one of them as {\em endogenous} and the remaining as {\em exogenous}. In $Q():- R_1(A), R_2(A,B), R_3(B,C), R_4(B,C), R_5(B,C)$, there are two endogenous relations $R_1$ and any one of $R_3, R_4, R_5$.  We generalize their observation on endogenous relations in~\cite{FreireGIM15} to the \ourprob\ problem, as stated in Lemma~\ref{lem:endogenous}, which will be used in this paper. 

\begin{lemma}
	\label{lem:endogenous}
	For any CQ $Q$, if $\ourprob(Q,D,k)$ problem is poly-time solvable, there exists a solution which only removes input tuples from endogenous relations. 
\end{lemma}

\begin{proof}
	Consider an arbitrary solution $\S$ for $\ourprob(Q, D, k)$. By contradiction, assume tuple $t \in R_j$ is removed by $\S$ where $R_j$ is an exogenous relation. Let $R_i \in \rel(Q)$ be the endogenous relation such that $\attr(R_i) \subset \attr(R_j)$, and $t'$ be the tuple such that $\pi_{\attr(R_i)} t = t'$. If $t' \in \S$, we observe that $\S -\{t\}$ also removes at least $k$ results from $Q(D)$, contradicting the optimality of $S$. Otherwise, $t' \notin \S$. Then we claim that $\S - \{t\} + \{t'\}$ is also an optimal solution for $\ourprob(Q, D, k)$. Applying this argument to each tuple removed from exogenous relation, we will obtain an optimal solution which only removes tuples from endogenous solution. Thus adding the restriction on $Q$ doesn't change the minimum number of tuples to be removed for $\ourprob(Q, D, k)$. 
\end{proof}

\section{Proof of Lemma~\ref{LEM:BIPARTITE-GRAPH}}
\label{APPENDIX:HARD-CORE}

We show the NP-hardness of each problem in Lemma~\ref{LEM:BIPARTITE-GRAPH} separately.

\paragraph{Hardness Proof of Problem (1).} With an equivalent definition, problem (1) is exactly the {\em Partial Vertex Cover for Bipartite Graphs} (\texttt{PVCB}) problem, which is known to be NP-hard~\cite{CaskurluMPS17}. 

\begin{definition}
	\label{def:pvc}
	The input to the  problem is an undirected bipartite graph $G(A, B, E)$ where $E$ is the set of edges between two sets of vertices $A$
	and $B$, and an integer $k$. The goal is to find a subset $S \subseteq A \cup B$ of minimum size such that at least $k$ edges from $E$ have at least one endpoint in $S$.
\end{definition}

\paragraph{Hardness Proof of Problem (2).} It is easy to relate problem (2) to the {\em k-Minimum Coverage} (\kmc) problem, which is known to be NP-hard~\cite{vinterboa2002note}.

\begin{definition}
	\label{def:kmc}
	Given a universe $\mathcal{U}$, a family $\S$ of subsets of $\U$ and an integer $k$, find $k$ subsets from $\mathcal{S}$ such that the size of their union is minimized. 
\end{definition}

We give a reduction from the \kmc\ problem that takes as input $(\U, \S)$ and $k$, denoted as \kmc$(\U, \S, k)$. Moreover, it can be easily checked that this reduction preserves the approximation, i.e., if there is an $\alpha$-approximation algorithm for the $\ourprob(\swingquery, D,k)$ problem, then there must exist an $\alpha$-approximation algorithm for the \kmc\ problem. 

Given an instance of the \kmc\ problem, we construct a bipartite graph $G = (A,B,E)$ as follows. For each element $u \in \U$, we include a vertex $b_u \in B$. For each subset $S \in \mathcal{S}$, we include a vertex $a_S \in A$. If $u \in S$, we add an edge $(a_S, b_u) \in E$. Next we show that the problem \kmc$(\U, \S, k)$ has a solution of size $\le c$ if and only if the problem (2) has a solution of size $\le c$.

{\bf The``only-if'' direction.} Suppose we are given a solution $\S' \subseteq \S$ for problem \kmc$(\U, \S, k)$ of size $\le c$. We then construct a solution for problem (2) as follows. If $u \in \bigcup_{S \in \S'} S$, then we remove $b_u$ from $B$. This solution removes at most $c$ vertices from $B$ since $|\bigcup_{S \in \S'} S| \le c$. Moreover, every vertex $a_S \in A$ is removed as long as $S \in S'$. The total number of vertices removed from $A$ is at least $k$, thus this is exactly a solution for problem (2) of size $\le k$.

{\bf The ``if'' direction.} Suppose we are given a solution for problem (2) of size $\le c$.  We choose $k$ arbitrary vertices from $A$ which is removed because of the removal of vertices in $B$, denoted as $A'$. We then construct a solution for \kmc$(\U, \S, k)$ as $\{S: a_S \in A'\}$. It can be easily argued that $|\bigcup_{S: a_S \in A'} S| \le c$. Suppose not, there must exist at least one vertex $b_u$ for $u \in \bigcup_{S: a_S \in A'}$ not removed. In this way, at least one vertex in $A'$ cannot be removed, coming to a contradiction.

\paragraph{Hardness Proof of Problem (3).} However, to our knowledge, there is no existing result directly implying the hardness of problem (3). We first elaborate it as the {\em Sided-Constrained Vertex Cover in Bipartite Graphs} (\svcb). 

\begin{definition}
	\label{def:svcb}
	The input to the  problem is an undirected bipartite graph $G(A, B, E)$ where $E$ is the set of edges between two sets of vertices $A$
	and $B$, and an integer $c$. The goal is to find a subset $S \subseteq A \cup B$ of minimum size such that each edge from $E$ have at least one endpoint in $S$ and at least $c$ vertices in $A$ are included in $S$.
\end{definition}

A related problem that has been studied is the {\em constrained minimum vertex cover}, which is known to be NP-complete~\cite{chen2003constrained}, but with a different settings from \svcb. It asks to find a minimum vertex cover $S \subseteq A \cup B$ such that $|A \cap S| \le k_1$ and $|B \cap S| \le k_2$ for some input integer $k_1, k_2$. As a side product, we also first show that \svcb\ problem is NP-hard in Lemma~\ref{LEM:SVCB}, whose proof is given in Appendix~\ref{appendix:svcb} of independent interest. 

\begin{lemma}
	\label{LEM:SVCB}
	The \svcb\ problem is NP-hard. 
\end{lemma}

We give a reduction from the decision version of \svcb\ problem that takes input $G = (A, B, E)$ and an integer $c \le |A|$, denoted as \svcb$(G,c)$. For simplicity, assume each vertex in $G$ is incident to at least one edge in $E$. Given an instance of the \svcb\ problem, we have the same bipartite graph $G$ for problem (3). Next we show that \svcb$(G,c)$ has a solution of size $\le c'$ if an only if the problem (3) with parameter $k = |A| -c$ has a solution of size $\le c'- c$.

{\bf The``only-if'' direction.} Suppose we are given a vertex cover $\mathcal{C}$ for the problem \svcb$(G,c)$ of size $\le c'$. Let $B_1 = \mathcal{C} \cap B$ and $A_1 = \mathcal{C} \cap A$, where $|A_1| \ge c$. As a complement of $\mathcal{C}$, $(A- A_1, B -B_1)$ form an independent set of $G$. This implies that for each vertex $a \in A - A_1$, if $(a,b) \in E$, then $b \in B_1$. 

We construct a solution $\S$ for problem (3) as follows. We choose arbitrary $|A_1|-c$ vertices from $A_1$, denoted as $A_2$. Let $\S = A_2 \cup B_1$. The size of $\S$ can be bounded as $|A_1| + |B_1| - c = |\mathcal{C}| - c' \le c' -c$. Moreover, it can be easily checked that $\mathcal{S}$ is a valid solution for problem (3). Each vertex $a \in A - A_1$ will be removed, since all of its neightbors are in $B_1$, which have been removed already. Additional $|A_1| - c$ vertices are also removed from $A_2$. Thus, the total number of vertices removed from $A$ is $|A - A_1| + |A_1| - c = |A| - c$.

{\bf The ``if'' direction.} Suppose we are given a solution $\S$ for the problem (3) with parameter $k = |A| -c$, of size $\le c' -c$. Let $B_1 = \S \cap B$ and $A_2 = \{a \in A: (a,b) \notin E \ \forall  b \in B -B_1\}$. We mention two important properties on $\mathcal{S}$ first. (\romannumeral 1) If $|A_2| \ge |A| -c$, then $\mathcal{S} = B_1$ with size $\le c' - c$. (\romannumeral 2) If $|A_2| < |A| - c$, there must be $|\mathcal{S} \cap (A - A_2)| \ge |A| - |A_2| - c$.  In this case, $c'-c \ge |\mathcal{S}| \ge |\mathcal{S} \cap (A - A_2)| + |B_1| \ge |A| - |A_2| - c + |B_1|$, thus $c' \ge |A| - |A_2| + |B_1|$. 

We construct a solution $\mathcal{C}$ for the problem \svcb$(G,c)$ as follows. If $|A_2| \ge |A| -c$, choose arbitrary $|A_2| - |A| + c$ vertices from $A_2$ as $A_3$ and set $\mathcal{C} = (A -A_2 + A_3, B_1)$. Otherwise, set $\mathcal{C} = (A - A_2, B_1)$.

Observe that $\mathcal{C}$ is a valid vertex cover since $(A_2, B -B_1)$ is an independent set of $G$. It remains to show that $|\mathcal{C}| \le c'$ and $|\mathcal{C} \cap A| \ge c$. Note that if $|A_2| \ge |A| -c$, we have $|\mathcal{C} \cap A| = |A|-|A_2| + |A_3| = c$ and $|\mathcal{C}| = |A| - |A_2| + |A_3| + |B_1| \le c + c' -c = c'$, implied by (\romannumeral 1). Otherwise, $|\mathcal{C} \cap A| = |A| - |A_2| \ge c$. Moreover, $|\mathcal{C}| = |A| - |A_2| + |B_1| \le c'$, implied by (\romannumeral 2).

%\begin{lemma}
%\label{lem:seesaw-svcb-approximation}
%If there is an $\alpha$-approximation algorithm for the $\ourprob(\swingquery, D, k)$ problem, then there must exist another $\alpha$-approximation algorithm for the \svcb\ problem.
%\end{lemma}
%\begin{proof}
%Let $\mathcal{A}$ be an $\alpha$-approximation algorithm for the  $\ourprob(\swingquery, D, k)$ problem and  $\mathcal{A}(\swingquery, D, k)$ be its solution. For any instance of \svcb$(G,c)$, we have the same bipartite graph $G$ for problem (3). We run algorithm $\mathcal{A}$ for the problem $\ourprob(\seesawquery, D, |A| - c)$ and return the solution $|\mathcal{A}(\seesawquery, D, k)| + c$ for \svcb$(G,c)$.
%
%Next, we analyze the approximation ratio of this result. Implied by our reduction, the optimal solution for $\svcb(G,c)$ is $|\ourprob(\seesawquery, D, |A| - c)| + c$. Observe that 
%\[|\mathcal{A}(\seesawquery, D, k)| + c \le \alpha \cdot (|\ourprob(\seesawquery, D, |A| -c)| + c)\]
%so we can obtain an $\alpha$-approximation for the $\svcb(G,c)$ problem, completing the whole proof.
%\end{proof}

\section{Proof of Lemma~\ref{LEM:SVCB}}
\label{appendix:svcb}

In this part, we prove the NP-hardness of \svcb\ problem by showing that the NP-complete problem of \clique\ in a regular graph~\cite{mathieson2008parameterized} is polynomial time reducible to it, denoted as \regular.

%\begin{lemma}
%\label{lem:regular-clique}
%If there is an algorithm that can approximate \regular$(G')$ within $m^{1/3-\epsilon}$ factor for any $\epsilon > 0$, where $m$ is the number of vertices in $G'$, then there must exist an algorithm that can approximate the \clique$(G)$\ problem within factor $n^{1-\epsilon}$ for any $\epsilon > 0$, where $n$ is the number of vertices in $G$. 
%\end{lemma}
%
%\begin{proof}
%For each graph $G = (V,E)$ with $|V| = n$, we construct another graph $G'$ by using the reduction from the \clique\ to the \regular\ problem in~\cite{mathieson2008parameterized}. Let $\Delta$ be the maximum degree of vertices in $G$. We can construct an regular graph $G'= (V', E')$, where each vertex in $V'$ has degree $2 \cdot \lceil \frac{\Delta}{2}\rceil$. Note that there are $|V'| \le n \cdot \Delta^2 < n^3$. Moreover, there is a clique of size $q$ in $G$ if and only if there is a clique of size $q$ in $G'$. 
%
%Let $\mathcal{A}$ be such an algorithm that can  approximates the \regular$(G)$\ problem within $n^{1/3-\epsilon}$ factor, where $n$ is the number of vertices in $G'$. We just run then $\mathcal{A}$ on $G'$, and return this solution to \clique$(G)$. Note that this gives $n^{3(\frac{1}{3} - \eps)} = n^{1 - \eps'}$ approximation for the \clique\ problem. 
%\end{proof}
%
%
%We also point out that this reduction in~\cite{mathieson2008parameterized} preserves the approximation ratio.

Recall that the input is an undirected bipartite graph $G(A, B, E)$ where $E$ is the set of edges between two sets of vertices $A$
and $B$, and an integer $c \le |A|$. The goal is to find a subset $S \subseteq A \cup B$ of minimum size such that each edge from $E$ have at least one endpoint in $S$ and at least $c$ vertices in $A$ are included by $S$.

\medskip \noindent {\bf Instance Construction.} Let $G' = (V',E')$ be a $d$-regular graph, where $|V'| = n$ and $|E'| = m$. Let $5\le q \le \frac{n-1}{2}$ be an integer. The \clique\ problem asks whether there exists a set of $q$ vertices in $V'$ such that each pair of vertices chosen are connected by an edge in $E'$. We construct an instance $G= (A \cup B, E)$ with $k_1, k_2$ as follows. Each vertex $u \in V'$ defines a vertex-block, in forms of a biclique $A_u \times B_u$, where $A_u \subseteq A$ contains $\lambda_1$ vertices and $B_u \subseteq B$ contains $\lambda_2$ distinct vertices. Moreover, $\lambda_1 - \lambda_2 \ge d$. Each edge $e \in E'$ defines a vertex $b_e \in B$. If vertex $u$ is the endpoint of edge $e$ in $G'$, we just add one edge from $b_e$ to one vertex in $A_u$ with degree $\lambda_2$. This is always possible since  $\lambda_1 > d$.  In our constructed graph, there is $|A| = \lambda_1 n$, $|B| = \lambda_2 n + m$ and $|E| = \lambda_1 \lambda_2 n+ 2m$. Set $c = \lambda_1 q$. 

Any $\lambda_1, \lambda_2, d$ satisfying the following the constraints work for this proof, say, $\lambda_1 = 2q(q+1), \lambda_2 = 2q^2, d = 2q$. 
\begin{enumerate}
	\item $\lambda_1 > \max\{2q-1, \frac{1}{2}q(q-1)\}$;
	\item $\lambda_1 - \lambda_2 \ge d \ge 2q$;
	\item $(\lambda_1 - \lambda_2)(n-q) + \frac{1}{2} q (q-1) \ge m$;
	\item $\lambda_1 \ge (q-1) \cdot d$;
	\item $\lambda_2 + \frac{1}{2}(q-1) > \frac{1}{2} \lambda_1$;
	\item $\lambda_2 > \frac{1}{2}\lambda_1 + \frac{1}{2}(q-1)(d-q)$;
	\item $2m = nd$;
	\item $d \le n-1$.
\end{enumerate}
But for generality, we still use $\lambda_1, \lambda_2, d$ for analysis. We will show that the original graph $G'=(V',E')$ has a clique of size $q$ if and only if the bipartite graph $G= (A \cup B, E)$ has a vertex cover $J$ such that $|J \cap A| \ge \lambda_1 q$ and $|J| \le \lambda_1 q + \lambda_2 (n-q) + m -\frac{1}{2}q(q-1)$.

\medskip
\textbf{``Yes'' instance:} If there exists a clique of size $q$ in $G'$, we construct the vertex cover as follows. If a vertex $u \in V'$ is in the clique, choose $A_u$; otherwise, choose $B_u$. For an edge $e = (u,u') \in E'$, if at least one of $u,u'$ is not in the clique, choose $e_u$. It can be easily checked that each edge is covered, so this is a valid vertex cover. Moreover, $|J \cap A| = \lambda_1 q$ and $|J| = \lambda_1 q + \lambda_2 (n-q) + m -\frac{1}{2}q(q-1)$.

\medskip
\textbf{``No'' instance:}If there exists no clique of size $q$ in $G'$, every vertex cover $J$ of $G$ with $|J \cap A| \ge \lambda q$, must have its size strictly larger than $\lambda_1 q + \lambda_2 (n-q) + m -\frac{1}{2}q(q-1)$. Let $J^*$ be the minimum one among the class of vertex covers with $|J \cap A| \ge \lambda_1 q$.

The first observation is that $|J^* \cap A| = \lambda_1 q$. By contradiction, assume $|J^* \cap A| > \lambda_1 q$. If we can find some $u \in V'$ with $A_u \subsetneq J^*$, then $A_u \subseteq J^*$; for each vertex $a \in A_u \cap J^*$, we remove $a$ from $J^*$ and add $b_e$ to $J^*$ if there is a edge block $b_e$ connected to $v$. Otherwise, for each $u \in V'$ with $A_u \cap J^* \neq \emptyset$, there is $A_u \cap J^* = A_u$. In this case, $|J^* \cap A| = \lambda_1 q''$ with $q'' > q$. For an arbitrary $u \in V'$ with $A_u = J^*$, we remove $A_u$ from $J^*$ and add $B_u \cup (\bigcup_{e \in E': u \in e} b_e)$ to $J^*$. Note that $|B_u \cup (\bigcup_{e \in E': u \in e} b_e)| = |B_u| + |\bigcup_{e \in E': u \in e} b_e| = \lambda_2 + d \le \lambda_1$. In this way, we can get a better (at least not worse) vertex cover while maintaining the constraint that $|J^* \cap A| \ge \lambda_1 q$.

Based on $J^*$, we divide vertices in $V'$ into three subsets:
\begin{itemize}
	\item[] $A_1= \{u \in V': A_u - J^* = \emptyset\}$;
	\item[] $A_2 = \{u \in V': A_u - J^* \neq \emptyset, A_u \cap J^* \neq \emptyset\}$;
	\item[] $A_3 = \{u \in V': A_u \cap J^* = \emptyset\}$;
\end{itemize}
Note that $J^*$ has to pick the $B_u$ for every $u \in A_2 \cup A_3$. We further consider two cases: (1)  $|A_1| = q$; (2) $|A_1| \le q-1$. Both cases are built on the following common observations. Consider an edge block $b_e$ with $e = (u,u')$. Let $a_{eu} \in A_u$ and $a_{eu'} \in A_{u'}$ be the two vertices incident to $b_e$ in $G$. Note that $b_e \notin J^*$ if and only if $a_{eu} \in J^*$ and $a_{eu'} \in J^*$. 

Case 1: $|A_1| = q$. In this case, $|A_2| + |A_3| = n-q$, and $A_2 = \emptyset$ since $|J^* \cap A| = \lambda_1 q$. For any edge $e = (u,u') \in E$, $b_e \notin J^*$ if and only if $u \in  A_1$ and $u' \in A_1$. Since there is no $q$-clique in $G'$, $J^*$ has size at least $\lambda_1 q + \lambda_2 (n-q) + m -\frac{1}{2}q(q-1)+1$.

Case 2: $|A_1| \le q-1$. Consider each edge $e = (u, u') \in G'$.  Observe that if one of $u,u'$ is in $A_3$, there must be $b_e \in J^*$. We further distinguish three more cases for $e$ when $b_e \notin J^*$. (\romannumeral 1) both $u,u' \in A_1$, $b_e \notin J^*$. Let $\alpha$ be the number of edges falling into this case. (\romannumeral 2) $u,u' \in A_2$, then $J^*$ has to choose both $a_{eu}, a_{eu'}$ for only exempting $e_u$.  (\romannumeral 3) one of $u,u'$ is in $A_1$ and the other in $A_2$, say $u \in A_1, u' \in A_2$, then $J^*$ has to choose $a_{eu'}$ for exempting $b_e$; and the number of such edges is at most $|A_1| \cdot d - 2 \alpha$. Note that $J^*$ will exempt as many as edge blocks as possible. With the additional budget of $\lambda_1(q-|A_1|)$ vertices in $A_2$, it will firstly exempt as many edge blocks in (\romannumeral 3) as possible; and then exempt edge blocks in (\romannumeral 2). Under the parameter constraint (1), $\lambda_1 (q - |A_1|) \ge |A_1| \cdot d \ge |A_1| \cdot d - 2 \alpha$ for any $|A_1| \in \{1,2,\cdots,q-1\}$. So the number of exempted edge blocks is at most
\begin{align*}
f(|A_1|) & = \alpha + |A_1| \cdot d - 2 \alpha + \frac{1}{2}  \big(\lambda_1 (q -|A_1|) - (|A_1| \cdot d - 2 \alpha) \big) \\
& = \frac{1}{2} |A_1| \cdot d + \frac{1}{2} \lambda_1(q - |A_1|)
\end{align*}
In this case, $J^*$ has size at least $\lambda_1 q + \lambda_2 (n-|A_1|) + m - f(|A_1|)$.
To show why it is always strictly larger than $\lambda_1 q + \lambda_2 (n-q) + m -\frac{1}{2}q(q-1)$, it suffices to show that 
\begin{align*}
&\lambda_2 (q-|A_1|) - f(|A_1|) + \frac{1}{2}q(q-1) > 0 
\end{align*}
for any $|A_1| \in \{0,1,2,\cdots,q-1\}$.
Rearranging the inequality, this is equivalent to show
\begin{align*}
(\lambda_2 -\frac{1}{2}\lambda_1)(q-x) +\frac{1}{2}q(q-1)- \frac{1}{2} x d  > 0 
\end{align*}
holds for any $x \in [0,q-1]$. 
Note that this is a monotone function, so it holds for the whole interval $[0,q-1]$ as long as it holds for both endpoints. For $x=0$, it holds if $\lambda_2 + \frac{1}{2}(q-1) > \frac{1}{2} \lambda_1$. For $x = q-1$, it holds if $\lambda_2> \frac{1}{2}\lambda_1 +\frac{1}{2} (q-1)(d-q)$. Both constraints are implied by the parameter settings.

\section{Proof of Theorem~\ref{THM:DICHOTOMY-STRUCTURE}}
\label{appendix:dichotomy-structure}

We will prove Theorem~\ref{THM:DICHOTOMY-STRUCTURE} by drawing an equivalence to Theorem~\ref{thm:dichotomy}. For simplicity, when there is a triad-like or strand structure, or the head join of non-dominated relations is non-hierarchical in $Q$, $Q$ is referred to {\em contain hard structure}.

We first show that these two simplification steps in procedure \isptime\ preserve the hard structures (Lemma~\ref{LEM:COMMON-HARD-STRUCTURE} and Lemma~\ref{LEM:DECOMPOSE-HARD-STRUCTURE}).We then investigate three base cases. Note that when $Q$ is boolean, there is no triad structure since $\head(Q) \cap \attr(R_i) = \emptyset$ for any $R_i \in \rel(Q)$. The head join of $Q$ has no attributes, thus always being hierarchical. On boolean CQ, Theorem~\ref{THM:DICHOTOMY-STRUCTURE} degenerates to Theorem~\ref{thm:boolean-dichotomy} directly. So, it remains to consider the case when there is a vacuum relation in $Q$ (Lemma~\ref{lem:vacuum-structure}) or $\isptime(Q)$ goes to ``other'' in Figure~\ref{fig:isptime} (Lemma~\ref{lem:others-structure}).

\begin{proof}[Proof Lemma~\ref{LEM:COMMON-HARD-STRUCTURE}]
	For each relation $R_i \in \rel(Q)$, let $R'_i$ be the corresponding relation in $Q_{-A}$, with $\attr(R'_i) = \attr(R_i) - \{A\}$. We first mention two important observations for $Q, Q_{-A}$: (1) there is a one-to-one correspondence of non-dominated (resp. endogenous) relations in $Q$ and $Q_{-A}$, i.e., $R_i$ is non-dominated (resp. endogenous) if and only $R'_i$ is non-dominated (resp. endogenous); (2) for a full CQ, $Q$ is hierarchical if and only if $Q_{-A}$ is hierarchical.  Both can be easily checked by definition.
	. 
	
	{\bf The ``only-if'' direction.} Suppose $Q$ contains hard structure, and we prove each case separately. 
	
	If there is a triad-structure with a triple of endogenous relations $R_1, R_2, R_3 \in \rel(Q)$ such that for each pair of relations, say $R_1, R_2$, there exists a path between $R_1, R_2$ only using attributes in $\attr(Q) - \head(Q) - \attr(R_3)$. Obviously, $A$ doesn't appear on this path since $A \in \attr(R_3)$. Correspondingly, this path between $R'_1, R'_2$ only uses attributes in $\attr(Q_{-A}) - \head(Q_{-A}) - \attr(R_3) = \attr(Q) - \head(Q) - \attr(R'_3)$. Similar argument applies for $R'_1, R'_3$ and $R'_2, R'_3$. Thus, $R'_1,R'_2, R'_3$ form a triad in $Q_{-A}$.
	
	If there is a strand with a pair of non-dominated relations $R_1, R_2 \in \rel(Q)$ such that (1) $\head(Q) \cap \attr(R_1) \neq \head(Q) \cap \attr(R_2)$; (2) $\attr(R_i) \cap \attr(R_j) - \head(Q) \neq \emptyset$. It can be easily checked that $\head(Q) \cap \attr(R'_1) \neq \head(Q) \cap \attr(R'_2)$, and $\attr(R'_i) \cap \attr(R'_j) - \head(Q_{-A}) = \attr(R_i) \cap \attr(R_j) - \head(Q)  \neq \emptyset$.  Thus, $R'_1,R'_2$ form a strand in $Q_{-A}$.
	
	If the head join of non-dominated relations in $Q$ is non-hierarchical, removing a universal attribute $A$ from all relations doesn't change this property. Thus, the head join of non-dominated relations in $Q_{-A}$ is also non-hierarchical.

	{\bf The ``if'' direction.} Suppose $Q_{-A}$ contains hard structure. This direction can be argued similarly with the ``only-if'' direction.
\end{proof}

\begin{proof}[Proof of Lemma~\ref{LEM:DECOMPOSE-HARD-STRUCTURE}]
	We first mention two important observations for a disconnected query: (1) the set of non-dominated (resp. endogenous) relations in $Q$ is just the disjoint union of non-dominated (resp. endogenous) relations in each subquery; (2) a full join is hierarchical, if each of its connected subqueries is hierarchical. Both can be easily checked by definition.

	{\bf The ``only-if'' direction.} Suppose $Q$ contains hard structure, and we prove each case separately.  
	
	If there is a triad-structure with a triple of endogenous relations $R_1, R_2, R_3 \in \rel(Q)$, they must come from the same subquery, say $Q_i$, since there exists a path between any pair of them by definition. It can be easily checked that $R_1, R_2, R_3$ still form a triad in $Q_i$.
	
	Similarly, if there is a strand with a pair of endogenous relations $R_1, R_2\in \rel(Q)$, they must come from the same subquery, say $Q_i$, since they are connected. It can be easily checked that $R_1, R_2$ still form a strand in $Q_i$.
	
	If the head join of non-dominated relations in $Q$ is non-hierarchical, we can identify two attributes $A,B$ and three non-dominated relations $R_1, R_2,R_3$ such that $A \in \attr(R_1) \cap \attr(R_2) - \attr(R_3)$ and $B \in \attr(R_3) \cap \attr(R_2) - \attr(R_1)$. In this way, $R_1, R_2 ,R_3$ must come from the same subquery, say $Q_i$. It can be easily checked that this condition still holds in $Q_i$, thus being non-hierarchical.
	
	{\bf The ``if'' direction.} Suppose $Q_i$ contains hard structure. It can be easily checked that any hard structure in $Q_i$ also exists in $Q$. 
\end{proof}

\begin{lemma}
	\label{lem:vacuum-structure}
	For a CQ $Q$, if there is a vacuum relation,  then $Q$ doesn't contain any hard structure.
\end{lemma}

\begin{proof}
	Let $R_i$ be the vacuum relation.  By definition, every remaining relation $R_j \in \rel(Q) - \{R_i\}$ is dominated by $R_i$. Thus, there is neither triad-like nor strand structure in $Q$. The head join of non-dominated relations in $Q$ only includes $R_i$, thus always being hierarchical. Overall, $Q$ doesn't contain any hard structure.
\end{proof}

\begin{lemma}
	\label{lem:others-structure}
	For a CQ $Q$, if $\isptime(Q)$ goes to ``other'' in Figure~\ref{fig:isptime},  then $Q$ contains hard structure.
\end{lemma}

\begin{proof}
	We follow the same proof plan for Lemma~\ref{lem:others}, by distinguishing the class of CQs characterized by Lemma~\ref{lem:others-structure} into three cases, as illustrated in Figure~\ref{fig:others}. Recall that any query characterized by Lemma~\ref{lem:others} is connected, without any universal attribute and vacuum relation. we show that $Q$ falling into any one case contains hard structure. 
	
	\medskip \noindent {\bf Case 1: head join contains at least one vacuum relation.} Let $R_i \in \rel(Q)$ be the relation such that $\attr(R_i) \neq \emptyset$ and $\attr(R_i) \subseteq \attr(Q) - \head(Q)$. We start from any non-output attribute $B \in \attr(R_i)$ and do a binary search until we find an output attribute $A$. Let $R_1, R_2$ be the consecutive pair of relations on this path between $A,B$, such that $A \in \attr(R_1)$. Note that $\attr(R_2) \subseteq \attr(Q) - \head(Q)$; otherwise, $R_2$ would be the first relation containing output attributes in our search. Moreover, $R_2$ is non-dominated since there is no vacuum relation in $Q$, and $R_1$ is also non-dominated since $\attr(R_1) \cap \attr(R_2) \subseteq \attr(Q) - \head(Q)$. In this way, $R_1, R_2$ form a stand in $Q$. 
	
	\medskip \noindent {\bf Case 2: head join is disconnected (and no vacuum relation).} As there is no vacuum relation in head join, $\head(Q) \cap \attr(R_i) \neq \emptyset$ holds for each relation $R_i \in \rel(Q)$.  Moreover, we can always identify a pair of attributes $X,Z \in \head(Q)$ such that there is no path between $X,Z$ in the head join. As $Q$ is connected, every path between $X,Z$ in $Q$ uses at least one non-output attribute.
	
	Consider any path between $X, Z$ in $Q$, in which there is a pair of consecutive relations $R_1, R_2$ such that $\attr(R_1) \cap \attr(R_2) \subseteq \attr(Q) - \head(Q)$; otherwise, $X,Z$ are connected in the head join. Obviously, $\attr(R_1) \cap \attr(R_2) \cap \head(Q) = \emptyset$. We claim that both $R_1, R_2$ are non-dominated. Suppose not, say $R_1$ is dominated by $R_i$. By definition, $\attr(R_i) \subseteq \attr(R_1)$. Observe that $\attr(R_i) - \attr(R_2) \neq \emptyset$ since $\attr(R_i) - \attr(R_2) \supseteq \attr(R_i) \cap \head(Q) - \attr(R_2) \supseteq  \attr(R_i) \cap  \attr(R_1) \cap \head(Q) - \attr(R_2) \neq \emptyset$. Implied by Definition~\ref{def:dominated}, $\attr(R_1) \cap \attr(R_2) \subseteq \attr(R_i) \cap \head(Q)$, coming to a contradiction. Applying a similar argument, we can show that $R_2$ is non-dominated. Moreover, $\attr(R_1) \cap \head(Q) \neq \emptyset$, $\attr(R_2) \cap \head(Q) \neq \emptyset$, and $\attr(R_1) \cap \attr(R_2) \cap \head(Q) = \emptyset$, thus $\attr(R_1) \cap \head(Q) \neq \attr(R_2) \cap \head(Q)$. In this way, $R_1, R_2$ form a strand in $Q$.

	\medskip \noindent {\bf Case 3: head join is connected (and no vacuum relation).} As there is no vacuum relation in head join, $\head(Q) \cap \attr(R_i) \neq \emptyset$ holds for each relation $R_i \in \rel(Q)$.  Note that there exists no universal attribute in $Q$. 
	
	We claim that the head join of non-dominated relations in $Q$ is also connected. Suppose not, there is a pair of attributes $A,B \in \head(Q)$ which becomes disconnected in the head join of non-dominated relations. Consider any path $P$ between $A,B$ in the head join of $Q$, a sequence of relations where each pair of consecutive relations share at least one output attribute. We construct another path $P'$ as follows. For each relation $R_j \in P$, if it is dominated by $R_i \in \rel(Q)$, then we just replace $R_j$ by $R_i$ in $P'$.  Let $R_1 \in P, R_1' \in P'$ be the first relation in each path respectively. If $A \notin \attr(R'_1)$, then add an arbitrary non-dominated relation $R'_0 \in \rel(Q)$ such that $A \in \attr(R'_0)$ before $R'_1$. The similar operation is applied for $B$. We next argue that $P'$ is a valid path between $A,B$. It suffices to show that for each pair of consecutive relations in $P'$, they share at least one output attribute.  
	
	If $R'_0$ exists, we first show that $\attr(R'_0) \cap \attr(R'_1) \cap \head(Q) \neq \emptyset$. In this case, $R_1 \neq R'_1$; otherwise, $A \in \attr(R_1)$. Observe that $\attr(R'_1) \attr(R'_0) \neq \emptyset$, then $A \in \attr(R_1) \cap \attr(R'_0) \subseteq \attr(R_1') \cap \head(Q) \subseteq \attr(R'_1)$, coming to a contradiction. Otherwise, $\attr(R'_1)  \subseteq \attr(R'_0)$, thus \[\attr(R'_0) \cap \attr(R'_1) \cap \head(Q) = \attr(R'_1) \cap \head(Q) \neq \emptyset.\]  The symmetric case when such a relation for $B$ is added can be argued similarly.
	
	Consider any pair of consecutive relations $R_1, R_2 \in P$. Let $R'_1, R'_2$ be the corresponding relations in $P'$. By contradiction, assume $R'_1 \cap R'_2 \cap \head(Q) =\emptyset$. If $R_1 = R'_1, R_2 = R'_2$, it comes to a contradiction. Otherwise, we further distinguish two cases. If only one of $R_1 = R'_1$ and $R_2 = R'_2$ holds, say $R_1 \neq R_1, R_2 = R'_2$. Since $\attr(R'_2) - \attr(R'_1) \neq \emptyset$, then $\attr(R_1) \cap \attr(R_2) = \attr(R_1) \cap \attr(R'_2) \subseteq \attr(R'_1) \cap \head(Q)$, which implies $\attr(R'_1) \cap \attr(R'_2) \cap \head(Q)$, coming to a contradiction. Otherwise, $R_1 \neq R_1, R_2 \neq R'_2$, which can be argued similarly.

	Note that if a full CQ is connected without a universal attribute, it must be non-hierarchical, implied by the definition of hierarchical join. In this way, the head join of non-dominated relations in $Q$ is non-hierarchical.
	
	\medskip 
	When $\isptime(Q)$ goes to ``others'', some hard structure has been identified in $Q$ in each case,  thus completing the whole proof.
\end{proof}